\documentclass[oneside,12pt]{article}

\usepackage{jheppub}
\usepackage{amsfonts}
\usepackage{amsmath}
\usepackage{amssymb}
\usepackage{amsthm}
\usepackage[mathscr]{eucal}
\usepackage{bbold}
\usepackage{braket}
\usepackage{color}
\usepackage{dsfont}
\usepackage{framed}
\usepackage{mathtools}
\usepackage{physics}
\usepackage[normalem]{ulem}
\usepackage{tensor}
\usepackage{thmtools}
\usepackage{thm-restate}
\usepackage{ulem}
\usepackage{tikz}
\usetikzlibrary{arrows.meta,arrows} 
\usepackage{centernot}
\usepackage{enumitem} 
\usepackage[a]{esvect}  
\usepackage{slashed}

\usepackage{yfonts}
\usepackage{float}
\usetikzlibrary{math} 
\usetikzlibrary{calc}		

\usepackage{subcaption}
\usepackage{float}
\usepackage{afterpage}
\renewcommand{\thefigure}{\arabic{figure}}

\usetikzlibrary{arrows}
\tikzset{elabelcolor/.style={color=blue} 
    vertex/.style={circle,draw,minimum size=1.5em},
    edge/.style={->,> = latex'}
}  

\usepackage{cleveref}

\definecolor{THcolor}{rgb}{0.1,0.7,0.1}

\definecolor{VHcolor}{rgb}{0.7,0.3,0.9}

\definecolor{MRocolor}{rgb}{0.1,0,1}

\newcommand{\ignore}[1]{}

\definecolor{colormacro}{rgb}{0.6,0.4,0.4}  
\definecolor{colormacromath}{rgb}{0.2,0.4,0.8}

\newtheorem{conjecture}{Conjecture}
\newtheorem{cor}{Corollary}
\newtheorem{defi}{Definition}

\newtheorem{lemma}{Lemma}

\newtheorem{question}{Question}

\newtheorem{thm}{Theorem}

\newtheorem*{obs}{Observation}

\newcommand{\N}{{\sf N}}
\newcommand{\D}{\sf{D}}
\newcommand{\E}{{\sf{E}}}

\newcommand{\I}{\mathscr{I}}
\newcommand{\J}{\mathscr{J}}
\newcommand{\K}{\mathscr{K}}
\newcommand{\uI}{\underline{\mathscr{I}}}
\newcommand{\uJ}{\underline{\mathscr{J}}}
\newcommand{\uK}{\underline{\mathscr{K}}}
\newcommand{\Il}{\check{\mathscr{I}}}
\newcommand{\Jl}{\check{\mathscr{J}}}
\newcommand{\Kl}{\check{\mathscr{K}}}
\newcommand{\uIl}{\check{\underline{\mathscr{I}}}}

\newcommand{\uKl}{\check{\underline{\mathscr{K}}}}
\newcommand{\Ih}{\hat{\mathscr{I}}}

\newcommand{\Kh}{\hat{\mathscr{K}}}
\newcommand{\uIh}{\hat{\underline{\mathscr{I}}}}

\newcommand{\uKh}{\hat{\underline{\mathscr{K}}}}
\newcommand{\Ilc}{{\Il^\dagger}}

\newcommand{\Klc}{{\Kl^\dagger}}
\newcommand{\Ihc}{{\Ih^\dagger}}

\newcommand{\Khc}{{\Kh^\dagger}}

\newcommand{\lat}[1]{\mathfrak{L}_{\text{#1}}}

\newcommand{\ent}{{\sf S}}  
\newcommand{\mi}{{\sf I}}  
\newcommand{\er}{\mathscr{R}} 
\newcommand{\erssa}{\er_{_\text{SSA}}} 
\newcommand{\erq}{\er_{_\text{Q}}} 
\newcommand{\erh}{\er_{_\text{H}}} 
\newcommand{\erkc}{\er_{_\text{KC}}} 
\newcommand{\ghyp}{\mathbb{H}}  
\newcommand{\comp}[1]{{#1}^c}   
\newcommand{\pmi}{\mathbb{P}}   
\newcommand{\pmimap}{\pi} 
\newcommand{\face}{\mathscr{F}}   
\newcommand{\psys}[1]{[\![#1]\!]}

\newcommand{\mgs}{\mathcal{E}}  
\newcommand{\vece}[1]{\vv{#1}}   

\newcommand{\genstirlingII}[3]{%
  \genfrac{\{}{\}}{0pt}{#1}{#2}{#3}%
}
\newcommand{\stirlingII}[2]{\genstirlingII{}{#1}{#2}}

\newcommand{\covb}{%
  \mathrel{%
    \begin{tikzpicture}[baseline=-\the\dimexpr\fontdimen22\textfont2\relax ]
    \draw[thick] (-0.9ex,0ex) -- (0.0ex,0ex);
    \draw[thick] (0ex,0ex) -- (0.6ex,0.6ex);
    \draw[thick] (0ex,0ex) -- (0.6ex,-0.6ex);
    \end{tikzpicture}
  }
}

\newcommand{\cov}{%
  \mathrel{%
    \begin{tikzpicture}[baseline=-\the\dimexpr\fontdimen22\textfont2\relax ]
    \draw[thick] (0.9ex,0ex) -- (0.0ex,0ex);
    \draw[thick] (0ex,0ex) -- (-0.6ex,0.6ex);
    \draw[thick] (0ex,0ex) -- (-0.6ex,-0.6ex);
    \end{tikzpicture}
  }
}

\newcommand{\covbs}{%
  \mathrel{%
    \begin{tikzpicture}[baseline=-\the\dimexpr\fontdimen22\textfont2\relax ]
    \draw[thick] (-0.9ex,0ex) -- (0.0ex,0ex);
    \draw[thick] (0ex,0ex) -- (0.6ex,0.6ex);
    \draw[thick] (0ex,0ex) -- (0.6ex,-0.6ex);
    \filldraw[black] (0.5ex,0ex) circle (0.1ex);
    \end{tikzpicture}
  }
}

\newcommand{\covs}{%
  \mathrel{%
    \begin{tikzpicture}[baseline=-\the\dimexpr\fontdimen22\textfont2\relax ]
    \draw[thick] (0.9ex,0ex) -- (0.0ex,0ex);
    \draw[thick] (0ex,0ex) -- (-0.6ex,0.6ex);
    \draw[thick] (0ex,0ex) -- (-0.6ex,-0.6ex);
    \filldraw[black] (-0.5ex,0ex) circle (0.1ex);
    \end{tikzpicture}
  }
}

\newcommand{\veedot}{%
  \!\!\mathrel{%
    \begin{tikzpicture}[baseline=-\the\dimexpr\fontdimen22\textfont2\relax ]
    \node at (0ex,0ex) {$\vee$};
    \filldraw[black] (0ex,0.5ex) circle (0.1ex);
    \end{tikzpicture}
  \!\!}
}

\newcommand{\pt}{%
  \mathrel{%
    \begin{tikzpicture}[baseline=-\the\dimexpr\fontdimen22\textfont2\relax ]
    \filldraw (-0.8ex,0ex) circle(0.3ex);
    \filldraw (0.8ex,0ex) circle(0.3ex);
    \filldraw(0ex,0.8ex) circle(0.3ex);
    \filldraw(0ex,-0.8ex) circle(0.3ex);
    \draw (-0.8ex,0ex) -- (0.8ex,0ex);
    \draw (0ex,-0.8ex) -- (0ex,0.8ex);
    \end{tikzpicture}
  }
}

\newcommand{\bpg}{%
  \mathrel{%
    \begin{tikzpicture}[baseline=-\the\dimexpr\fontdimen22\textfont2\relax ]
    \filldraw (-0.8ex,0ex) circle(0.3ex);
    \filldraw(0.8ex,0ex) circle(0.3ex);
    \draw(-0.8ex,0ex) -- (0.8ex,0ex);
    \end{tikzpicture}
  }
}

\newcommand{\bp}[2]{%
  \mathrel{%
    \begin{tikzpicture}[baseline]
    \filldraw (-0.8ex,0ex) circle(0.3ex);
    \filldraw(0.8ex,0ex) circle(0.3ex);
    \draw(-0.8ex,0ex) -- (0.8ex,0ex);
    \node at (-0.9ex,1.3ex) {{\tiny \textit{#1}}};
    \node at (0.9ex,1.3ex) {{\tiny \textit{#2}}};
    (-.13ex, 0ex) (1.3ex + .13ex, 0ex); 
    \end{tikzpicture}
  }
}

\makeatletter\def\@fpheader{~}\makeatother
\newcommand{\kcjoin}{\vee_{_{_{\!\!\text{kc}}}}} 
\newcommand{\zds}{{\zeta\text{-down-set}}}
\newcommand{\zdss}{{\zeta\text{-down-sets}}}

\title{On the relation between the subadditivity cone and the quantum entropy cone}

\author[a]{Temple He,}
\emailAdd{templehe@caltech.edu}

\author[b]{Veronika E. Hubeny,}
\emailAdd{veronika@physics.ucdavis.edu}

\author[c]{Massimiliano Rota}
\emailAdd{m.rota@uva.nl}

\affiliation[a]{Walter Burke Institute for Theoretical Physics \\ California Institute of Technology, Pasadena, CA 91125 USA}
\affiliation[b]{Center for Quantum Mathematics and Physics (QMAP)\\ 
Department of Physics \& Astronomy, University of California, Davis, CA 95616 USA}
\affiliation[c]{Institute for Theoretical Physics, University of Amsterdam,
Science Park 904, Postbus 94485, 1090 GL Amsterdam, The Netherlands}

\abstract{Given a multipartite quantum system, what are the possible ways to impose mutual independence among some subsystems, and the presence of correlations among others, such that there exists a quantum state which satisfies these demands? This question and the related notion of a \textit{pattern of marginal independence} (PMI) were introduced in \cite{Hernandez-Cuenca:2019jpv}, and then argued in \cite{Hernandez-Cuenca:2022pst} to be central in the derivation of the holographic entropy cone. Here we continue the general information theoretic analysis  of the PMIs allowed by \textit{strong subadditivity} (SSA) initiated in \cite{Hernandez-Cuenca:2019jpv}.
We show how the computation of these PMIs simplifies when SSA is replaced by a weaker constraint, dubbed \textit{Klein's condition} (KC), which follows from the necessary condition for the saturation of subadditivity (SA).
Formulating KC in the language of partially ordered sets, we show that the set of PMIs compatible with KC forms a lattice, and we investigate several of its structural properties. One of our main results is the identification of a specific lower dimensional face of the SA cone that contains on its boundary all the extreme rays (beyond Bell pairs) that can possibly be realized by quantum states. We verify that for four or more parties, KC is strictly weaker than SSA, but nonetheless the PMIs compatible with SSA can easily be derived from the KC-compatible ones. For the special case of 1-dimensional PMIs, we conjecture that KC and SSA are in fact equivalent. To make the presentation self-contained, we review the key ingredients from lattice theory as needed.}

\begin{document}
 
 
\maketitle

\section{Introduction}

An important endeavor in quantum information theory is the derivation of universal constraints satisfied by information quantities that characterize various properties of quantum states. For some information quantities, the classification of these constraints has been completed, in the sense that it has been possible to show that no additional independent constraints can exist besides the known ones. For example, this is the case for the quantum relative entropy, as it was shown in \cite{Ibinson_2006} that it cannot satisfy inequalities other than \textit{non-negativity} (also known as \textit{Klein's inequality}) and \textit{monotonicity} (also known as the \textit{data processing inequality}). 

Interestingly, this program is much more complicated in the case of the von Neumann entropy, and only two fundamental inequalities are currently known, namely \textit{subadditivity} (SA) and \textit{strong subadditivity} (SSA), which may be understood respectively as non-negativity and monotonicity of \textit{mutual information}, a measure of the total amount of correlation between subsystems.  Since both these inequalities are implied by the monotonicity of relative entropy, one may wonder if these could be the only inequalities. The fact that this cannot quite be the case is clarified by the existence of \textit{constrained} inequalities \cite{Linden:2004ebt,Cadney_2012}, i.e., inequalities which are satisfied only when other inequalities are saturated. However, it remains an open question whether additional unconstrained inequalities exist. The eventual discovery of new inequalities would likely have important implications, since the saturation of an entropy inequality is typically associated with a particular structure of the density matrix with interesting properties. For example, the saturation of subadditivity is related to independence between subsystems, while the saturation of strong subadditivity is attained only when the density matrix is a quantum Markov chain \cite{Hayden_2004}.

While these investigations often aim at the analysis of constraints which are universal, in the sense that they hold for arbitrary quantum states,  a similar program can also be pursued for restricted classes of states, like stabilizer states or classical probability distributions. For classical probabilities, the von Neumann entropy reduces to the Shannon entropy, which additionally satisfies \textit{monotonicity} and, for four or more parties, infinitely many new inequalities \cite{Zhang1998OnCO, Makarychev:2003, Dougherty:2006, Matus2006}. Similarly, the entropy of stabilizer states has been shown to satisfy all ``balanced'' classical inequalities \cite{Gross_2013} and a certain type of \textit{linear rank inequalities} \cite{LindenRuskaiWinter}.

A class of particular interest for us is that of \textit{geometric states} in quantum gravity, in the context of gauge/gravity duality \cite{Maldacena:1997re,Gubser:1998bc,Witten:1998qj}. More precisely, by a geometric state we mean a state of a holographic conformal field theory which describes a classical bulk geometry. In this context, the von Neumann entropy can be computed using the HRRT prescription \cite{Ryu:2006bv,Hubeny:2007xt}, and has likewise 
been shown to satisfy new inequalities specific to this class. The first of such \textit{holographic entropy inequalities}, known as \textit{monogamy of mutual information}, was found in \cite{Hayden:2011ag}, while a systematic search for new inequalities was started in \cite{Bao:2015bfa}. In particular, it was shown in \cite{hayden2016holographic} that this class of states satisfies all stabilizer entropy inequalities.

Drawing on the ideas of \cite{Hubeny:2018trv,Hubeny:2018ijt}, some of the authors argued in the recent work \cite{Hernandez-Cuenca:2022pst} that all holographic entropy inequalities can be derived from the solution of a conceptually simpler problem, dubbed the \textit{holographic marginal independence problem} (HMIP). This is the holographic version of the general \textit{quantum marginal independence problem} (QMIP) introduced in \cite{Hernandez-Cuenca:2019jpv} for arbitrary quantum states. Intuitively, the QMIP asks the following question: For a given multipartite quantum system, if we demand that certain collections of subsystems are independent while \textit{all} others manifest some correlation, is there a density matrix whose marginals satisfy these demands? An arbitrary choice of such demands was defined in \cite{Hernandez-Cuenca:2019jpv} as a ``pattern of marginal independence'' (PMI), so the QMIP asks which PMIs can be realized by arbitrary quantum states.

As mentioned above, the independence among subsystems is captured by the vanishing of mutual information, or equivalently by the saturation of SA. Since for a multipartite quantum system there exist multiple linearly dependent instances of SA, many combinations are trivially not realizable by either linear dependence or SA. By taking into account these basic restrictions, one can then focus on a more meaningful set of possibilities. Following \cite{Hernandez-Cuenca:2022pst}, we will then update the original definition of a PMI from \cite{Hernandez-Cuenca:2019jpv} to ignore all the possible demands which are incompatible with linear dependence and SA. As we will review, a convenient way to formalize the QMIP with this notion of a PMI is via the construction of polyhedral cone in the space of entropy vectors,\footnote{\, As we will review in detail, an entropy vector is a vector whose components are the entropies of all subsystems of a system described by a given density matrix.} carved by all instances of SA, which will be called the \textit{subadditivity cone} (SAC). The PMIs can then be defined as the linear subspaces spanned by the faces of the SAC. 

For any number of parties, the SAC is an outer bound to the quantum entropy cone, and one can easily imagine that most of these PMIs are non-realizable in quantum mechanics. This is indeed the case, and at least for a small number of parties, the key reason is SSA. Specifically, for most faces of the SAC, all vectors in their interior violate SSA and therefore are not entropy vectors for any density matrix. The goal of the present work is to continue the general analysis initiated in \cite{Hernandez-Cuenca:2019jpv} for arbitrary quantum states, explore in more detail the structure of the set of PMIs which are allowed by SSA, and suggest a new route for the explicit computation of this set. 

From a quantum information theoretical perspective, the analysis presented here is a necessary step, following \cite{Hernandez-Cuenca:2019jpv}, towards the solution to the QMIP. This in turn could lead to new information about the structure of the full quantum entropy cone, as well as an important set of restrictions on the form of possible new entropy inequalities. For the class of geometric states, the hope is that this would allow us to answer a deep question about the very origin of the holographic constraints on the entropy. The results of \cite{Hernandez-Cuenca:2022pst} indeed strongly suggest that the essence of these constraints is captured by the set of holographically realized extreme rays of the SAC. Quantum mechanics imposes severe constraints on the possible extreme PMIs, but does the holographic set-up impose even more stringent restrictions? If it does not, then combined with the result of \cite{Hernandez-Cuenca:2022pst}, we would have a complete characterization of the holographic entropy constraints. Otherwise, there must exist even more fundamental constraints whose origin remains elusive.

To explicitly construct the set of PMIs which are allowed by SSA, one way to proceed would be to slice the SAC with all instances of SSA, and then determine the PMIs of the faces of the resulting cone \cite{Hernandez-Cuenca:2019jpv}. This procedure however is highly inefficient, since many faces would correspond to the same PMIs. Therefore, as we will explain in detail, to characterize the set of PMIs which are allowed by SSA, we will consider a different constraint, which is weaker than SSA but has the advantage of being formulated solely in terms of SA. This constraint on the form of a PMI, which we will call \textit{Klein's condition} (KC), arises from the fact that the factorization of a density matrix is not only a sufficient condition for the saturation of SA, but also a necessary one.

Physically, KC states that any subsystems of independent systems remain independent.  Since this is a constraint solely on the form of a collection of SA inequalities, we can impose it on the set of PMIs without having to introduce new entropy inequalities. More specifically, KC will be phrased in terms of \textit{down-sets} in a certain partially ordered set of instances of mutual information, and it can be thought of as a \textit{combinatorial} constraint on the \textit{lattice of faces} of the SAC. This formulation will allow us to show that the set of PMIs compatible with KC has the structure of a lattice, and by constructing this lattice explicitly for $\N \leq 4$, we prove a series of results about the structure of this lattice for an arbitrary number of parties. The usefulness of this formulation will be most evident in the proof of a theorem that identifies, for any number of parties, a specific face of the SAC which contains all the extreme rays that can possibly be realized by quantum states. Since this face is lower dimensional, this result dramatically simplifies the explicit computation of these extreme rays \cite{He:2022wip}.

This analysis will also show that KC and SSA are not equivalent constraints for the QMIP, and that KC can only be seen as an approximation to SSA for the QMIP.\footnote{\, While it is obvious that KC is strictly weaker than SSA for vectors of entropy space, this result is non-trivial for PMIs, since the compatibility of a PMI with SSA only requires the existence of a single SSA-satisfying vector on the corresponding face of the SAC.} Indeed, for $\N \geq 4$, we will prove the existence of KC-compatible PMIs corresponding to faces of the SAC that violate SSA in the interior and therefore are not realized by quantum states. Nevertheless, we will show that the set of PMIs allowed by SSA is also a lattice, with the same meet operation as the lattice of KC-compatible PMIs, and we will comment on how this property allows for the construction of the former lattice from the latter. Furthermore, by taking advantage of the computational speedup obtained from \Cref{thm:new_coatoms}, we will be able to construct all extreme rays of the SAC that are KC-compatible for $\N \leq 5$ and verify that they all satisfy SSA. This observation will motivate us to conjecture that for an arbitrary number of parties, the approximation of SSA by KC is in fact exact for 1-dimensional PMIs. 

The plan of the paper is as follows. In \S\ref{sec:pmi_and_KC} we first review the precise definition of a PMI along with the notions of realizability and compatibility with SSA, and then we introduce KC and the poset of instances of the mutual information. In \S\ref{sec:KC_PMIs} we look at the structure of the set of PMIs satisfying KC, showing that it is a lattice, deriving some of its properties (for an arbitrary number of parties), and explaining their implications for the QMIP. In \S\ref{sec:realizability} we comment on realizability, the structure of the set of realizable PMIs for general classes of quantum states, and the relation between KC and SSA. We end in \S\ref{sec:discussion} with an extensive discussion of open questions and interesting future directions of investigation.

\section{Patterns of marginal independence and Klein's condition}
\label{sec:pmi_and_KC}

We begin by reviewing in \S\ref{subsec:cones} the basic definitions of entropy vectors, entropy space, and entropy cones. We next introduce in \S\ref{subsec:sac_pmi} the mutual information arrangement (MIA), the subadditivity cone (SAC), and the pattern of marginal independence (PMI). The notion of a PMI was originally introduced in \cite{Hernandez-Cuenca:2019jpv}, but we will instead use the updated, more ``geometric'' definition from \cite{Hernandez-Cuenca:2022pst}. We will discuss the realizability of PMIs in \S\ref{subsec:ssa_k_condition}, where we give a precise definition for a PMI to be compatible with strong subadditivity (SSA), and also introduce KC as a necessary condition for such compatibility. Finally, in \S\ref{subsec:mi_poset}, we introduce the poset of instances of mutual information, and show how KC can naturally be phrased in terms of its down-sets.

\subsection{Preliminaries: entropy vectors and the quantum entropy cone}
\label{subsec:cones}

Consider an $\N$-party quantum system, where we denote the parties numerically as $1,2,\ldots,\N$.\footnote{\, Throughout this paper, for small $\N$ we will also oftentimes  denote the parties alphabetically as $A,B,C,\ldots$.} Given a density matrix $\rho$ corresponding to a state in this system, we can associate to it a vector $\vec\ent(\rho)$ whose components are the subsystem entropies $\ent_\I$, where $\I \subseteq [\N] \equiv \{1,2\ldots,\N\}$.\footnote{\, We will always implicitly restrict to $\I\neq\varnothing$. \label{fn:nonempty}} These vectors live in the vector space $\mathbb{R}^{\D}$, with $\D=2^{\N}-1$, referred to as \textit{entropy space}. 

While every density matrix corresponds to a vector in entropy space, not every vector in entropy space is the vector of entropies of a density matrix. We will refer to any vector in entropy space as an \textit{entropy vector}, even if it is not the vector of entropies for any density matrix, and say that it is \emph{realizable} if such a density matrix exists. A trivial illustration of the fact that not all entropy vectors are realizable is that any entropy vector that is not in the positive orthant of entropy space is not realizable, since entropies are constrained to be non-negative, i.e., $\ent_\I \geq 0$ for all $\I$.

It was shown in \cite{1193790} that, for an arbitrary number of parties $\N$, 
the set of realizable entropy vectors is a convex cone in entropy space, called the $\N$-party \emph{quantum entropy cone} (QEC$_\N$).\footnote{\, More precisely, it is the topological closure of the set of realizable entropy vectors which is a convex cone \cite{1193790}.} The QEC$_\N$ is known for $\N=2,3$, and it coincides with the polyhedral cone specified by all instances of the following inequalities:
\begin{align}
\label{eq:basic_ineq1}
    & \ent_{\I}+\ent_{\K}\geq \ent_{\I\K}\qquad  & \text{subadditivity (SA)}\\
    & \ent_{\I}+\ent_{\I\K}\geq \ent_{\K}\qquad  & \text{Araki-Lieb (AL)} \label{eq:basic_AL} \\
    & \ent_{\I\K}+\ent_{\J\K}\geq \ent_{\K} + \ent_{\I\J\K} & \text{strong subadditivity (SSA)} \label{eq:basic_SSA} \\
    & \ent_{\I\K}+\ent_{\J\K}\geq \ent_{\I} + \ent_{\J} & \text{weak monotonicity (WM)}
\end{align}  
where we use the shorthand notation $\ent_{\I\K}$ for $\ent_{\I\cup\K}$, and we always assume that $\I\cap\K=\varnothing$ (and similarly for other pairs of indices).
As we already mentioned in the introduction, for larger $\N$ no other unconstrained inequality is known.\footnote{\, The reason why we have not included the non-negativity of the entropy in the list above is that it is not a fundamental inequality, as can immediately be seen by adding SA and AL.\label{fn:nonnegS}}

The quantum entropy cone is clearly symmetric under all permutations of the $\N$ parties, but its symmetry group is actually Sym$_{\N+1}$ rather than just Sym$_{\N}$. The reason for this enhanced symmetry is the existence of a \textit{purification} for any density matrix, and the fact that for pure states the entropy of a subsystem and that of its complement are equal. Explicitly, given an $\N$-party density matrix $\rho_{\N}$, we can construct an auxiliary Hilbert space corresponding to an additional party, for convenience denoted by the number $0$ (or by the letter $O$), which we dub the \textit{purifier}, and a pure state $\ket{\psi}$ in the enlarged Hilbert space such that
\begin{equation}\label{eq:pure}
    \rho_{\N}=\text{Tr}_{0}\ket{\psi}\bra{\psi}\,.
\end{equation}
Introducing a new index 
$\uI \subseteq $\psys{\N}$ \coloneqq \{0,1,\ldots,\N\}$,\footnote{\, We use $\I,\J,\K$ to denote subsets of the parties that \textit{do not} include the purifier, and $\uI,\uJ,\uK$ to denote subsets of the parties that \textit{can} include the purifier.} we then have for all $\uI$
\begin{align}
\label{eq:purification}
    \ent_{\uI} = \ent_{\comp{\uI}}\,.
\end{align}
where the superscript $c$ on $\uI$ indicates complement, or $\comp{\uI} \coloneqq \psys{\N} \setminus \uI$.

Given an inequality involving any of the $\N+1$ parties, we can then use this symmetry to obtain a new inequality by permuting the $\N+1$ parties and using \eqref{eq:purification} to remove the purifier when it appears. As an example, for $\N=2$ with parties $A,B$ and purifier $O$, we can transform an instance of the Araki-Lieb inequality into an instance of subadditivity as follows:
\begin{align}
\begin{split}
    &\ent_A + \ent_{AB} \geq \ent_{B} \\ 
    \implies\quad & \ent_A + \ent_{AO} \geq \ent_{O} \\
    \implies\quad & \ent_{A} + \ent_B \geq \ent_{AB}\,,
\end{split}
\end{align}
where we obtained the second line by exchanging $B$ and $O$, and in the last line we applied the purification symmetry. We leave it as an exercise for the reader to verify that one can similarly map an instance of weak monotonicity to an instance of strong subadditivity and vice-versa.

In the rest of the paper we will find it useful to always exploit this symmetry and write the inequalities of \eqref{eq:basic_ineq1} as
\begin{align}
\label{eq:basic_ineq2}
    & \ent_{\uI}+\ent_{\uK}\geq \ent_{\uI\uK}\qquad  & \text{subadditivity (SA)}\\
    & \ent_{\uI\uK}+\ent_{\uJ\uK}\geq \ent_{\uK} + \ent_{\uI\uJ\uK} & \text{strong subadditivity (SSA)}\,,
\end{align}
therefore including all instances of the Araki-Lieb inequality in the set of instances of subadditivity, and all instances of weak monotonicity in the set of instances of strong subadditivity. Throughout this work the purifier will always be considered an ancillary party that \textit{could} be added to the $\N$-party system to purify a state, but is never really part of an $\N$-party system. It will merely be viewed as a tool to write instances of AL as instances of SA.

\subsection{Patterns of marginal independence and the subadditivity cone}
\label{subsec:sac_pmi}

We now recall the notion of a \textit{pattern of marginal independence} (PMI) given in \cite{Hernandez-Cuenca:2022pst}. Before giving the formal definition, let us briefly comment on the intuition behind it and introduce some of the geometric constructs that will be used for the definition. Given an $\N$-party density matrix $\rho_{\N}$, the \textit{mutual information} (MI) between any non-intersecting pair of subsystems $\uI$ and $\uK$ is defined as
\begin{align}\label{eq:MI_def}
    \mi(\uI:\uK) \coloneqq \ent_{\uI} + \ent_{\uK} - \ent_{\uI\uK}\,.
\end{align}
The MI is a measure of the total amount of correlation between two subsystems, and it vanishes if and only if the two subsystems are \textit{independent}, i.e., if and only if the \textit{marginal} $\rho_{\uI\uK}$ of $\rho_{\N}$ factorizes:\footnote{\, More precisely, if $\uI$ or $\uK$ includes the purifier, $\rho_{\uI\uK}$ is a marginal of the density matrix $\ket{\psi}\!\bra{\psi}$, where $\ket{\psi}$ is the purification defined in \eqref{eq:pure}.} 
\begin{equation}
\label{eq:mi_vanishing}
    \mi(\uI:\uK)=0\qquad \Longleftrightarrow\qquad \rho_{\uI\uK}=\rho_{\uI}\otimes\rho_{\uK}\, .
\end{equation}

Given an $\N$-party quantum system, we denote by $\mgs$ the collection of all instances of MI (including the purifier).\footnote{\, It is clear from \eqref{eq:MI_def} that MI is manifestly symmetric in its arguments, i.e., $\mi(\uI:\uK) = \mi(\uK:\uI)$. The collection of MI instances in $\mgs$ takes this into account, so that to avoid over-counting, for every non-intersecting pair $\uI$ and $\uK$, $\mgs$ will include only $\mi(\uI:\uK)$ or $\mi(\uK:\uI)$, but not both.} Consider now a bipartition $(\mgs^0,\mgs^*)$ of $\mgs$, and suppose that we want to construct a state for the system such that all instances of MI in $\mgs^0$ vanish, while all the instances in $\mgs^*$ do not. One may wonder whether it is useful to define a PMI to be any such bipartition. However, since we are interested in the realizability of a given pattern, defining the complete set of patterns such that most of them are manifestly unrealizable is not the most convenient starting point.  Instead, we will make an initial reduction of the full set to a more manageable subset by eliminating patterns that are trivially unrealizable. 

One immediate class of such unrealizable patterns stems from the fact that the instances of MI are not all linearly independent. As a simple example, consider the case of $\N=3$, and suppose that we specify a bipartition of $\mgs$ such that
\begin{align}\label{eq:dis_bip}
\begin{split}
    & \{\,\mi(A\!:\!B),\,\mi(B\!:\!C),\,\mi(AB\!:\!C)\}\subseteq\, \mgs^0 \\
    & \{\,\mi(A\!:\!BC)\}\subseteq\, \mgs^* \,.
\end{split}
\end{align}
Then a density matrix which satisfies these requirements cannot exist simply because of the following identity: 
\begin{align}\label{eq:ld}
    \mi(A\!:\!B) + \mi(AB\!:\!C) - \mi(B\!:\!C) - \mi(A\!:\!BC) = 0 \,.
\end{align}

To take such trivial constraints into account, so that we can ignore all bipartitions of $\mgs$ violating linear dependence, it is convenient to use a geometric description. For any MI instance $\mi(\uI:\uK)$, the equation
\begin{equation}
\label{eq:mihyp}
    \ghyp:\qquad \mi(\uI:\uK)=0
\end{equation}
specifies a \textit{hyperplane} $\ghyp$ in the entropy space. The collection of all hyperplanes of this form, for any fixed $\N$, is a hyperplane arrangement in entropy space, which we now define.
\begin{defi}[MI hyperplane arrangement (MIA$_{\N}$)]
    For a given number of parties $\N$, the mutual information hyperplane arrangement is the collection of all hyperplanes 
    $\ghyp$ specified by \eqref{eq:mihyp}.
\end{defi}

Using this definition,\footnote{\, The reader familiar with the definition given in \cite{Hernandez-Cuenca:2019jpv} will notice that here we additionally include the hyperplanes where $\uI\cup\uK = \psys{\N}$. As we will discuss (cf.\ \Cref{lem:facet}), these hyperplanes do not correspond to facets of the subadditivity cone, so this inclusion will have no effect on the set of PMIs defined below.}
it should then be clear that the bipartitions of $\mgs$ that are compatible with linear dependence correspond precisely to the linear subspaces of entropy space that are intersections of hyperplanes, which are called the \textit{flats} of the arrangement.\footnote{\, By convention, the set of flats also includes the entire entropy space, which can be thought of as the intersection of the empty collection of hyperplanes.} Indeed, the disallowed bipartition \eqref{eq:dis_bip} would never be considered, since \eqref{eq:ld} implies that the intersection of the hyperplanes $\mi(A\!:\!B) = 0$, $\mi(B\!:\!C) = 0$ and $\mi(AB\!:\!C) = 0$ must lie on the hyperplane $\mi(A\!:\!BC) = 0$. 

One may then be tempted to define a PMI to be a flat of the MIA. However, this is likewise inconvenient, as most flats correspond to bipartitions of $\mgs$ that are disallowed by another basic constraint, namely SA. Rephrased in terms of  MI, SA is simply the statement that MI is always non-negative, and specifies a \textit{halfspace} $\ghyp^+$ of entropy space:
\begin{align}
\label{eq:mi_sa}
\ghyp^+:\qquad    \mi(\uI:\uK) \geq 0\, .
\end{align}
Consider now the following bipartition of $\mgs$ (again for $\N=3$):
\begin{equation}
\label{eq:dis_bip2}
    \mgs^0 = \{\,\mi(A\!:\!B),\,\mi(AB\!:\!C)\},\qquad 
    \mgs^* = \mgs\setminus\mgs^0\, .
\end{equation}
We leave it as an exercise to the reader to verify that there is no instance of MI in $\mgs^*$ which is a linear combination of the instances in $\mgs^0$, so this bipartition corresponds to a flat of MIA$_3$. However, notice that \eqref{eq:ld} then reduces to 
\begin{equation}
\label{eq:flat_sa_violation}
     - \mi(B\!:\!C) - \mi(A\!:\!BC) = 0\,,
\end{equation}
and \eqref{eq:mi_sa} then implies that $\mi(B\!:\!C)=\mi(A\!:\!BC)=0$. Therefore, the bipartition in \eqref{eq:dis_bip2} is ruled out not because it violates linear dependence, but because any non-vanishing entropy vector that satisfies \eqref{eq:dis_bip2} and \eqref{eq:flat_sa_violation} cannot satisfy all instances of SA.

To also take into account these implications of SA, we again use a geometric construction. The collection of all halfspaces of the form \eqref{eq:mi_sa} specifies the following convex polyhedral cone.\footnote{\, Convexity follows from the fact that we are intersecting halfspaces, polyhedrality from the finiteness of the set of SA instances, and the fact that the resulting convex polyhedron is a cone is implied by the fact that SA is a homogeneous inequality.}

\begin{defi}[Subadditivity cone (SAC$_{\N}$)]
For a given number of parties $\N$, the \emph{subadditivity cone} is the polyhedral cone in entropy space specified by the intersection of all halfspaces $\ghyp^+$ given by \eqref{eq:mi_sa}.
\end{defi}

We define a \textit{face} $\face$ of the SAC$_{\N}$ as the intersection of the SAC$_{\N}$ with an arbitrary linear hyperplane such that the SAC$_{\N}$ is entirely contained in a closed halfspace. Equivalently, we can obtain lower-dimensional faces from intersections of higher-dimensional ones.  Two particularly important classes of faces are the \textit{facets}, which are $(\D-1)$-dimensional faces (corresponding to a region of a given hyperplane $\ghyp$ of \eqref{eq:mihyp} resulting from its intersection with the SAC$_\N$), 
and the \textit{extreme rays}, which are the $1$-dimensional faces. It is conventional to include in the set of faces also a $\D$-dimensional face, which is the entire SAC$_\N$. 

Using the above definition of a face, we define a PMI as the following linear subspace of entropy space:\footnote{ \, This definition is the same as the one used in \cite{Hernandez-Cuenca:2022pst}, which improves the original one given in \cite{Hernandez-Cuenca:2019jpv} to account for SA.}

\begin{defi}[Pattern of marginal independence (PMI)]
\label{def:pmi}
For a given number of parties
$\N$, a pattern of marginal independence $\pmi$ is the linear span of a face of the \emph{SAC$_{\N}$}.
\end{defi}

One natural specification of a given PMI is the set $\mgs^0$, since this automatically prescribes the bipartition $(\mgs^0,\mgs^* \coloneqq \mgs\setminus\mgs^0)$ of $\mgs$. This definition is motivated by the previous discussion regarding the consistency of a bipartition of $\mgs$,  both with linear dependence among the MI instances and with all instances of SA. Indeed, any PMI as defined above is a flat of the MIA, and therefore respects linear dependence. Furthermore, it is guaranteed by construction to contain at least one entropy vector that satisfies all instances of SA, as it contains the corresponding face of the SAC$_\N$. In fact, the PMIs as defined are \textit{precisely} the bipartitions of $\mgs$ with these properties (see \cite{Hernandez-Cuenca:2019jpv} for more details).\footnote{\, Note however that even though PMIs satisfy both linear dependence and SA, most of them are in fact still unrealizable by any quantum state, because they violate SSA or other inequalities.  We will return to SSA in \S\ref{subsec:ssa_k_condition}.
}

For any given $\N$, we will denote the set of all PMIs by $\lat{PMI}$, and similarly, we will denote the set of all faces of the SAC$_{\N}$ by $\lat{SAC}$ 
(when we wish to specify $\N$ explicitly we will include it as a superscript, but to avoid cluttering the notation, we will often keep the $\N$ dependence implicit). As we will discuss in detail in \S\ref{sec:KC_PMIs}, these sets have the structure of a lattice,\footnote{\, A lattice is a particular type of partially ordered set that additional structural properties (see \Cref{def:lattice}).} hence the choice of notation. We will denote a PMI by $\pmi$, and with a slight abuse of notation, we will use the same symbol to refer to both its \textit{geometric description} as the set of vectors in entropy space that belong to that subspace, and to the corresponding \textit{combinatorial description} as a bipartition of $\mgs$. It should always be clear from context which particular representation of $\pmi$ we have in mind.  In a similar fashion, we will likewise use the notation $\face$ for the dual purpose of referring to a (geometric) face of the SAC$_{\N}$ as well as a (combinatorial) element of the set $\lat{SAC}$.

On the other hand, we will be careful to distinguish a PMI from the corresponding face of the SAC, since geometrically these are different objects. Nevertheless, for any fixed $\N$, there is the following bijection between the faces of the SAC and the PMIs:
\begin{align}
\label{eq:lambda}
    \mu :\quad  \lat{SAC}\; & \to\;\lat{PMI}\nonumber\\
     \face\; & \mapsto\; \pmi=\mu(\face)\coloneqq\text{span}(\face)\,,
\end{align}
with its inverse given by
\begin{align}
\label{eq:lambdainv}
    \mu^{-1} :\quad \lat{PMI}\; & \to\;\lat{SAC}\nonumber\\
     \pmi\; & \mapsto\; \face=\mu^{-1}(\pmi)=\pmi\cap \text{SAC}_{\N}\,.
\end{align}
We will prove in \S\ref{ssec:proof_lattice} that $\mu$ is not only a bijective map between $\lat{SAC}$ and $\lat{PMI}$, but that for any $\N$ it is in fact a lattice isomorphism.

We conclude this subsection with a mathematical observation about the instances of SA
that correspond (when saturated) to facets of the SAC$_{\N}$. We defined the SAC$_{\N}$ above as the intersection of a collection of halfspaces corresponding to some linear inequalities. In such a construction, an inequality corresponds to a facet of the resulting cone if it is \textit{non-redundant}, i.e., if it is not a positive linear combination of other inequalities in the set. We now prove that for any $\N$, the redundant instances of SA are precisely those of the form $\mi(\uI:\comp{\uI}) \geq 0$, which reduce to non-negativity of the entropy:\footnote{\, In deriving the r.h.s.\ of \eqref{eq:trivial_ins}, we used the fact that $\ent_{\uI\comp{\uI}}=0$, since $\uI\comp{\uI}=\psys{\N}$ comprises of the full system and the purifier.}
\begin{equation}
\label{eq:trivial_ins}
    0\leq\mi(\uI:\comp{\uI})=\ent_{\uI}+ \ent_{\comp{\uI}}-\ent_{\uI\comp{\uI}}=2\ent_{\uI}\,.
\end{equation}

\begin{lemma}\label{lem:facet}
For an arbitrary number of parties $\N$, a halfspace $\mi(\uI:\uK) \geq 0$ corresponds\footnote{\, In the sense that the span of the facet is the hyperplane $\mi(\uI:\uK)= 0$.} to a facet of the \emph{SAC}$_\N$ if and only if $\uK \neq \comp{\uI}$.
\end{lemma}
\begin{proof}
    We will prove the lemma in two steps.

    i) As shown in \eqref{eq:trivial_ins}, if $\uK = \comp{\uI}$, then SA reduces to the non-negativity of the entropy, which (as we already mentioned in \Cref{fn:nonnegS}) is trivially implied by adding \eqref{eq:basic_ineq1} and \eqref{eq:basic_AL}. 
    
    ii) Conversely, to prove that any instance of SA with $\uK \neq \comp{\uI}$ is not redundant, we will show that it cannot be obtained as a linear combination of other MI instances. Consider a particular instance $\mi(\uI^*\!:\!\uK^*)$ where $\uK^* \not= \uI^{*c}$, and suppose that it can be written as a positive linear combination of a collection of $n$ instances $\mi(\uI^i\!:\!\uK^i)$ for $i=1,\ldots,n$. By the argument in (i), we may without loss of generality assume that $\uK^i \not= (\uI^i)^c$ for all $i$. Denoting by $\vec\mi(\uI:\uK)$ the vector of coefficients of the entropies in the expression \eqref{eq:MI_def} of $\mi(\uI:\uK)$, i.e., the normal vector to the hyperplane \eqref{eq:mihyp} with the appropriate orientation, we can write\footnote{\, For all $i\leq n$, if either $\uI^i$ or $\uK^i$ includes the purifier, we use the purification symmetry \eqref{eq:purification} to rewrite $\mi(\uI^i\!:\!\uK^i)$ in terms of the entropies $\ent_{\J}$, with $\J\subseteq [\N]$. We similarly do this for $\uI^*$ and $\uK^*$.}
    \begin{align}
    \label{eq:vec_comb}
        \vec\mi(\uI^*\!:\!\uK^*) = \sum_{i=1}^n \alpha_i\; \vec\mi(\uI^i\!:\!\uK^i), \quad\text{where $\alpha_i > 0$\,.}
    \end{align}
    Notice that any vector $\vec\mi(\uI:\uK)$ has precisely two $+1$ components and one $-1$ component, so summing all the components of each vector on both sides of \eqref{eq:vec_comb} we get
    \begin{equation}
    \label{eq:alpha_cons}
    1=\sum_{i=1}^n \alpha_i \,.
    \end{equation}
    Consider now the component of $\vec\mi(\uI^*:\uK^*)$ corresponding to the entropy of $\uI^*$ (or $\uI^{*c}$ if $\uI^*$ includes the purifier), which must be $+1$. Due to \eqref{eq:alpha_cons}, the r.h.s.\ of \eqref{eq:vec_comb} can have a $+1$ in the $\uI^*$ (or $\uI^{*c}$) component only if every $\vec\mi(\uI^i:\uK^i)$ also has $+1$ in the $\uI^*$ (or $\uI^{*c}$) component. This means 
    \begin{equation} \label{eq:I_constraint}
        \uI^i=\uI^*\quad \text{or}\quad \uI^i = \uI^{*c} \quad \forall\;i\, ,
    \end{equation}
    and we can then rewrite \eqref{eq:vec_comb} as
    \begin{align}
    \label{eq:vec_comb2}
        \vec\mi(\uI^*:\uK^*) = \sum_{i:\,\uI^i = \uI^*} \alpha_i\; \vec\mi(\uI^*:\uK^i) + \sum_{i:\, \uI^i = \uI^{*c}} \alpha_i\; \vec\mi(\uI^{*c}:\uK^i)\,.
    \end{align}
   Following the same reasoning for $\uK^*$, we need to impose on \eqref{eq:vec_comb2} the analogue of \eqref{eq:I_constraint} for $\uK^*$. However, notice that the assumption $\uK^* \not= \uI^{*c}$ (along with  $\uI^* \cap \uK^{*} = \varnothing$) means that out of the four possible combinations of $(\uI^i,\uK^i)$, the only pair with disjoint arguments is the one with $\uI^i = \uI^{*}, \uK^i = \uK^{*}$, which means the sum trivializes to all terms having the form of the l.h.s.\ of \eqref{eq:vec_comb2}. This completes the proof. 
\end{proof}

\subsection{Realizability of PMIs, strong subadditivity, and Klein's condition}
\label{subsec:ssa_k_condition}

Having defined a PMI, we now turn to the notion of its realizability. We begin by reviewing the map that associates a PMI to an entropy vector in the SAC$_{\N}$ from \cite{Hernandez-Cuenca:2019jpv}: 
\begin{align}
\label{eq:pmimap}
\pmimap:\quad &\text{SAC}_{\N}\; \to\; \lat{PMI}\nonumber\\
& \vec\ent\; \mapsto\; \pmi=\pmimap(\vec\ent)= \bigcap_{\vec\ent\in\pmi'} \pmi' \,.
\end{align}
In practice, this map associates an entropy vector $\vec\ent\in\text{SAC}_{\N}$ to the PMI $\pmi$ of \textit{lowest dimension} that contains $\vec\ent$. Since the map $\pi$ is restricted to the SAC$_{\N}$, from now on we will always implicitly restrict to vectors of entropy space that belong to SAC$_{\N}$ (in particular, they are compatible with all instances of SA). The existence and uniqueness of such a PMI should be intuitively clear, and will be formalized and guaranteed by the aforementioned lattice structure of $\lat{PMI}$ in \S\ref{sec:KC_PMIs} (see \cite{Hernandez-Cuenca:2019jpv} for more details). 

Denoting by $\Omega$ an arbitrary class of quantum states, which can either be the set of all possible quantum states, or a more restricted class like the set of stabilizer states, we introduce the following definition.

\begin{defi}[Realizable PMI in $\Omega$] \label{def:realizable_pmi}
For a given class of states $\Omega$ and a \emph{PMI} $\pmi$, we say that $\pmi$ is \emph{realizable in $\Omega$} if there exists a state $\rho\in\Omega$ such that $\pmimap(\vec\ent(\rho)) = \pmi$.
\end{defi}

For any given $\N$, it easy to see that because of strong subadditivity, many PMIs are actually not realizable by any quantum state. The reason is simply the fact that for a PMI $\pmi$, it can happen that all entropy vectors in the interior of the corresponding face $\face=\mu^{-1}(\pmi)$ of the SAC violate at least one instance of SSA. 

Since SSA is a universal quantum mechanical inequality, it is satisfied by all classes of states of interest. Therefore, we should only focus on the set of PMIs that are not excluded by SSA, leading us to the following definition.

\begin{defi}[SSA-compatible PMI]
A \emph{PMI} $\pmi$ is said to be \emph{compatible} with \emph{SSA} if there exists an entropy vector $\vec\ent$ such that $\pmimap(\vec\ent) = \pmi$ and $\vec\ent$ satisfies all instances of 
\emph{SSA}.
\end{defi}

In light of this definition, one would like to further characterize the set of SSA-compatible PMIs, and one way to proceed in this direction would be to construct a new entropy cone, carved by all instances of SA and SSA, and consider its faces. However, this would give rise to a rich, yet irrelevant, structure, since there would be numerous faces which would correspond to the same PMI.\footnote{\, As a simple example, consider the $\N=3$ case, where the cone specified by SA and SSA (which is the full QEC$_3$) has the extreme ray $\vec\ent=(1,1,1,1,1,1,1)$, i.e., the entropy vector of a density matrix obtained from the $4$-party GHZ state by tracing out any one of the qubits. As one can immediately verify, no instance of SA is saturated by this entropy vector, and its PMI (given by \eqref{eq:pmimap}) is therefore the full space $\mathbb{R}^{\D}$. 
Moreover, a moment's thought reveals that $\mathbb{R}^{\D}$ is also the PMI of any vector in the interior of any face of QEC$_3$ which includes $\vec\ent$ on its boundary, since $\vec\ent$ is in the interior of the SAC$_3$. \label{ft:ssa}} Because of this, we will follow another direction, where we replace SSA with a weaker constraint (which we will call ``Klein's condition'') that can be formulated purely combinatorially via the faces of the SAC, without having to introduce any additional inequality. 

As we discussed in the previous subsection, the definition of a PMI takes into account the linear dependence among the instances of MI, as well as all instances of SA. However, we have not fully taken into account the fact that, as indicated in \eqref{eq:mi_vanishing}, the MI between two subsystems vanishes \textit{if and only if} the density matrix \textit{factorizes}. For completeness, let us first briefly review why this is the case. The easiest way to see this is via the application of Klein's inequality\footnote{\, For all positive-definite Hermitian matrices $\rho$ and $\sigma$, and for all differentiable convex functions $f:(0,\infty) \to \mathbb{R}$, the following inequality holds:
\begin{equation*}
    \text{Tr}(f(\rho)-f(\sigma)-(\rho-\sigma)f'(\sigma))\geq 0\,.
\end{equation*}
Furthermore, if $f$ is strictly convex, equality holds if and only if $\rho=\sigma$. For the quantum relative entropy, one obtains \eqref{eq:monrelent} and \eqref{eq:monrelentsat} by choosing $f(t)=t\;\text{log}\;t$.} to the quantum relative entropy: 
\begin{equation}
\label{eq:monrelent}
    {\sf R}\,(\rho\,||\,\sigma)\;=\;\text{Tr}\,\rho\,\text{log}\,\rho-\text{Tr}\,\rho\,\text{log}\,\sigma\; \geq \;0\qquad \forall\;\rho,\sigma\,,
\end{equation}
and in particular 
\begin{equation}
\label{eq:monrelentsat}
    {\sf R}\,(\rho\,||\,\sigma)\;=\; 0\qquad \Longleftrightarrow\qquad \rho=\sigma\,.
\end{equation}
Setting $\rho \equiv \rho_{\uI\uK}$ and $\sigma \equiv \rho_{\uI} \otimes \rho_{\uK}$, it is easy to see that this translates to \eqref{eq:mi_vanishing}. 

Suppose now that the MI instance $\mi(\uI:\uJ\uK)$ vanishes. We then have
\begin{align}
\label{eq:KC_implications}
    & \mi(\uI:\uJ\uK)=0\nonumber\\
    \implies \quad & \rho_{\uI\uJ\uK}=\rho_{\uI}\otimes\rho_{\uJ\uK}\nonumber\\
    \implies \quad & \rho_{\uI\uJ}=\rho_{\uI}\otimes\rho_{\uJ}\quad \text{and}\quad \rho_{\uI\uK}=\rho_{\uI}\otimes\rho_{\uK}\nonumber\\
    \implies \quad & \mi(\uI:\uJ)=0\quad \text{and}\quad \mi(\uI:\uK)=0 \,.
\end{align}
It is crucial to notice that this implication, which is respected by the instances of MI for any entropy vector of a density matrix, is independent from the constraints that we considered when we defined PMIs. Indeed, as was proven in \Cref{lem:facet}, the hyperplane $\mi(\uI:\uJ\uK)=0$ (for $\uI\uJ\uK \not= \psys{\N}$) supports a facet of the SAC, and is therefore by itself a PMI $\pmi$ (i.e., in the specification of this PMI we are assuming that no other instance of MI vanishes, so that in particular $\mi(\uI:\uJ) \not= 0$ and $\mi(\uI:\uK) \not= 0$). Any vector $\vec\ent$ in the \textit{interior} of the face $\face=\mu^{-1}(\pmi)$ is therefore not compatible with the implications in \eqref{eq:KC_implications}, and it cannot be the entropy vector of a density matrix. Thus, the implication in \eqref{eq:KC_implications} is a necessary condition for the PMI to be realizable by a quantum state. Since this restriction follows immediately from Klein's inequality, we will call it \textit{Klein's condition}.

\begin{defi}[KC-compatible PMIs]\label{def:kcv1}
    \emph{Klein's condition (KC)}  for a \emph{PMI} $\pmi$ specified by a set $\mgs^0$ is the requirement that for all $\uI,\uK$
    \begin{align}
    \label{eq:KC}
         \mi(\uI:\uK)\in\mgs^0 \quad\implies \quad 
            \mi(\uI':\uK')\in\mgs^0 \quad \forall \ 
            \uI' \subseteq \uI \ , \  \uK' \subseteq \uK \, .
    \end{align}
    If $\pmi$ satisfies \emph{KC}, then we call $\pmi$ a \emph{KC-compatible PMI}.
\end{defi}    

It is clear in general KC is weaker than SSA, since SSA can also be rewritten as \emph{monotonicity of mutual information}:
\begin{equation}
\label{eq:ssa}
\mi(\uI:\uJ\uK)-\mi(\uI:\uK)\geq 0 \,,
\end{equation}
from which \eqref{eq:KC} straightforwardly follows.\footnote{\,
More precisely, KC is implied not by a single instance of SSA, but rather by a collection of SSAs and SAs.} In particular, it is immediately obvious that for entropy vectors KC is \textit{strictly} weaker, since (for example) any entropy vector in the interior of the SAC trivially satisfies \eqref{eq:KC}, while not all these vectors satisfy SSA.
For PMIs, this means that KC is a necessary condition for SSA-compatibility. However, it is far from obvious that KC is not effectively equivalent to SSA for PMIs, since in this case SSA-compatibility of a KC-compatible PMI $\pmi$ only requires the existence of a \textit{single} entropy vector which satisfies SSA and whose PMI is $\pmi$. Indeed, one of our results will be that they are not equivalent even for PMIs (cf. \Cref{thm:KC_SSA}). However, as mentioned above, we still view the construction of the set of KC-compatible PMIs as a useful strategy for the derivation of the SSA-compatible ones, and for this reason from now on we will focus primarily on KC rather than on SSA. To this end, it will be useful to first formalize KC in a slightly different fashion, using the language of partially ordered sets. This will be the goal of the next subsection. 

We conclude this subsection by commenting on another condition that needs to be satisfied by a realizable PMI, but, as it turns out, is implied by SA. First recall that in quantum mechanics, the von Neumann entropy of a density matrix $\rho$ vanishes if and only if $\rho$ has unit rank, in which case all subsystems of $\rho$ have the same entropy as their complements (cf.\ \eqref{eq:purification}). Explicitly,
\begin{equation} \label{eq:purity}
    \ent_{\I}(\rho)=0 \qquad \implies \qquad \ent_{\K}=\ent_{\I\setminus\K},\quad \forall\,\K \subset \I\, . 
\end{equation}
However, this implication follows already from SA, in fact independently from the realizability of a PMI, since 
\begin{align}
\begin{split}
    & \ent_{\I\setminus\K}-\ent_{\K}+\ent_{\I}\geq 0\\
    - &\ent_{\I\setminus\K}+\ent_{\K}+\ent_{\I}\geq 0\,
\end{split}    
\end{align}
implies \eqref{eq:purity} after setting $\ent_{\I}=0$.

\subsection{The mutual information poset}
\label{subsec:mi_poset}

In this subsection, we reformulate KC (introduced in \Cref{def:kcv1}) in the language of order theory, by exploiting a certain partial order on the set of MI instances for any fixed $\N$. We begin by reviewing the definition of a partially ordered set (for more details we refer the reader to \cite{birkhoff1967lattice} and \cite{davey1990introduction}). 

Given a set $\mathcal{P}$, a \textit{partial order} on it is a binary relation $\preceq$ which is reflexive, antisymmetric, and transitive, i.e., for any $x,y,z\in \mathcal{P}$ we have
\begin{align}
\begin{split}
    & x\preceq x \\
    & x\preceq y\;\; \text{and}\;\; y\preceq x\quad \implies\quad x=y \\
    & x\preceq y\;\; \text{and}\;\; y\preceq z\quad \implies\quad x\preceq z \,.
\end{split}
\end{align}
A partially ordered set $(\mathcal{P},\preceq)$, or \textit{poset}, is a set $\mathcal{P}$ with a partial order $\preceq$. A partial order on $\mathcal{P}$ gives rise to the following relation $\prec$ of \textit{strict inequality}: 
\begin{equation}
    x\prec y\quad \Longleftrightarrow\quad x\preceq y\;\; \text{and}\;\; x\neq y \,.
\end{equation}
Two elements $x,y\in \mathcal{P}$ are said to be \emph{incomparable} if $x\npreceq y$ and $y\npreceq x$. 

We define the \textit{mutual information poset} as follows:

\begin{defi}[Mutual information poset (MI-poset)]\label{def:MIposet}
    For a given number of parties $\N$, the \emph{mutual information poset} $(\mgs,\preceq)$ is the set $\mgs$ of instances of \emph{MI} with the partial order given by\footnote{\, The two alternatives on the r.h.s.\ are necessary due to the fact that the MI is symmetric, and the partial order that we want to introduce here is insensitive to the ordering of the argument.}
    \begin{equation}
    \label{eq:po-def}
    \mi(\uI:\uK) \preceq \mi(\uI':\uK') \quad \iff\quad \uI\subseteq\uI'\; \text{\emph{and}}\; \uK\subseteq\uK'\quad  \text{\emph{or}}\quad  \uI \subseteq \uK'\; \text{\emph{and}}\; \uK \subseteq \uI'\,.
\end{equation}
\end{defi}
For given $\N$ the MI-poset has cardinality
\begin{equation}
\label{eq:e_def}
    \E=\stirlingII{\N+2}{3}\,,
\end{equation}
where we have already modded out by the symmetry between the arguments of MI and $\stirlingII{n}{k}$ is the Stirling number of the second kind.\footnote{\, In \cite{Hubeny:2018ijt,Hernandez-Cuenca:2019jpv} the number of MI instances was $3\stirlingII{\N+1}{3}$ because the trivial instances $\mi(\uI:\comp{\uI})$ were not included. Indeed one can immediately check that $3\stirlingII{\N+1}{3}+2^{\N}-1=\stirlingII{\N+2}{3}$.} The strict inequality 
\begin{equation}
    \mi(\uI:\uK) \prec \mi(\uI':\uK') 
\end{equation}
is attained when at least one of the inclusions is strict, i.e, when any of the $\subseteq$ in \eqref{eq:po-def} is replaced with $\subset$.

A convenient way of representing a poset, at least when the number of elements is not too large, is via a \textit{Hasse diagram}, whose construction we now review. In a poset $(\mathcal{P},\preceq)$, an element $x$ is said to be \textit{covered by} an element $y$, which we denote as $x\covb y$, if 
\begin{equation}
    \label{eq:cover}
    x\prec y\qquad \text{and}\qquad x\preceq z \prec y\;  \implies \;  z=x \,.
\end{equation}
Equivalently, if this implication holds, $y$ is said to \textit{cover} $x$, and we write $y\cov x$. Essentially, this means that there are no elements between $x$ and $y$ distinct from $x$ and $y$ themselves. The Hasse diagram of a poset $(\mathcal{P},\preceq)$ is a representation of the cover relation between the elements of the poset. It is obtained by first drawing a vertex for each element in the set $\mathcal{P}$, in such a way that for each pair $x,y$, if $x\prec y$ then the vertex for $x$ is drawn below the vertex for $y$. One then connects the vertices for $x$ and $y$ only if $x\covb y$.
As an example, we have drawn the Hasse diagram of the MI-poset for $\N=2$ in \Cref{fig:MIPosetN2}. 

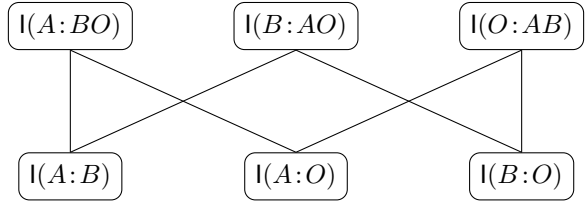
\begin{figure}[tb]
    \centering
    \begin{tikzpicture}
    \node[draw, rounded corners] (abo) at (-3,2) {{\footnotesize $\mi(A\!:\!BO)$}};
    \node[draw, rounded corners] (bao) at (0,2) {{\footnotesize $\mi(B\!:\!AO)$}};
    \node[draw, rounded corners] (oab) at (3,2) {{\footnotesize $\mi(O\!:\!AB)$}};
    \node[draw, rounded corners] (ab) at (-3,0) {{\footnotesize $\mi(A\!:\!B)$}};
    \node[draw, rounded corners] (ao) at (0,0) {{\footnotesize $\mi(A\!:\!O)$}};
    \node[draw, rounded corners] (bo) at (3,0) {{\footnotesize $\mi(B\!:\!O)$}};
    \draw[-] (abo.south) -- (ab.north);
    \draw[-] (abo.south) -- (ao.north);
    \draw[-] (bao.south) -- (ab.north);
    \draw[-] (bao.south) -- (bo.north);
    \draw[-] (oab.south) -- (ao.north);
    \draw[-] (oab.south) -- (bo.north);
    \end{tikzpicture}
    \caption{The Hasse diagram for the MI-poset for $\N=2$.}
    \label{fig:MIPosetN2}
\end{figure}

Given a poset $(\mathcal P, \preceq)$, a special class of subsets that will be of interest to us are the \textit{down-sets} (sometimes also called \textit{order ideals}). A subset $\mathcal{Q}\subseteq\mathcal{P}$ is a down-set if for all $x\in\mathcal{Q}$ and $y\in\mathcal{P}$,
\begin{equation}
\label{eq:ds_def}
	y \preceq x\quad \implies \quad y\in \mathcal{Q}\,.
\end{equation}
Using this definition, we can then rephrase KC for PMIs conveniently as follows.

\begin{lemma}
\label{lem:KCdownset}
    A \emph{PMI} is \emph{KC}-compatible if and only if its set $\mgs^0$ of vanishing \emph{MI} instances is a down-set in the \emph{MI}-poset.
\end{lemma}
\begin{proof}
    If $\mgs^0$ is a down-set in the MI-poset, then \eqref{eq:KC} is automatically satisfied given the partial order relation \eqref{eq:po-def} for MIs. Hence, $\pmi$ satisfies KC by \Cref{def:kcv1}. Conversely, if $\pmi$ is KC-compatible, then $\eqref{eq:KC}$ is precisely the statement that $\mgs^0$ is a down-set with partial order relation \eqref{eq:po-def}, thus proving our claim.
\end{proof}
In what follows, given a KC-compatible PMI, we will often refer to the corresponding down-set of vanishing MI instances as the \emph{$\zds$}.

We conclude this subsection with a few general comments about the structure of the MI-poset, which necessitates us to introduce a few more definitions from the theory of partially ordered sets. A \textit{chain} in a poset $(\mathcal{P},\preceq)$ is a subset $\mathcal{Q}\subseteq\mathcal{P}$ that is totally ordered, i.e., it does not contain any pair of incomparable elements, and takes the form (with $q=|\mathcal{Q}|$)
\begin{equation}
    x_1\prec x_2\prec\ldots\prec x_q \,.
\end{equation}
A chain is \textit{maximal} if it is not contained in any other chain. The \textit{length} of a chain is the number of its elements minus one (the number of ``jumps''), and the \textit{height} $h(x)$ of an element $x$ is the maximal length of a chain whose greatest element is $x$. In the case of the MI-poset, one can easily verify that the height function is given by 
\begin{equation}
\label{eq:height}
    h(\mi(\uI:\uK))=|\uI|+|\uK| - 2 \,,
\end{equation}
where $|\uI|$ and $|\uK|$ are the cardinalities of $\uI$ and $\uK$ respectively. It follows that for any given $\N$, the height function of the MI-poset is bounded by $0\leq h(x)\leq \N-1$. 

A poset $(\mathcal{P},\preceq)$ is said to satisfy the \textit{Jordan-Dedekind chain condition} (JDCC) if given any two elements $x,y \in \mathcal{P}$, all maximal chains with endpoints $x,y$ have the same length. It is easy to see that the MI-poset satisfies the JDCC for any $\N$, since given any two instances $\mi(\uI:\uK)\prec\mi(\uI':\uK')$, one obtains a maximal chain with these endpoints as a sequence of elements obtained by removing a single party at each step. Removing the parties in different orders will yield different maximal chains, but they all clearly have the same length.

A poset $(\mathcal{P},\preceq)$ is said to be \textit{graded} if there exists a function $g:(\mathcal{P},\preceq)\to\mathbb{N}$ such that
\begin{align}
\label{eq:garding}
    & x\succ y\quad \implies\quad g(x)> g(y)\nonumber\\
    & x \cov y\quad \implies\quad g(x)=g(y)+1 \,.
\end{align}
Since the MI-poset satisfies the JDCC, it follows that it is graded, and in particular that it is graded by its height function (see \cite{birkhoff1967lattice}).  

A particularly useful type of a poset, which will play a crucial role in what follows, is a \textit{lattice}, defined as:
\begin{defi}[Lattice]\label{def:lattice}
    A poset $(\mathcal{P},\preceq)$ is called a \emph{lattice} if for any two elements $x,y\in\mathcal{P}$, there exist both the least upper bound $x \vee y$ (called the \emph{join}) as well as the greatest lower bound $x \wedge y$ (called the \emph{meet}).
\end{defi}

One may then wonder if the MI-poset itself admits this additional structure. One can see that the MI-poset is \textit{not} a lattice for two different reasons. The simpler reason is that it is a finite poset, but it has neither a \emph{top} (an element which is greater than any other element) nor a \emph{bottom} (an element which is less than any other element). This implies that the maximal elements have no join, and that the minimal elements have no meet. However, this can be trivially fixed by formally adding these missing elements, and in the $\N=2$ case illustrated in \Cref{fig:MIPosetN2}, it is easy to check that by adding a formal top and bottom, one does obtain a lattice. However, for $\N\geq 3$, adding a top and a bottom is not sufficient to yield a lattice because of a deeper issue.

To see why this is the case, it is sufficient to consider the subset of the $\N=3$ MI-poset shown in \Cref{fig:MIPoset-nonlattice}. The elements $\mi(AB\!:\!CO)$ and $\mi(AC\!:\!BO)$ are common upper bounds to $\mi(A\!:\!O)$ and $\mi(B\!:\!C)$, but neither is the least upper bound as they are incomparable. Similarly, $\mi(A\!:\!O)$ and $\mi(B\!:\!C)$ are common lower bounds to $\mi(AB\!:\!CO)$ and $\mi(AC\!:\!BO)$, but neither is the greatest lower bound. This is sufficient to prove that the $\N=3$ MI-poset is not a lattice. Finally, notice that the subset in \Cref{fig:MIPoset-nonlattice} also appears as a subdiagram of the Hasse diagram of the MI-poset for any $\N>3$, and any new element that may appear cannot be less than both $\mi(AB\!:\!CO)$ and $\mi(AC\!:\!BO)$, since it involves new parties. This implies that the MI-poset cannot be turned into a lattice by adding top and bottom elements for any $\N\geq 3$. Nevertheless, while the MI-poset does not have the structure of lattice, in the next section we will see that one can build from it useful constructs which do have such a structure.

\begin{figure}[tb]
    \centering
    \begin{tikzpicture}
    \node[draw, rounded corners,fill=red!20!] (abco) at (-2,4) {{\footnotesize $\mi(AB\!:\!CO)$}};
    \node[draw, rounded corners,fill=red!20!] (acbo) at (2,4) {{\footnotesize $\mi(AC\!:\!BO)$}};
    \node[draw, rounded corners] (abo1) at (-6.3,2) {{\scriptsize $\mi(AB\!:\!O)$}};
    \node[draw, rounded corners] (aco1) at (-4.5,2) {{\scriptsize $\mi(AC\!:\!O)$}};
    \node[draw, rounded corners] (abo2) at (-2.7,2) {{\scriptsize $\mi(A\!:\!BO)$}};
    \node[draw, rounded corners] (aco2) at (-0.9,2) {{\scriptsize $\mi(A\!:\!CO)$}};
    \node[draw, rounded corners] (abc) at (0.9,2) {{\scriptsize $\mi(AB\!:\!C)$}};
    \node[draw, rounded corners] (boc) at (2.7,2) {{\scriptsize $\mi(BO\!:\!C)$}};
    \node[draw, rounded corners] (bac) at (4.5,2) {{\scriptsize $\mi(B\!:\!AC)$}};
    \node[draw, rounded corners] (bco) at (6.3,2) {{\scriptsize $\mi(B\!:\!CO)$}};
    \node[draw, rounded corners,fill=blue!20!] (ao) at (-2,0) {{\footnotesize $\mi(A\!:\!O)$}};
    \node[draw, rounded corners,fill=blue!20!] (bc) at (2,0) {{\footnotesize $\mi(B\!:\!C)$}};
    \draw[-] (abo1.south) -- (ao.north);
    \draw[-] (aco1.south) -- (ao.north);
    \draw[-] (abo2.south) -- (ao.north);
    \draw[-] (aco2.south) -- (ao.north);
    \draw[-] (abc.south) -- (bc.north);
    \draw[-] (boc.south) -- (bc.north);
    \draw[-] (bac.south) -- (bc.north);
    \draw[-] (bco.south) -- (bc.north);
    \draw[-] (abco.south) -- (abo1.north);
    \draw[-] (abco.south) -- (aco2.north);
    \draw[-] (abco.south) -- (abc.north);
    \draw[-] (abco.south) -- (bco.north);
    \draw[-] (acbo.south) -- (aco1.north);
    \draw[-] (acbo.south) -- (abo2.north);
    \draw[-] (acbo.south) -- (boc.north);
    \draw[-] (acbo.south) -- (bac.north);
    \end{tikzpicture}
    \caption{A subset of the $\N=3$ MI-poset, which shows that it is not a lattice. The elements in red are common upper bounds of the ones in blue, but neither is the least upper bound. Similarly, neither of the blue elements is the greatest lower bound of the red ones.}
    \label{fig:MIPoset-nonlattice}
\end{figure}
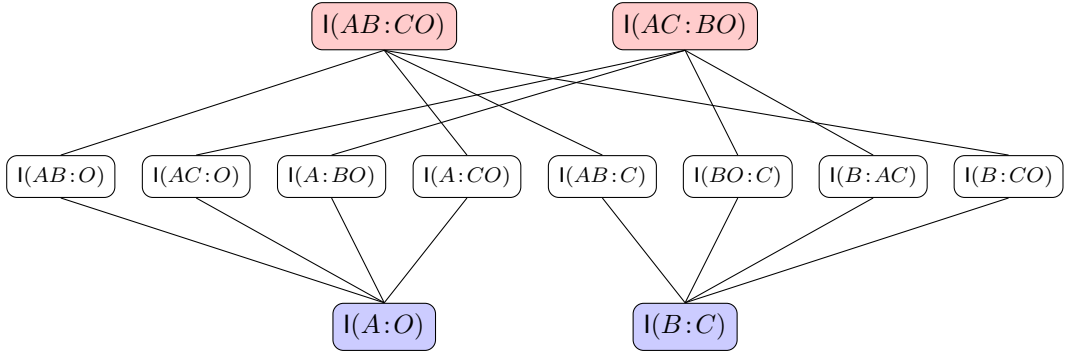

\section{The lattice of PMIs satisfying Klein's condition}
\label{sec:KC_PMIs}

In this section, we prove the main results of this work. We begin in \S\ref{ssec:proof_lattice} by describing the lattice $\lat{PMI}$ of all PMIs, building up from the set $\lat{SAC}$ of all the faces of the SAC$_{\N}$. This will allow us to show that the subset of KC-compatible PMIs, which we denote by $\lat{KC}$, is also a lattice. We then present $\lat{KC}$ for $\N=3$ in \S\ref{ssec:lattice3}, where we also discuss its structural properties in detail, and further comment on some of these properties for $\N=4$ in \S\ref{ssec:lattic4}. Using these observations, combined with a canonical construction of $\N'$-party PMIs from $\N$-party PMIs (with $\N'>\N$), we then prove in \S\ref{ssec:latticN} some general results about the structure of $\lat{KC}$ for an arbitrary number of parties. In particular, we prove a theorem about the structure of the $\zds$ of the KC-compatible extreme rays of the SAC$_{\N}$, which largely simplifies their explicit computation. Finally, in \S\ref{subsec:constructions}, we discuss generalizations of the canonical construction introduced in \S\ref{ssec:latticN}.

\subsection{The set of KC-compatible PMIs is a lattice}
\label{ssec:proof_lattice}

\subsubsection{SAC lattice \texorpdfstring{$\lat{SAC}$}{}}
\label{sss:SAClat}

We begin by explaining that the set $\lat{SAC}^\N$ of all the faces of the SAC$_{\N}$ for any fixed $\N$ is a lattice, and in particular we will describe the join and meet operations. While we will focus on the SAC, it is true that the set of faces of any polyhedral cone forms a lattice, and we will mainly follow \cite{MR1226888,z-lop-95}. 

Given $\lat{SAC}^\N$, we introduce the following partial order on this set: 
\begin{equation}
\label{eq:face_order}
    \face\preceq \face'\quad \Longleftrightarrow \quad \face\supseteq \face' \,.
\end{equation}
Notice that this is reverse inclusion, rather than inclusion.\footnote{\, Both inclusion and reverse inclusion are common in the literature. We prefer reverse inclusion for consistency with \cite{Hernandez-Cuenca:2019jpv} and because it allows for some of the results in later sections to take a more natural form.} To prove that this poset is a lattice, it will be convenient to use the following result from lattice theory.\footnote{\, Of course, the \textit{dual} of \Cref{thm:lattice_condition}, where ``bottom'' is replaced by ``top'' and ``join'' is replaced by ``meet'' also holds. We will use this dual version later in this subsection.} 

\begin{thm}
\label{thm:lattice_condition}
    A finite poset $(\mathcal{P},\preceq)$ is a lattice if and only if it has a bottom, and the join $x\vee y$ exists for every pair of elements $x,y\in\mathcal{P}$.
\end{thm}
\begin{proof}
   See for example Theorem 2.31 of \cite{davey1990introduction}.
\end{proof}

This theorem immediately implies that the set $\lat{SAC}$, with partial order \eqref{eq:face_order}, is a lattice. The bottom is the codimension-$0$ face that corresponds to the full cone, and the join is given by 
\begin{equation}
\label{eq:join_faces}
    \face\vee \face' = \face''\qquad \text{with}\qquad \face''=\face\cap \face' \,.
\end{equation}
To see that this is indeed a join, first notice that the face $\face''=\face\cap \face'$ is a common upper bound of $\face$ and $\face'$, since $\face'' \subseteq \face$ and $\face'' \subseteq \face'$, and so $\face\preceq\face''$ and $\face'\preceq\face''$. Furthermore, it is the least upper bound, since $\face \cap \face'$ is the largest dimensional face contained in both $\face$ and $\face'$, so any other face $\face'''$ contained in both necessarily obeys $\face''' \subseteq \face''$, or equivalently $\face''' \succeq \face''$. 

In what follows, it will be also useful to consider a representation of the faces, which is particularly convenient to describe the join. Since we are dealing with polyhedral cones, each face $\face$ can be represented by the set $\widehat{\face}$ of extreme rays that generate it,\footnote{\, Any face is the conical hull, i.e., a linear combination with non-negative coefficients, of the  extreme rays on its boundary.} and in terms of these sets, \eqref{eq:join_faces} can then be rewritten as
\begin{equation}
\label{eq:join_face_er}
    \face\vee \face' = \face''\qquad \text{with}\qquad \widehat{\face}''=\widehat{\face}\cap \widehat{\face}' \,.
\end{equation}
From now on, by $\lat{SAC}^\N$ we mean not just the \textit{set} of all faces of the SAC$_\N$, but the \textit{lattice} given by the partial order \eqref{eq:face_order} on such a set. 

While the existence of the join suffices to prove that $\lat{SAC}$ is a lattice by \Cref{thm:lattice_condition}, for our purposes it will be useful to also determine the meet operation explicitly. This is most easily done if we introduce the following notation (for more details see \S4.1 of \cite{MR1226888}). To every vector $\vec\ent$ in the SAC$_{\N}$ we can associate an element of $\{0,+\}^{\E}$ (with $\E$ the total number of MI instances given by \eqref{eq:e_def}), called a \textit{sign vector}, such that
\begin{align}
\label{eq:signed_vec}
    \sigma:\quad & \text{SAC}_{\N}\; \to\; \{0,+\}^{\E} \quad \nonumber\\
            & \vec\ent\; \mapsto\; \sigma(\vec\ent)=\{\,\text{sign}\,(\,\vec\ent\cdot\vec\mi(\uI^e:\uK^e)\,),\,\forall e\in [\E]\,\} \,,
\end{align}
where $[\E]={1,\ldots,\E}$, and $\vec\mi(\uI^e:\uK^e)$ is the vector of coefficients of the MI instance $\mi(\uI^e:\uK^e)$ introduced in the proof of \Cref{lem:facet}, or equivalently the vector normal to the hyperplane \eqref{eq:mihyp} and directed towards the positive halfspace \eqref{eq:mi_sa}.\footnote{\, Obviously we do not obtain negative signs because we are restricting the map $\sigma$ to the set of vectors that belong to the SAC$_{\N}$.
Also note that even though not all instances of MI in $\mgs$ correspond to facets of the SAC$_{\N}$, cf.\ \Cref{lem:facet}, to avoid introducing additional notation, we do not remove the redundant ones from the MI-poset. It should be clear that this choice has no effect on the following discussion and only affects the explicit form of the sign vectors.} 
Two vectors in the SAC$_{\N}$ correspond to the same sign vector if and only if they belong to the \textit{interior} of the same face. We can therefore represent each face $\face$ of the SAC$_{\N}$ by the sign vector $\sigma(\vec\ent)$ of any vector $\vec\ent$ in the interior of $\face$, motivating us to define the sign vector \textit{of a face} as
\begin{equation}
\label{eq:signed_vec_face}
    \vece{\face}=\sigma(\vec\ent)\quad \text{for any}\quad \vec\ent\in\text{int}\,(\face) \, .
\end{equation}

Given the sign vectors $\vece\face,\vece\face'$ of two faces $\face,\face'$, we can then define the following (commutative) composition:
\begin{equation}
\label{eq:composition}
    (\vece\face\circ \vece\face')_e=
    \begin{cases}
    0 \quad &\text{if}\;\; \vece\face_e=\vece\face'_e= 0 \\
    + \quad &\text{otherwise}
    \end{cases}\, ,
\end{equation}
where $\vece\face_e$ is the $e$-th component of $\vece\face$. We can interpret this composition geometrically as follows: given two sign vectors $\vece\face,\vece\face'$, the sign vector in \eqref{eq:composition} is that of a new entropy vector $\vec\ent''\in\text{int}\,(\face'')$ obtained by averaging $\vec\ent$ and $\vec\ent'$. Using sign vectors, we can now conveniently express the meet operation in $\lat{SAC}$ as follows:
\begin{equation}
\label{eq:meet_face}
    \face\wedge\face'=\face''\qquad \text{with}\qquad \vece\face''=\vece\face\circ\vece\face' \,.
\end{equation}
The fact that this operation is indeed the meet follows from the observation that in the entropy space, the face $\face''$ obtained via this construction is the face of lowest dimension that contains both $\face$ and $\face'$ on its boundary. 

Whereas the representation in terms of the conical hull of extreme rays was useful in describing the join \eqref{eq:join_face_er}, it is no longer convenient for describing the meet. In particular, the meet of two faces is not in general the union of the sets of extreme rays that generate them, since the resulting face may include additional extreme rays.
Conversely, the description of a face by its sign vector is not convenient to express the join. In particular, the sign vector of the join of two faces cannot in general be obtained by \textit{only} assigning a zero component to $\vece\face''$ whenever $\vece\face$ or $\vece\face'$ have a zero component, since the face resulting from the join can belong to additional hyperplanes \eqref{eq:mihyp}, and therefore have additional zero components. A simple example of these operations and the role of different representations for the faces is shown in \Cref{fig:cartoon_lattice} for a non-simplicial cone in $\mathbb{R}^3$ (we warn the reader that \Cref{fig:cartoon_lattice} is included only for the purpose of illustrating the join and meet operations in the lattice of faces of a generic polyhedral cone, and that the figure has no relation to the SAC$_{\N}$ for any $\N$).

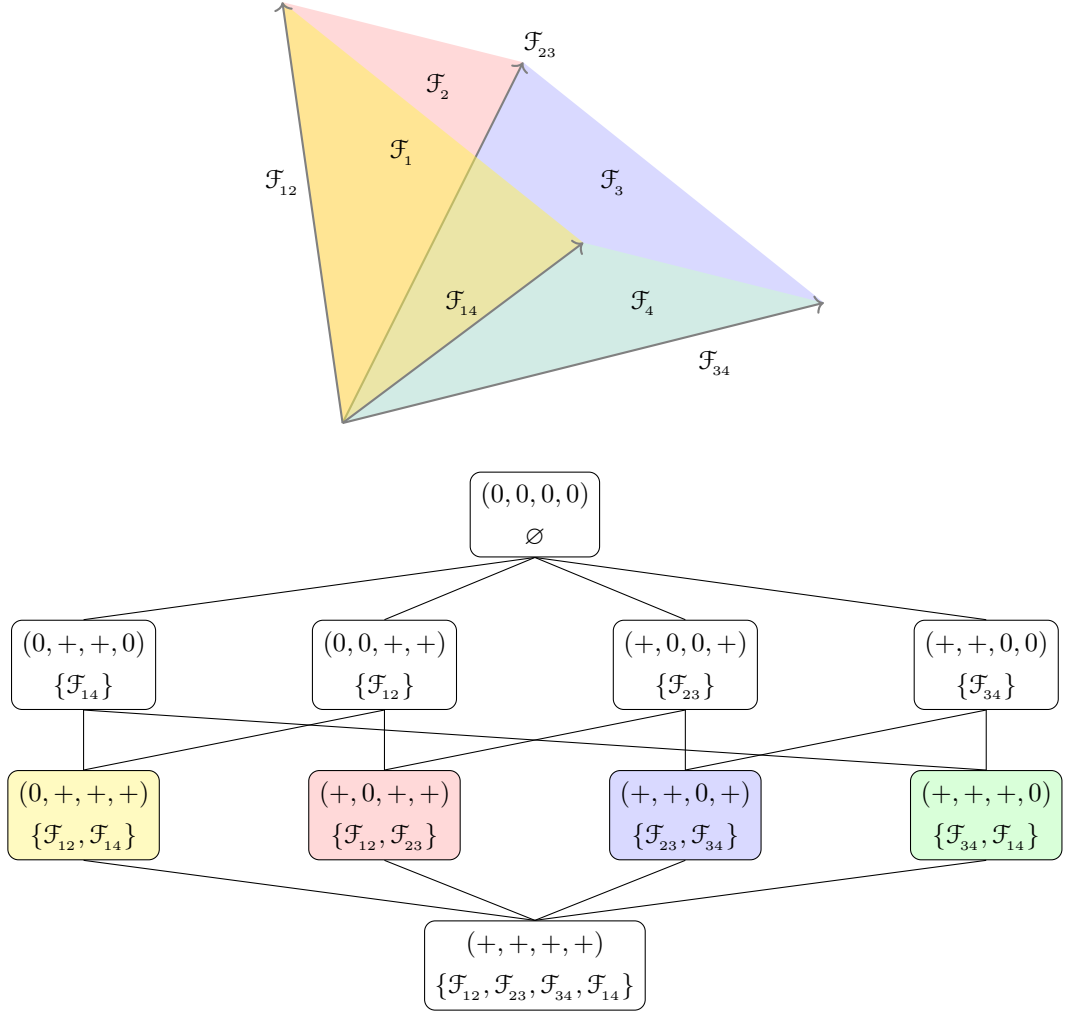
\begin{figure}[tbp]
    \centering
    \begin{subfigure}{0.9\textwidth}
    \centering
    \begin{tikzpicture}[scale=0.8]
    \fill[fill=blue!15!] (0,0) -- (8,2) -- (3,6) ;
    \fill[fill=red!15!] (0,0) -- (3,6) -- (-1,7) ;
    \draw[->, gray, thick] (0,0) -- (3,6) ;
    \fill[fill=yellow!80, opacity=0.5] (0,0) -- (-1,7) -- (4,3) ;
    \fill[fill=green!20, opacity=0.5] (0,0) -- (8,2) -- (4,3) ;
    \draw[->, gray, thick] (0,0) -- (4,3) ;
    \draw[->, gray, thick] (0,0) -- (-1,7) ;
    \draw[->, gray, thick] (0,0) -- (8,2) ;
    \node (f1) at (1,4.5) {{\footnotesize $\face_{_{\!1}}$}}; 
    \node (f2) at (1.6,5.6) {{\footnotesize $\face_{_{\!2}}$}}; 
    \node (f3) at (4.5,4) {{\footnotesize $\face_{_{\!3}}$}}; 
    \node (f4) at (5,2) {{\footnotesize $\face_{_{\!4}}$}};
    \node (f12) at (-1,4) {{\footnotesize $\face_{_{\!12}}$}}; 
    \node (f23) at (3.3,6.3) {{\footnotesize $\face_{_{\!23}}$}}; 
    \node (f34) at (6.2,1) {{\footnotesize $\face_{_{\!34}}$}}; 
    \node (f14) at (2,2) {{\footnotesize $\face_{_{\!14}}$}}; 
    \end{tikzpicture}
    \end{subfigure}
    
    \vspace{0.6cm}
    \begin{subfigure}{\textwidth}
    \centering
    \begin{tikzpicture}
    \node[draw, align=center, rounded corners] (point) at (0,2) {{\footnotesize $(0,0,0,0)$}\\{\footnotesize $\varnothing$}};
    \node[draw, align=center, rounded corners] (line1) at (-6,0) {{\footnotesize $(0,+,+,0)$}\\{\footnotesize $\{\face_{_{\!14}}\}$}};
    \node[draw, align=center, rounded corners] (line2) at (-2,0) {{\footnotesize $(0,0,+,+)$}\\{\footnotesize $\{\face_{_{\!12}}\}$}};
    \node[draw, align=center, rounded corners] (line3) at (2,0) {{\footnotesize $(+,0,0,+)$}\\{\footnotesize $\{\face_{_{\!23}}\}$}};
    \node[draw, align=center, rounded corners] (line4) at (6,0) {{\footnotesize $(+,+,0,0)$}\\{\footnotesize $\{\face_{_{\!34}}\}$}};
    \node[draw, align=center, rounded corners, fill=yellow!30!] (face1) at (-6,-2) {{\footnotesize $(0,+,+,+)$}\\{\footnotesize $ \{\face_{_{\!12}},\face_{_{\!14}}\}$}};
    \node[draw, align=center, rounded corners, fill=red!15!] (face2) at (-2,-2) {{\footnotesize $(+,0,+,+)$}\\{\footnotesize $ \{\face_{_{\!12}},\face_{_{\!23}}\}$}};
    \node[draw, align=center, rounded corners, fill=blue!15!] (face3) at (2,-2) {{\footnotesize $(+,+,0,+)$}\\{\footnotesize $ \{\face_{_{\!23}},\face_{_{\!34}}\}$}};
    \node[draw, align=center, rounded corners, fill=green!15!] (face4) at (6,-2) {{\footnotesize $(+,+,+,0)$}\\{\footnotesize $ \{\face_{_{\!34}},\face_{_{\!14}}\}$}};
    \node[draw, align=center, rounded corners] (cone) at (0,-4) {{\footnotesize $(+,+,+,+)$}\\{\footnotesize $\{\face_{_{\!12}},\face_{_{\!23}},\face_{_{\!34}},\face_{_{\!14}}\}$}};
    \draw[-] (point.south) -- (line1.north);
    \draw[-] (point.south) -- (line2.north);
    \draw[-] (point.south) -- (line3.north);
    \draw[-] (point.south) -- (line4.north);
    \draw[-] (line1.south) -- (face1.north);
    \draw[-] (line1.south) -- (face4.north);
    \draw[-] (line2.south) -- (face1.north);
    \draw[-] (line2.south) -- (face2.north);
    \draw[-] (line3.south) -- (face2.north);
    \draw[-] (line3.south) -- (face3.north);
    \draw[-] (line4.south) -- (face3.north);
    \draw[-] (line4.south) -- (face4.north);
    \draw[-] (face1.south) -- (cone.north);
    \draw[-] (face2.south) -- (cone.north);
    \draw[-] (face3.south) -- (cone.north);
    \draw[-] (face4.south) -- (cone.north);
    \end{tikzpicture}
    \vspace{0.4cm}
    \end{subfigure}
    \caption{An example of a $3$-dimensional non-simplicial polyhedral cone and its lattice of faces, ordered by reverse inclusion. We give two equivalent descriptions for each face: the sign vector from \eqref{eq:signed_vec_face}, and the set of extreme rays that generate it. Since the join of two faces is their intersection (cf.\ \eqref{eq:join_faces}), it is simply given by the intersection of sets of extreme rays that generate them (cf.\ \eqref{eq:join_face_er}).
    However, in terms of sign vectors, the join is not straightforward to compute, as it is \emph{not} the ``union'' of the sets of zero components (a component of the sign vector of the join of two faces can be zero even if the corresponding components for both faces are not). An example demonstrating this is $\face_{_{\!1}}\vee\face_{_{\!3}}$. 
    On the other hand, the representation in terms of sign vectors is convenient for computing the meet of two faces (cf.\ \eqref{eq:meet_face}), which is given by the ``intersection'' of the sets of zero components (a component is zero if and only if the corresponding components for both faces are zero). However, the meet is not given by the union of the sets of extreme rays that generate the faces, since the resulting face may contain additional extreme rays. An example demonstrating this is $\face_{_{\!1}}\wedge\face_{_{\!2}}$.}
    \label{fig:cartoon_lattice}
\end{figure}

\subsubsection{PMI lattice \texorpdfstring{$\lat{PMI}$}{}}
\label{sss:PMIlat}

Having described the join and meet operations in $\lat{SAC}$, we can now discuss the analogous structure in $\lat{PMI}$, the set of all PMIs. As we explained in \Cref{subsec:sac_pmi}, for any given $\N$ there is a bijection between the \emph{set} of faces of the SAC$_{\N}$ and the \emph{set} of PMIs that is given by the map $\mu$ defined in \eqref{eq:lambda} and \eqref{eq:lambdainv}. Starting from the set $\lat{PMI}$, we obtain a \emph{poset} by introducing the following partial order (again we use reverse inclusion, cf.\ \eqref{eq:face_order})
\begin{equation}
\label{eq:pmi_order}
    \pmi\preceq \pmi'\quad \Longleftrightarrow \quad \pmi\supseteq \pmi'\,.
\end{equation}
Since we have endowed both $\lat{SAC}$ and $\lat{PMI}$ with a partial order, we can then ask how the two posets are related, and intuitively it should already be clear that structurally they are identical. This intuition is confirmed by \Cref{lem:one-to-one}, where we show that the map $\mu$ is not just a bijective map between the two sets, but an \textit{order isomorphism}. Given two posets $(\mathcal{P},\preceq_{_{\mathcal{P}}})$ and $(\mathcal{Q},\preceq_{_{\mathcal{Q}}})$, an \emph{order isomorphism} is defined as a bijective map $\phi$ from $\mathcal{P}$ to $\mathcal{Q}$ such that for every $x,y \in \mathcal{P}$,
\begin{equation}
    x\preceq_{_{\mathcal{P}}} y\quad \Longleftrightarrow\quad \phi(x)\preceq_{_{\mathcal{Q}}} \phi(y)\,.
\end{equation}
In other words, order isomorphisms preserve partial order, and this is the reason why we can use the same symbol $\preceq$ for the partial order in both $\lat{SAC}$ and $\lat{PMI}$. It should be clear from context which one we mean.

\begin{lemma}\label{lem:one-to-one}
    The bijective map $\mu$ given in \eqref{eq:lambda} and \eqref{eq:lambdainv} is an order isomorphism between the posets $\lat{\emph{SAC}}$ and $\lat{\emph{PMI}}$, with partial order given respectively in \eqref{eq:face_order} and \eqref{eq:pmi_order}.
\end{lemma}
\begin{proof}
   For the purpose of this proof, we will need to distinguish the partial orders in $\lat{SAC}$ and $\lat{PMI}$, and we respectively denote them as $\preceq_{_{\face}}$ and $\preceq_{_{\pmi}}$. For any two faces $\face,\face'$, we have
   \begin{equation}
       \face\preceq_{_{\face}}\face'\quad \implies\quad \text{span}(\face)\supseteq\text{span}(\face')\quad \implies\quad \mu(\face)\preceq_{_{\pmi}}\mu(\face')\,.
   \end{equation}
   Vice versa, for any two PMIs $\pmi,\pmi'$ we have
   \begin{equation}
       \pmi\preceq_{_{\pmi}}\pmi'\quad \implies\quad \pmi\cap\text{SAC}_{\N}\supseteq\pmi'\cap\text{SAC}_{\N}\quad \implies\quad \mu^{-1}(\pmi)\preceq_{_{\face}}\mu^{-1}(\pmi')\,.
   \end{equation}
   This completes the proof that $\mu$ is an order isomorphism.
\end{proof}

The fact that there exists an order isomorphism between the two posets $\lat{SAC}$ and $\lat{PMI}$ implies that they have the same order structure. Thus, since $\lat{SAC}$ is a lattice, $\lat{PMI}$ is also a lattice, and moreover the join and meet of any two elements in $\lat{PMI}$ correspond to those in $\lat{SAC}$.

From an algebraic point of view, this means that the map $\mu$ respects meet and join, i.e., it is a lattice homomorphism between $\lat{SAC}$ and $\lat{PMI}$, and in fact is a lattice isomorphism. Explicitly, we have for all $\face,\face'$
\begin{align}
\begin{split}
    & \mu(\face\wedge\face')=\mu(\face)\wedge\mu(\face')\\
    & \mu(\face\vee\face')=\mu(\face)\vee\mu(\face')\,,
\end{split}
\end{align}
and since $\mu$ is also an isomorphism, we have for all $\pmi,\pmi'$
\begin{align}
\begin{split}
    & \mu^{-1}(\pmi\wedge\pmi')=\mu^{-1}(\pmi)\wedge\mu^{-1}(\pmi')\\
    & \mu^{-1}(\pmi\vee\pmi')=\mu^{-1}(\pmi)\vee\mu^{-1}(\pmi')\,.
\end{split}
\end{align}
Using these relations, we can then compute the meet and join for PMIs in terms of meet and join for faces of the SAC via
\begin{align}
    & \pmi\wedge\pmi'=\mu\left[\mu^{-1}(\pmi)\wedge\mu^{-1}(\pmi')\right] \label{eq:KCmeet} \\
    & \pmi\vee\pmi'=\mu\left[\mu^{-1}(\pmi)\vee\mu^{-1}(\pmi')\right]\label{eq:KCjoin}\,.
\end{align}

Let us briefly comment on these expressions, starting with the meet \eqref{eq:KCmeet}. Consider two PMIs $\pmi_1$ and $\pmi_2$ written in terms of their sets of vanishing MI instances $\mgs^0_1$ and $\mgs^0_2$.\footnote{\, Observe that when we write the PMIs in terms of their sets of vanishing MI instances, reverse inclusion of PMIs corresponds to inclusion of these sets, i.e., both $\pmi_1 \supseteq \pmi_2$ and (equivalently) $\mgs^0_1 \subseteq \mgs^0_2$ correspond to $\pmi_1\preceq\pmi_2$. \label{fn:revincl}}
The sign vectors $\vece\face_1$ and $\vece\face_2$ of the faces $\face_1=\mu^{-1}(\pmi_1)$ and $\face_2=\mu^{-1}(\pmi_2)$ are then given by %
\begin{equation}
    (\vece\face_i)_e=\begin{cases}
    0\quad & \text{if}\;\;\mi(\uI_e:\uK_e)\in\mgs^0_i\\
    +\quad & \text{otherwise}
    \end{cases},
    \qquad e\in\,[\E]\,,\;\; i \in\{1,2\}\,.
\end{equation}
Using \eqref{eq:composition} and \eqref{eq:meet_face} we can write the sign vector of $\face_1\wedge\face_2$ as
\begin{equation}
\label{eq:meetmgs0}
    (\vece\face_1\circ\vece\face_2)_e=\begin{cases}
    0\quad & \text{if}\;\;\mi(\uI_e:\uK_e)\in\mgs^0_1\cap\mgs^0_2\\
    +\quad & \text{otherwise}
    \end{cases},
    \qquad e\in\,[\E]\,.
\end{equation}
From \eqref{eq:lambda} and \eqref{eq:KCmeet}, we have 
\begin{equation}
    \pmi_1\wedge\pmi_2=\text{span}\,(\face_1\wedge\face_2)\,,
\end{equation}
and from \eqref{eq:meetmgs0} it follows that the span of the face $\face_1\wedge\face_2$ is the subspace resulting from the intersection of the hyperplanes \eqref{eq:mihyp} associated to the MI instances in $\mgs^0_1\cap\mgs^0_2$. We can then write the meet of two PMIs conveniently as
\begin{equation}\label{eq:PMImeet}
\mgs^0(\pmi_1\wedge\pmi_2) = \mgs_1^0(\pmi_1)\cap\mgs_2^0(\pmi_2)\,,
\end{equation}
where by $\mgs^0(\pmi)$ we denote the set of vanishing MI instances of $\pmi$.

On the other hand, to determine the join between two PMIs, it is convenient to express the faces $\face_1$ and $\face_2$ in terms of their extreme rays, obtaining from \eqref{eq:join_face_er}
\begin{equation}
\label{eq:PMIjoin}
    \pmi_1\vee\pmi_2=\text{span}\,(\face_1\vee\face_2)=\text{span}\,(\widehat{\face}_1\cap\widehat{\face}_2)\,.
\end{equation}
Notice in particular that unlike the case of the faces, the join of two PMIs is \textit{not} necessarily equal to their intersection,\footnote{\, For example, the intersection of two PMIs could be a linear subspace that is not the span of a face of the SAC, and therefore not a PMI. We could nevertheless geometrically describe the join of two PMIs as the span of their intersection with the SAC.} and in general we only have
\begin{equation}\label{eq:join-contained}
    \pmi_1\vee\pmi_2\subseteq\pmi_1\cap\pmi_2\,.
\end{equation}

Having described the structure of the lattices $\lat{SAC}$ and $\lat{PMI}$ and their relationship in detail, we now prove the main result of this subsection.

\subsubsection{KC lattice \texorpdfstring{$\lat{KC}$}{}}
\label{sss:KClat}

The subset $\lat{KC}$ of PMIs that satisfy KC can naturally be seen as a poset with partial order induced by $\lat{PMI}$. However, it is not immediately obvious that this set has any additional structure. The reason is that in general a subset of a lattice need not be a lattice, since the meet or the join for some pair of elements may not exist. Furthermore, even when it is a lattice, it is not necessarily a \emph{sublattice}\footnote{\, A sublattice is a subset of a lattice $\lat{}$ which is also a lattice, with the same meet and join as the original ones in $\lat{}$.} since the meet and join operations may change, and several other structural properties may differ substantially from those of the parent lattice. We begin by proving that $\lat{KC}$ is in fact a lattice for an arbitrary number of parties, and we will devote most of the rest of this work to analyzing its structure in detail.

\begin{thm}
\label{thm:KC-PMIlattice}
    For any given number of parties $\N$, the subset $\lat{\emph{KC}}^\N \subseteq\lat{\emph{PMI}}^\N$ of \emph{KC}-compatible \emph{PMIs}, with partial order induced from $\lat{\emph{PMI}}^\N$, is a lattice. Furthermore, the meet in $\lat{\emph{KC}}^\N$ is the same as in $\lat{\emph{PMI}}^\N$, and is given by \eqref{eq:KCmeet}, or equivalently \eqref{eq:PMImeet}.
\end{thm}
\begin{proof}
    For a given $\N$, consider the poset $\lat{KC}^\N\subseteq\lat{PMI}^\N$ of PMIs that are KC-compatible, with partial order induced from $\lat{PMI}^\N$. Notice that the top element of $\lat{KC}^\N$ is the trivial subspace that corresponds to the origin of the SAC$_\N$, and obviously satisfies KC since its $\zds$ is the entire MI-poset, which is clearly a down-set. Hence, since $\lat{KC}^\N$ is a finite poset, it suffices by (the order dual of) \Cref{thm:lattice_condition} to prove that it has a meet, and that it is the same as the one in $\lat{PMI}^\N$. 
    Consider two elements $\pmi_1,\pmi_2\in\lat{PMI}^\N$, with respective $\zdss$ $\mgs_1^0$ and $\mgs_2^0$. If both PMIs are in $\lat{KC}^\N$, by \Cref{lem:KCdownset} both $\mgs_1^0$ and $\mgs^0_2$ are down-sets in the MI-poset. The meet $\pmi_1\wedge\pmi_2$ in $\lat{PMI}^\N$ is given by \eqref{eq:PMImeet}, and since the intersection of two down-sets is a down-set, $\pmi_1\wedge\pmi_2$ is also an element of $\lat{KC}^\N$. This proves that the meet in $\lat{KC}^\N$ is precisely the meet in $\lat{PMI}^\N$ restricted to $\lat{KC}^\N$, completing the proof.
\end{proof}

In light of \Cref{thm:KC-PMIlattice}, we now want to determine if, like the meet, the join in $\lat{KC}$ is also the same operation as in $\lat{PMI}$, in which case $\lat{KC}$ would be a sublattice of $\lat{PMI}$. Since $\lat{KC}$ is finite, the join of two elements $\pmi,\pmi'$ can formally be written using the meet as \cite{davey1990introduction}
\begin{equation}
    \label{eq:KC_join}
    \pmi \kcjoin \pmi'=\bigwedge_{\pmi'' \in \{\pmi,\pmi'\}^u}\pmi''\,,
\end{equation}
where
\begin{equation}
\label{eq:common_up}
    \{\pmi,\pmi'\}^u=\{\pmi'' \in \lat{KC}: \;\pmi''\succeq\pmi\;\;\text{and}\;\;\pmi''\succeq\pmi'\}
\end{equation}
is the set of common upper bounds to $\pmi$ and $\pmi'$ in $\lat{KC}$. Notice that we used a new symbol ($\kcjoin$) for the join, since in principle this can be a different operation from the one in $\lat{PMI}$. On the other hand, because of \Cref{thm:KC-PMIlattice}, we denote the meet in $\lat{KC}$ by the same symbol that we used for the meet in $\lat{PMI}$. 

To understand \eqref{eq:KC_join} geometrically, let us first discuss the case of the analogous formula for $\lat{SAC}$. In that case, the join (least upper bound) is the largest dimensional face on the boundary of both $\face$ and $\face'$ (because of reverse inclusion), which is simply their intersection. Furthermore, the set $\{ \face, \face' \}^u$ corresponds to the set of all faces on the boundary of both $\face$ and $\face'$, and the meet (greatest lower bound) of all such faces is then the smallest dimensional face containing them (again because of reverse inclusion), which is simply the composition of all the faces in $\{ \face, \face' \}^u$. This shows why the version of \eqref{eq:KC_join} for $\lat{SAC}$ holds.

With this intuition from $\lat{SAC}$, we now want to understand what changes in the case of $\lat{KC}$. By the correspondence established in \Cref{lem:one-to-one} between faces of the SAC and PMIs, rather than working with KC-compatible PMIs, it will be convenient to use the corresponding ``KC-compatible faces'' in the SAC. We denote the lattice of these faces by $\lat{$\face$KC}$ (which is obviously isomorphic to $\lat{KC}$). The main difference now is that the join of KC-compatible faces is not necessarily their intersection, but might be a face \textit{contained} in such an intersection. 

Consider two KC-compatible faces $\face$ and $\face'$ which are incomparable.\footnote{\, If they are comparable, the join is just the greater 
of the two faces, and \eqref{eq:KC_join} trivializes.} Their join $\face^*=\face\vee\face'$ in $\lat{SAC}$ is a face $\face^*$ that may or may not be KC-compatible. If $\face^*$ is KC-compatible, then the join of $\face$ and $\face'$ in $\lat{$\face$KC}$ coincides with $\face^*$, and we just explained why \eqref{eq:KC_join} holds. Hence, suppose that $\face^*$ is not KC-compatible. We denote by $d$ the dimension of $\face^*$, and by $\partial_{(\kappa)}\face^*$ the set of KC-compatible faces on the boundary of $\face^*$ with dimension $d-\kappa$. Notice that the union of all such sets for $\kappa=1,\ldots,d$\, is precisely the set defined in the version of \eqref{eq:common_up} for KC-compatible faces. Let $\kappa^*$ be the smallest value of $\kappa$ such that $\partial_{(\kappa)}\face^*$ is non-empty. Notice that since the top of $\lat{$\face$KC}$ is the origin, which is trivially KC-compatible, $\kappa^*$ is guaranteed to exist. Suppose now that $\partial_{(\kappa^*)}\face^*$ contains more than one element. Since the  meet in $\lat{SAC}$ and $\lat{$\face$KC}$ are the same, the meet of these elements in $\lat{SAC}$ is then a KC-compatible face $\widetilde{\face^*}$ of greater dimension, which must also be on the boundary of $\face^*$.\footnote{\, It cannot be $\face^*$ since $\face^*$ is not KC-compatible by assumption, and it cannot be a larger face that contains $\face^*$ because the meet is the composition of faces (cf.\ \eqref{eq:meet_face}).} This contradicts the assumption that $\kappa^*$ is minimal. It follows that $\partial_{(\kappa^*)}\face^*$ must only contain a single element, and this is precisely the largest dimensional KC-compatible face contained in $\face^*$. Hence, because of reverse inclusion, the single element in $\partial_{(\kappa^*)}\face^*$ is the join of $\face$ and $\face'$ in $\lat{$\face$KC}$, which is precisely the l.h.s.\ of \eqref{eq:KC_join}. Notice that $\partial_{(\kappa^*)}\face^*$ is also the composition of all the KC-compatible faces on the boundary of $\face$ and $\face'$, which is exactly the r.h.s.\ of the version of \eqref{eq:KC_join} for $\lat{$\face$KC}$.

In the above argument we considered the possibility that the PMI of the face $\face^*$ is not in $\lat{KC}$. If that is the case, and there exist two KC-compatible PMIs $\pmi$ and $\pmi'$ such that $\pmi \vee \pmi' = \mu(\face^*)$, then the lattice $\lat{KC}$ is \textit{not} a sublattice of $\lat{PMI}$, since $\pmi \kcjoin \pmi'$ would be a different PMI. As we will see in \Cref{ssec:latticN}, this occurs whenever $\N\geq 4$. This also implies that if a face corresponds to a PMI in $\lat{KC}$, in general it is not the case that all the faces on its boundary necessarily correspond to PMIs in $\lat{KC}$. Similarly, it is clear that if the PMI of a face is not in $\lat{KC}$, the PMIs corresponding to some of the faces on its boundary can still be in $\lat{KC}$ (e.g., the origin is in $\lat{KC}$ and is on the boundary of all the faces). It is in principle even possible for a face to correspond to a PMI in $\lat{KC}$ and yet have none of its facets correspond to a PMI in $\lat{KC}$. As we will see, this indeed happens for $\N\geq 4$, and it implies that $\lat{KC}$ does not satisfy the JDCC for any $\N \geq 4$.

We conclude this subsection with a comment about the relation between $\lat{KC}$ and the MI-poset, further motivating the detailed analysis that we will carry out in later subsections. A \textit{distributive lattice} is a lattice $\lat{}$ with a particularly nice structure, such that for all elements $x,y,z \in \lat{}$, the distributive laws for the meet and join hold:
\begin{align}
\begin{split}
    x \wedge (y \vee z) &= (x \wedge y) \vee (x \wedge z)\\
    x \vee (y \wedge z) &= (x \vee y) \wedge (x \vee z)\,.
\end{split}
\end{align}
For any poset $(\mathcal{P},\preceq)$, the set of all its down-sets, with partial order given by inclusion, is a distributive lattice, with the meet and join corresponding respectively to the intersection and union. Since KC can be conveniently formulated in terms of down-sets of the MI-poset, we define the following lattice.
\begin{defi}[Mutual information down-set (MID) lattice] \label{def:mid_lattice}
    For a given number of parties $\N$, the 
    \emph{mutual information down-set lattice} $\lat{\emph{MID}}^{\N}$ is the of lattice of down-sets of the $\N$-party \emph{MI}-poset.
\end{defi}

\begin{figure}[t]
    \centering
    \begin{tikzpicture}[scale=0.9]
    \node[draw,rounded corners, fill=blue!20!] (11) at (0,4) {{\footnotesize $\{\,\mi(A\!:\!BO), \mi(B\!:\!AO), \mi(O\!:\!AB)\,\}$}};
    \node[draw,rounded corners] (21) at (-5,2) {{\footnotesize $\{\,\mi(A\!:\!BO), \mi(B\!:\!AO)\,\}$}};
    \node[draw,rounded corners] (22) at (0,2) {{\footnotesize $\{\,\mi(A\!:\!BO), \mi(O\!:\!AB)\,\}$}};
    \node[draw,rounded corners] (23) at (5,2) {{\footnotesize $\{\,\mi(B\!:\!AO), \mi(O\!:\!AB)\,\}$}};
    \node[draw,rounded corners] (31) at (-5,0) {{\footnotesize $\{\,\mi(A\!:\!BO),\mi(B\!:\!O)\,\}$}};
    \node[draw,rounded corners] (32) at (0,0) {{\footnotesize $\{\,\mi(B\!:\!AO), \mi(A\!:\!O)\,\}$}}; 
    \node[draw,rounded corners] (33) at (5,0) {{\footnotesize $\{\,\mi(O\!:\!AB), \mi(A\!:\!B)\,\}$}};
    \node[draw,rounded corners] (41) at (-6,-2) {{\footnotesize $\{\,\mi(A\!:\!B), \mi(A\!:\!O), \mi(B\!:\!O)\,\}$}};
    \node[draw,rounded corners, fill=blue!20!] (42) at (-1,-2) {{\footnotesize $\{\,\mi(A\!:\!BO)\,\}$}};
    \node[draw,rounded corners, fill=blue!20!] (43) at (2.5,-2) {{\footnotesize $\{\mi(B\!:\!AO)\,\}$}};
    \node[draw,rounded corners, fill=blue!20!] (44) at (6,-2) {{\footnotesize $\{\,\mi(O\!:\!AB)\,\}$}};
    \node[draw,rounded corners] (51) at (-5,-4) {{\footnotesize $\{\,\mi(A\!:\!B), \mi(A\!:\!O)\,\}$}};
    \node[draw,rounded corners] (52) at (0,-4) {{\footnotesize $\{\,\mi(A\!:\!B), \mi(B\!:\!O)\,\}$}};
    \node[draw,rounded corners] (53) at (5,-4) {{\footnotesize $\{\,\mi(A\!:\!O), (B\!:\!O)\,\}$}};
    \node[draw,rounded corners, fill=blue!20!] (61) at (-5,-6) {{\footnotesize $\{\,\mi(A\!:\!B)\,\}$}};
    \node[draw,rounded corners, fill=blue!20!] (62) at (0,-6) {{\footnotesize $\{\,\mi(A\!:\!O)\,\}$}};
    \node[draw,rounded corners, fill=blue!20!] (63) at (5,-6) {{\footnotesize $\{\,\mi(B\!:\!O)\,\}$}};
    \node[draw,rounded corners, fill=blue!20!] (71) at (0,-8) {{\footnotesize $\varnothing$}};
    \draw[-] (11.south) -- (21.north);
    \draw[-] (11.south) -- (22.north);
    \draw[-] (11.south) -- (23.north);
    \draw[-] (21.south) -- (31.north);
    \draw[-] (21.south) -- (32.north);
    \draw[-] (22.south) -- (31.north);
    \draw[-] (22.south) -- (33.north);
    \draw[-] (23.south) -- (32.north);
    \draw[-] (23.south) -- (33.north);
    \draw[-] (31.south) -- (41.north);
    \draw[-] (31.south) -- (42.north);
    \draw[-] (32.south) -- (41.north);
    \draw[-] (32.south) -- (43.north);
    \draw[-] (33.south) -- (41.north);
    \draw[-] (33.south) -- (44.north);
    \draw[-] (41.south) -- (51.north);
    \draw[-] (41.south) -- (52.north);
    \draw[-] (41.south) -- (53.north);
    \draw[-] (42.south) -- (51.north);
    \draw[-] (43.south) -- (52.north);
    \draw[-] (44.south) -- (53.north);
    \draw[-] (51.south) -- (61.north);
    \draw[-] (51.south) -- (62.north);
    \draw[-] (52.south) -- (61.north);
    \draw[-] (52.south) -- (63.north);
    \draw[-] (53.south) -- (62.north);
    \draw[-] (53.south) -- (63.north);
    \draw[-] (61.south) -- (71.north);
    \draw[-] (62.south) -- (71.north);
    \draw[-] (63.south) -- (71.north);
    \end{tikzpicture}
    \caption{
    The Hasse diagram of the down-set lattice
    $\lat{MID}$ of the $\N=2$ MI-poset. Each vertex represents a down-set, labeled by the antichain generating it. The elements in blue are the elements of $\lat{PMI}$, which for $\N=2$ is a subset of the down-set lattice.
    The remaining (white) elements are not in $\lat{PMI}$, because there are additional MI instances which necessarily vanish as a consequence of linear dependence.  For example, $\{\mi(A\!:\!B), (A\!:\!O)\}$ both vanishing implies  $\mi(A\!:\!BO)=0$ as well.}
    \label{fig:ALatticeMIPosetN2}
\end{figure}
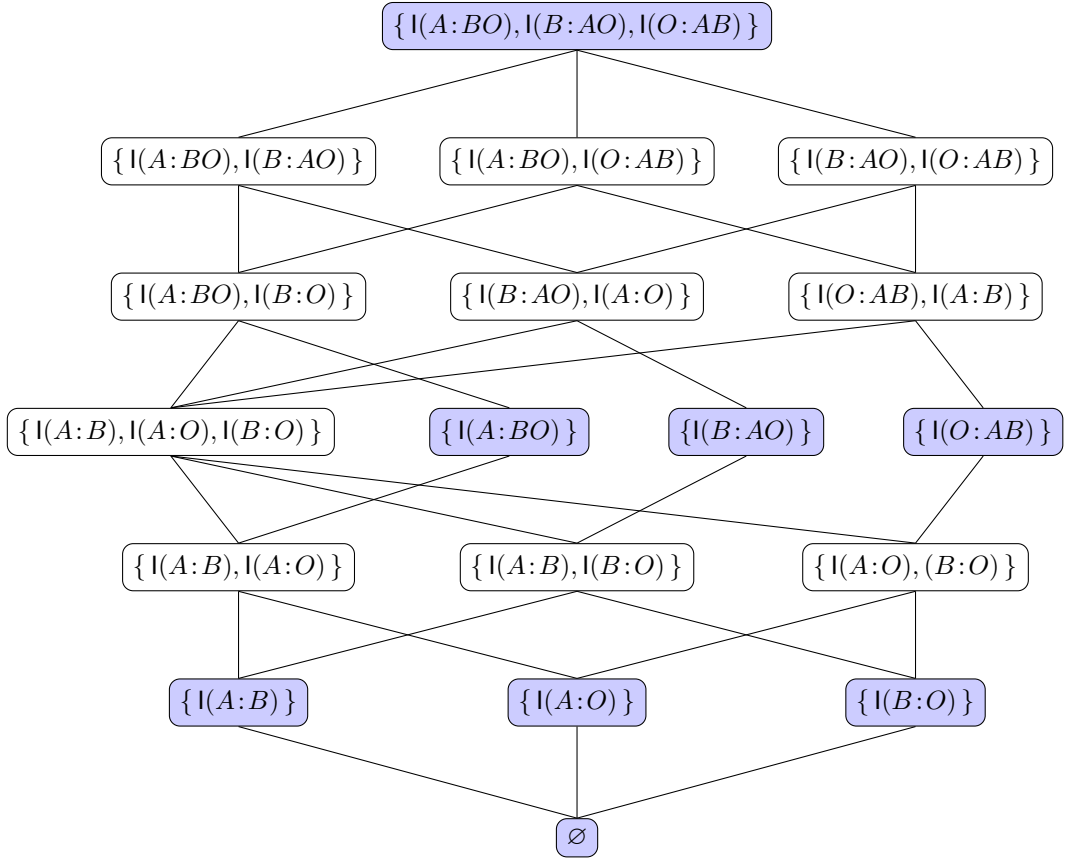

For any set $\mathcal{X}$, the lattice obtained by endowing the power set with a partial order given by inclusion is known as the \textit{power set lattice} of $\mathcal{X}$, and we denote it by $\mathfrak{P}_\mathcal{X}$. As for down-set lattices, the meet and join correspond to the intersection and union respectively, so it follows that $\lat{MID}$ is a sublattice of $\mathfrak{P}_\mgs$, the power set lattice of the set of MI instances $\mgs$. Furthermore, since any PMI is specified by its $\zds$ $\mgs^0\in \mathfrak{P}_\mgs$, $\lat{PMI}$ can be viewed as a subset of $\mathfrak{P}_\mgs$. It is however not a sublattice, since the join in $\lat{PMI}$ is not (in general) the union of the sets of vanishing MI instances. Nevertheless, since both $\lat{MID}$ and $\lat{PMI}$ are subsets of $\mathfrak{P}_\mgs$, we have (as sets)
\begin{equation}
\label{eq:lattice_intersection}
    \lat{KC}=\lat{PMI}\cap\lat{MID}\,.
\end{equation}
Furthermore, the order relation in all these lattices is the same as the one inherited from $\mathfrak{P}_\mgs$ (cf.\ \Cref{fn:revincl}).

We show in \Cref{fig:ALatticeMIPosetN2} the simple example of these lattices for $\N=2$, where we label each element by the \textit{antichain} that generates it.\footnote{\, An \textit{antichain} $\mathcal{A}$ is a subset $\mathcal{A}\subseteq\mathcal{P}$ of a poset $(\mathcal{P},\preceq)$ that is totally unordered, i.e., all elements in $\mathcal A$ are incomparable. Any antichain $\mathcal{A}$ naturally generates the down-set $\mathcal{D}_{\!\mathcal{A}}$ via the relation
\begin{equation*}
    \mathcal{D}_{\!\mathcal{A}}=\{x\in\mathcal{P},\;x\preceq\mathcal{A}\}\,.
\end{equation*}
Since the MI-poset is finite, any down-set is generated by an antichain in this fashion.} Finally, we stress that when we represent a PMI by its set of vanishing MI instances $\mgs^0$, all these lattices have the same meet operation, which is simply given by intersection. 

Since both the lattices $\lat{MID}$ and $\lat{PMI}$ have several nice structural properties, in light of \eqref{eq:lattice_intersection} it is natural to ask which of these properties, if any, are preserved in $\lat{KC}$ for an arbitrary number of parties, and if $\lat{KC}$ has additional properties which are not inherited from its parent lattices. The main motivation for these investigations is that there exist classes of lattices whose theory is well developed, and for which there are powerful tools to characterize their structure in more detail.

\subsection{The lattice \texorpdfstring{$\lat{KC}^3$}{} of KC-compatible PMIs for \texorpdfstring{$\N=3$}{N=3}}
\label{ssec:lattice3}

\begin{table}[tb]
    \centering
    \footnotesize
    \begin{tabular}{c|c|c|c|c|c|c|c|c|}
        \cline{2-9}
        & 0 & 1 & 2 & 3 & 4 & 5 & 6 & 7\\
        \hline
        \multicolumn{1}{|c|}{$\lat{PMI}^3$} & 1 & 11 & 48 & 107 & 127 & 75 & 18 & 1\\
        \hline
        \multicolumn{1}{|c|}{$\lat{KC}^3$} & 1 & 7 & 21 & 35 & 32 & 15 & 6 & 1\\
        \hline
    \end{tabular}
    \caption{The number of subspaces in $\lat{PMI}^3$ and $\lat{KC}^3$ for each dimension $0\leq d\leq 7$.}
    \label{tab:n3}
\end{table}

This subsection is devoted to a detailed description of the lattice of KC-compatible PMIs for $\N=3$, which we will denote as $\lat{KC}^3$. We will look at several structural properties, which we will then analyze for larger $\N$ in the following subsections.
Throughout this subsection we will briefly review 
relevant definitions and results from lattice theory; for more details the reader is referred to the standard books \cite{birkhoff1967lattice,gratzerbook,davey1990introduction}.

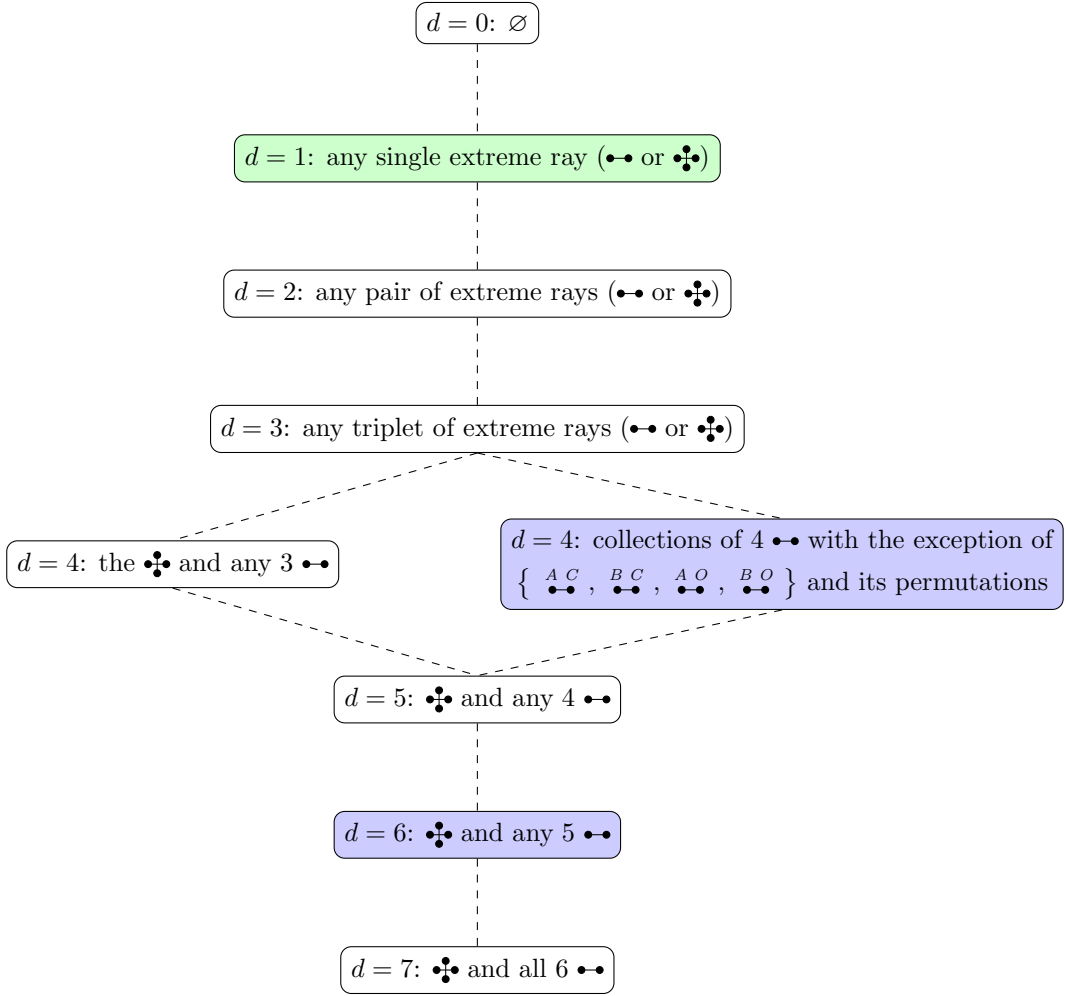
\begin{figure}[tb]
    \centering
    \begin{tikzpicture}[scale=0.9]
    \node[draw,rounded corners] (0) at (0,8) {{\footnotesize $d=0$: $\varnothing$}};
    \node[draw,rounded corners, fill=green!20!] (1) at (0,6) {{\footnotesize $d=1$: any single extreme ray ($\bpg$ or $\pt$)}};
    \node[draw,rounded corners] (2) at (0,4) {{\footnotesize $d=2$: any pair of extreme rays ($\bpg$ or $\pt$)}};
    \node[draw,rounded corners] (3) at (0,2) {{\footnotesize $d=3$: any triplet of extreme rays ($\bpg$ or $\pt$)}};
    \node[draw,rounded corners] (41) at (-4.5,0) {{\footnotesize $d=4$: the $\pt$ and any $3$ $\bpg$}};
    \node[draw, align=center,rounded corners, fill=blue!20!] (42) at (4.5,0) 
    {{\footnotesize $d=4$: collections of 4 $\bpg$ with the exception of}\\ {\footnotesize $\bigl\{\,\bp{A}{C}, \bp{B}{C}, \bp{A}{O}, \bp{B}{O}\bigr\}$ and its permutations}};
    \node[draw,rounded corners] (5) at (0,-2) {{\footnotesize $d=5$: $\pt$ and any $4$ $\bpg$}};
    \node[draw,rounded corners, fill=blue!20!] (6) at (0,-4) {{\footnotesize $d=6$: $\pt$ and any $5$ $\bpg$}};
    \node[draw,rounded corners] (7) at (0,-6) {{\footnotesize $d=7$: $\pt$ and all $6$ $\bpg$}};
    \draw[dashed] (0.south) -- (1.north);
    \draw[dashed] (1.south) -- (2.north);
    \draw[dashed] (2.south) -- (3.north);
    \draw[dashed] (3.south) -- (41.north);
    \draw[dashed] (3.south) -- (42.north);
    \draw[dashed] (41.south) -- (5.north);
    \draw[dashed] (42.south) -- (5.north);
    \draw[dashed] (5.south) -- (6.north);
    \draw[dashed] (6.south) -- (7.north);
    \end{tikzpicture}
    \caption{A schematic representation (not the Hasse diagram) of the lattice $\lat{KC}^3$, where each block corresponds to a collection of PMIs. The dashed lines describe the cover relations schematically. For each block, we have indicated the type of collection of SAC$_3$ extreme rays that spans the PMI. The symbols $\protect\bpg$ and $\protect\pt$ represent respectively the entropy vector of a Bell pair, and the entropy vector of the $4$-party perfect state. The join-irreducible elements (including the $d=6$ atoms) are shown in blue, and the meet-irreducible elements 
    (here just the $d=1$ coatoms) are shown in green (the terminology is defined below).}  
    \label{fig:N3KPMIlattice}
\end{figure}


\paragraph{Description of $\lat{KC}^3$:} The lattice $\lat{KC}^3$ has a total of 118 elements. In \Cref{tab:n3} we report the number of elements for each dimension $0\leq d\leq 7$, and compare it to the number of elements of the same dimension in the full lattice of PMIs $\lat{PMI}^3$. Given the fairly large number of elements, drawing the full Hasse diagram of $\lat{KC}^3$ is impractical, and we draw instead a more schematic representation of its structure in \Cref{fig:N3KPMIlattice}. In this diagram, each block corresponds to a collection of PMIs with the indicated dimension. The dashed lines describe the cover relations schematically. When two blocks are connected by a line, we mean that every element in one block covers (or is covered by) at least one element in the other block, and at the same time that every element does not cover (and is not covered by) any element in any unconnected block. The PMIs in a block are described slightly differently from the representations that we have seen earlier. Instead of writing a PMI as the set of all vanishing MI instances $\mgs^0$, or representing it by all the extreme rays generating the corresponding face of the SAC$_3$, we describe each PMI $\pmi$ by a collection of elements of a fixed basis of entropy space consisting of 1-dimensional PMIs in $\lat{KC}^3$; the said collection comprises of all such 1-dimensional PMIs contained in $\pmi$. 
The fact that this is possible is not obvious, and is in fact a special property of the $\N=3$ case, as we will discuss later. 

Each 1-dimensional PMI in $\lat{KC}^3$ corresponds to an extreme ray of SAC$_3$, and there are two types of such rays in $\lat{KC}^3$. One is the entropy vector of a Bell pair state (denoted by $\bpg$), and there are six possibilities, depending on which pair of parties we choose among $\{A,B,C,O\}$. The second is the entropy vector associated to the 4-party perfect state (denoted by $\pt$), which is unique since this state is invariant under any permutation of the parties.\footnote{\, This is the entropy vector of the density matrix obtained by tracing out any of the parties from a 4-party state which is maximally entangled for all bipartitions.} 

From \Cref{fig:N3KPMIlattice} we can make the following observations (explained in the next paragraph): Any collection of generators which include $\pt$ is an element of $\lat{KC}^3$, and this is also the case for any collection of three or fewer $\bpg$. Almost all collections of four $\bpg$ (without $\pt$) also belong to $\lat{KC}^3$, with the only three exceptions being indicated in \Cref{fig:N3KPMIlattice}. Lastly, any collection of five or six $\bpg$ (again without $\pt$) is not an element of $\lat{KC}^3$. 

To develop more intuition for why $\lat{KC}^3$ has this structure, let us first consider any collection of generators which include $\pt$.  Any entropy vector in the interior of the corresponding face has mutual information which is strictly positive whenever it involves at least three parties, e.g.\ $\mi(A\!:\!BC)>0$ (since by SA the positive contribution from $\pt$ cannot be canceled by any other contribution). This means that the only instances which can possibly vanish are MIs involving single parties, so any such collection is automatically a down-set (since it only contains minimal elements of the MI-poset). On the other hand, consider the collection $\{\,\bp{A}{C}, \bp{B}{C}, \bp{A}{O}, \bp{B}{O}\,\}$.  Any vector in the span of these Bell pairs has $\mi(A\!:\!B)=0$ and $\mi(C\!:\!O)=0$, while all instances involving at least three parties are strictly positive.  This means that $\pt$ lies within this subspace (i.e., its addition does not alter the PMI).  Therefore since $\pt$ was not included in our collection, such a collection is not in $\lat{PMI}^3$ and hence cannot lie in $\lat{KC}^3$. It is clear that adding any more Bell pairs to our collection still retains the redundancy of $\pt$ (and hence the requirement of its presence in order for the collection to be in $\lat{PMI}^3$); conversely, we leave it as an exercise for the reader to verify that any fewer Bell pairs would always maintain some vanishing mutual information involving 3 parties (so that $\pt$ does not lie within this subspace).

Armed with this understanding, we can additionally derive the numbers in the last row of \Cref{tab:n3}.  For example, as we already mentioned above, there are seven 1-dimensional PMIs (six $\bpg$ and one $\pt$), 21 2-dimensional PMIs (15 of the type $\bpg \ \bpg$ and 6 $\bpg \ \pt$), and so on.  At the facet end ($d=6$), there are six facets which correspond to $\pt$ with five $\bpg$'s, the single missing $\bpg$ from this collection giving the only pair of single parties that have vanishing mutual information.

Finally, notice that since in \Cref{fig:N3KPMIlattice} we have represented each PMI by a basis, the full collection of extreme rays that generate the face corresponding to a PMI is not always visible. We stress that in general, a PMI in $\lat{KC}^{\N}$ can correspond to a face which includes on its boundary extreme rays that generate PMIs which are not in $\lat{KC}^{\N}$. This in fact already happens for $\N=3$, for any face of dimension $d\geq 5$.


\paragraph{The join operation:} 

Since $\lat{KC}^3$ is a subset of $\lat{PMI}^3$ and $\lat{MID}^3$  (cf.\ \eqref{eq:lattice_intersection}), with the same meet operation, it is natural to ask if the join in $\lat{KC}^3$ has any relation to the join in $\lat{PMI}^3$ or $\lat{MID}^3$, and in particular if $\lat{KC}^3$ is a sublattice of either 
of the other two. 

To this end, the representation of the elements of $\lat{KC}^3$ that we have used in \Cref{fig:N3KPMIlattice} is particularly convenient. Given a PMI $\pmi$ in $\lat{KC}^3$, we will denote its representation using the basis shown in \Cref{fig:N3KPMIlattice} by $\widehat{\pmi}$. Then one can immediately see from \Cref{fig:N3KPMIlattice} that given two PMIs $\pmi,\pmi'\in\lat{KC}^3$, their join in $\lat{KC}^3$ corresponds to their intersection 
\begin{equation} \label{eq:kc3join}
    \pmi''=\pmi\kcjoin\pmi' \qquad \text{with}\qquad \widehat{\pmi''}=\widehat{\pmi}\cap\widehat{\pmi'}\,.
\end{equation}
This means that the face corresponding to $\pmi\kcjoin\pmi'$ must be the intersection of the faces corresponding to $\pmi$ and $\pmi'$, implying by \eqref{eq:PMIjoin} that $\lat{KC}^3$ is in fact a sublattice of $\lat{PMI}^3$.

On the other hand, $\lat{KC}^3$ is not a sublattice of $\lat{MID}^3$. To see why, consider the three join-irreducible elements of $\lat{KC}^3$ with $\mgs^0$ corresponding to the principal down-sets of the following three instances in the MI-poset: 
\begin{equation} \label{eq:kcmidex}
    \mi(A\!:\!BC)\,,\quad \mi(O\!:\!AB)\,,\quad \mi(C\!:\!BO)\,.
\end{equation}
The join (in $\lat{MID}^3$) of these principal down-sets is just their union. On the other hand, the linear subspace corresponding to the vanishing of all such MI instances is the origin of $\mathbb{R}^{\D}$, so the join in $\lat{KC}^3$ of the principal down-sets of the instances in \eqref{eq:kcmidex} is the entire MI-poset. Since the join in $\lat{KC}^3$ is not the same as that in $\lat{MID}^3$, $\lat{KC}^3$ cannot be a sublattice of $\lat{MID}^3$.


\paragraph{Atomistic and coatomistic lattices:} In a poset with a bottom element, the elements that cover the bottom are called \textit{atoms}, and a lattice is said to be \textit{atomistic} if every element (other than the bottom) is the join of a collection of atoms. In a finite lattice $\lat{}$, an element $x\in\lat{}$ is said to be \textit{join-irreducible} if it is not the bottom element, and the following implication holds
\begin{equation}
    x=y\vee z \quad\implies\quad x=y\;\; \text{or}\;\; x=z\quad \forall\,y,z\in\lat{}\,.
\end{equation}
Notice that an element is join-irreducible if and only if it covers precisely one element, so in particular all atoms are join-irreducible.

All these notions have a dual version. If a poset has a top element, the elements covered by the top are called \textit{coatoms}, and a lattice is said to be \textit{coatomistic} if every element (other than the top) is the meet of a collection of coatoms. In a finite lattice $\lat{}$, an element $x\in\lat{}$ is said to be \textit{meet-irreducible} if it is not the top element, and the following implication holds
\begin{equation}
    x=y\wedge z \quad\implies\quad x=y\;\; \text{or}\;\; x=z\quad \forall\,y,z\in\lat{}\,.
\end{equation}
Correspondingly, an element is meet-irreducible if and only if it is covered by precisely one element, so in particular all coatoms are meet-irreducible.

For any $\N$, the lattice $\lat{SAC}^{\N}$ is both atomistic and coatomistic, with atoms corresponding to facets, and coatoms to extreme rays.\footnote{\, Obviously this is true for any polyhedral cone.} Any face can in fact be obtained from the intersection (join) of a collection of facets, or as the composition (meet) of a collection of extreme rays. On the other hand, the lattice $\lat{MID}$ is neither atomistic nor coatomistic. The reason is that in any down-set lattice, the elements corresponding to the down-sets generated by single-element antichains (also called \textit{principal}) are join-irreducible, and therefore are not in general joins of atoms.\footnote{\, However, if the poset is totally unordered, the down-set lattice is just a power-set lattice, which is atomistic because all join-irreducible elements are atoms.} Furthermore, it is a general fact that in any down-set lattice,\footnote{\label{footnote:birkhoff}\, In fact, this is true for any distributive lattice, since by Birkhoff's representation theorem any distributive lattice is the lattice of down-sets of the poset of its join-irreducible elements \cite{birkhoff1967lattice}.} the poset of join-irreducible elements (with the induced partial order from the lattice) is isomorphic to the order dual of the poset of meet-irreducible elements (this can be seen explicitly for $\N=2$ in \Cref{fig:ALatticeMIPosetN2}), making the lattice in general not coatomistic.\footnote{\, Again, this holds except when the poset is totally unordered.}

In light of \eqref{eq:lattice_intersection}, since in a finite lattice any element can be obtained as the join of join-irreducible elements, or as the meet of meet-irreducible elements, it is then interesting to investigate the structure of the set of such elements for $\lat{KC}^\N$, and in particular whether $\lat{KC}^\N$ is atomistic or coatomistic.  
As one can immediately verify from \Cref{fig:N3KPMIlattice}, in $\lat{KC}^3$ there are six atoms, corresponding to hyperplanes of the form
\begin{equation} \label{eq:singleton}
    \ghyp:\qquad \mi(\uI:\uK)=0,\quad \text{with}\quad |\uI|=|\uK|=1\,.
\end{equation}
As we mentioned above, these are atoms of $\lat{PMI}^3$. The fact they are also atoms of $\lat{MID}^3$ follows from them being the minimal elements of the MI-poset, and therefore each single MI instance in \eqref{eq:singleton} is by itself a down-set. 

To check whether $\lat{KC}^3$ is atomistic, it suffices to note that the join of all the atoms is the PMI generated by the extreme ray $\pt$. This immediately follows from \eqref{eq:kc3join}, since as one can see from \Cref{fig:N3KPMIlattice}, $\pt$ is the intersection of the collections of generators of all the atoms. The fact that the join of all the atoms is not the top means that the top cannot be the join of any collection of atoms, and so $\lat{KC}^3$ is not atomistic.

It is also immediate to see that $\lat{KC}^3$ \textit{is} coatomistic, since in the representations of the PMIs given \Cref{fig:N3KPMIlattice}, every element is explicitly a collection of coatoms. Furthermore, it is clear that in $\lat{KC}^3$ all coatoms are 1-dimensional PMIs, and that there are no meet-irreducible elements other than the coatoms (as it must be, given that the lattice is coatomistic).

We conclude the analysis of the meet and join-irreducible elements in $\lat{KC}^3$ with an intriguing observation that we leave as an exercise for the reader to verify. The join-irreducible elements of $\lat{KC}^3$ correspond precisely to the principal down-sets of the subposet\footnote{\, A subposet of a poset $(\mathcal{P},\preceq)$ is a subset of $\mathcal{P}$ with the partial order induced by $(\mathcal{P},\preceq)$.
} 
of the MI-poset obtained by simply deleting the trivial MI instances given in \eqref{eq:trivial_ins}. We will comment again on this property in the next subsection, where we discuss the general case of $\N\geq 4$.


\paragraph{Modularity and semimodularity:} From our previous observations, it is immediately clear that $\lat{KC}^3$ is not distributive, since the poset of join-irreducible elements is not dually isomorphic to the poset of meet-irreducibles (cf.\ \Cref{footnote:birkhoff}). There are however two important generalizations of distributivity, namely \textit{modularity} and \textit{semi-distributivity}, that are interesting to explore. Let us first consider the former.

Modular lattices are defined as those lattices that obey the \textit{modular law}, which states that for any three elements $x,y,z \in \lat{}$,
\begin{align}
\label{eq:modular_condition}
    x \preceq y \quad\implies\quad x \vee (z \wedge y) = (x \vee z) \wedge y\,.
\end{align}
A well known result in lattice theory asserts that any lattice is modular if and only if it does not contain a sublattice isomorphic to the \textit{pentagon} lattice $\mathfrak{N}_5$. As shown in \Cref{fig:N5}, $\lat{KC}^3$ contains an instance of $\mathfrak{N}_5$ as a sublattice (which we leave as an exercise for the reader to show\footnote{\, Notice that the cover relations in $\mathfrak N_5$ need not coincide with those in $\lat{KC}^3$. For example, unlike in $\lat{KC}^3$, $e$ covers $b$ in $\mathfrak N_5$ (cf.\  \Cref{fig:N5}).}), and therefore it is not modular. There is however a generalization of modularity which is also interesting to explore.

\begin{figure}
\centering
\begin{tikzpicture}[scale=0.9]
\node[draw,rounded corners] (top) at (0,6) {{\footnotesize $a = \bigl\{\,\bp{A}{C}, \bp{B}{C}\,\bigr\}$}};
\node[draw,rounded corners] (1left) at (-4,4) {{\footnotesize $e = \bigl\{\, \bp{A}{C}, \bp{B}{C}, \bp{A}{O}\,\bigr \}$}};
\node[draw,rounded corners] (1right) at (4,4) {{\footnotesize $d = \bigl\{\, \bp{A}{C}, \bp{B}{C}, \bp{B}{O}\,\bigr \}$}};
\node[draw,rounded corners] (2right) at (4,2) {{\footnotesize $c = \bigl\{\, \bp{A}{C}, \bp{B}{C}, \bp{B}{O}, \pt\,\bigr \}$}};
\node[draw,rounded corners] (bottom) at (0,0) {{\footnotesize $b = \bigl\{\,\bp{A}{C}, \bp{B}{C}, \bp{A}{O}, \bp{B}{O}, \pt\,\bigr \}$}};
\draw[-] (top.south) -- (1left.north);
\draw[-] (top.south) -- (1right.north);
\draw[-] (1left.south) -- (bottom.north);
\draw[-] (1right.south) -- (2right.north);
\draw[-] (2right.south) -- (bottom.north);
\end{tikzpicture}
\caption{An explicit realization of the pentagon lattice $\mathfrak{N}_5$ as a sublattice of $\lat{KC}^3$, implying that $\lat{KC}^3$ is not modular. Notice that even if $c \preceq d$, we have $c \vee (e \wedge d)=c\vee b=c$ while $(c \vee e) \wedge d = a \wedge d=d$, violating \eqref{eq:modular_condition}.}
\label{fig:N5}
\end{figure}
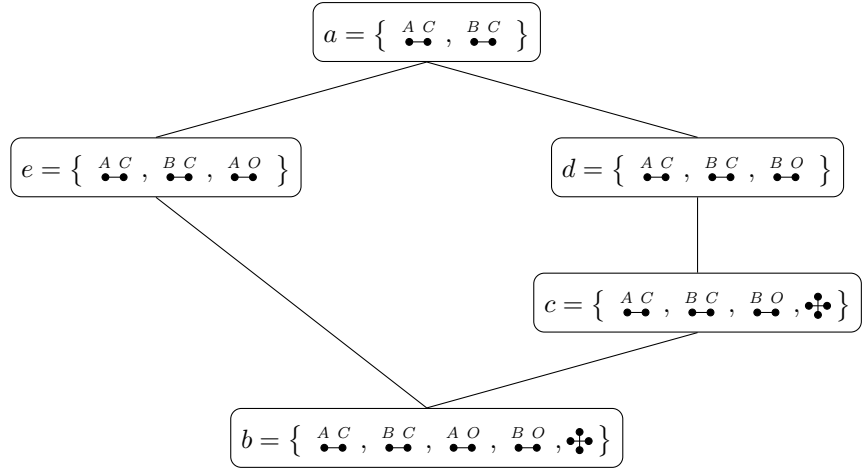

A \emph{semimodular lattice}\footnote{\, We follow the common convention of calling a lattice simply \emph{semimodular} when it is actually \textit{upper semimodular}. The dual notion of \emph{lower semimodularlity} is introduced below.} is a lattice $\lat{}$ such that for any $x,y,z \in \lat{}$,
\begin{align}
\label{eq:semimodular_cond}
    x \covb \,y \quad\implies\quad x \vee z \covb \,y \vee z \;\;\; \text{or}\;\;\; x \vee z = y \vee z\,.
\end{align}
Dually, a \emph{lower semimodular lattice} is one that satisfies the dual condition to \eqref{eq:semimodular_cond}, i.e., such that for any $x,y,z \in \lat{}$,
\begin{align}
    x \covb y \quad\implies\quad x \wedge z \covb y \wedge z\;\;\; \text{or}\;\;\; x \wedge z = y \wedge z\,.
\end{align}
A finite lattice is modular if and only if it is both upper and lower semimodular. We will now verify that $\lat{KC}^{3}$ is semimodular, thus also implying that it is not lower semimodular.

To see that $\lat{KC}^{3}$ is semimodular, it will be convenient to use a different, but equivalent, characterization of semimodularity. A lattice $\lat{}$ is semimodular if and only if it is graded by its height function (and therefore satisfies the JDCC) and its height function is \textit{submodular}, that is, for any $x,y \in \lat{}$, we have\footnote{\, Similarly, a lattice is lower-semimodular if it is graded by its height function, and the height function is \textit{supermodular}, i.e., $h(x) + h(y) \leq h(x \vee y) + h(x \wedge y)$ for any $x,y \in \lat{}$.}
\begin{align}\label{eq:height-semimod}
    h(x) + h(y) \geq h(x \vee y) + h(x \wedge y)\,.
\end{align}

Thus, we should first check that $\lat{KC}^3$ satisfies the JDCC. This is clear from \Cref{fig:N3KPMIlattice}, since starting from any $\pmi\in\lat{KC}^3$, we can get any $\pmi'$ covering it (or covered by it) by simply removing (or adding) a single $\bpg$ or $\pt$ from its representation $\widehat{\pmi}$. Furthermore, one can immediately verify using \Cref{fig:N3KPMIlattice} that the height function is given by $h(\pmi)=7-|\widehat{\pmi}|$, where $|\widehat{\pmi}|$ is the cardinality of $\widehat{\pmi}$. Since the join in this representation of PMIs in $\lat{KC}^3$ corresponds to intersection (cf.\ \eqref{eq:kc3join}), given two elements $\pmi$ and $\pmi'$, we have
\begin{align}
    h(\pmi \vee \pmi') = 7 - (|\widehat{\pmi}\cap\widehat{\pmi'}|)\,.
\end{align}
Furthermore, if either $\widehat{\pmi}$ or $\widehat{\pmi'}$ contains $\pt$ we have
\begin{equation} 
    \pmi''=\pmi\wedge\pmi' \qquad \text{with}\qquad \widehat{\pmi''}=\widehat{\pmi}\cup\widehat{\pmi'}\,,
\end{equation}
which means
\begin{align}
    h(\pmi \wedge \pmi') = 7 - (|\widehat{\pmi}\cup\widehat{\pmi'}|)\,,
\end{align}
and \eqref{eq:height-semimod} holds as an equality. Therefore it only remains to check the case where neither $\widehat{\pmi}$ nor $\widehat{\pmi'}$ contains $\pt$. In this case, notice again from \Cref{fig:N3KPMIlattice} that
\begin{equation} 
    \pmi''=\pmi\wedge\pmi' \qquad \text{with}\qquad \widehat{\pmi''}= \widehat{\pmi}\cup\widehat{\pmi'}\cup\{\pt\}\,.
\end{equation}
This implies
\begin{equation}
    h(\pmi\wedge\pmi') < 7 - (|\widehat{\pmi}\cup\widehat{\pmi'}|)\,,
\end{equation}
from which it follows that \eqref{eq:height-semimod} is strictly satisfied.


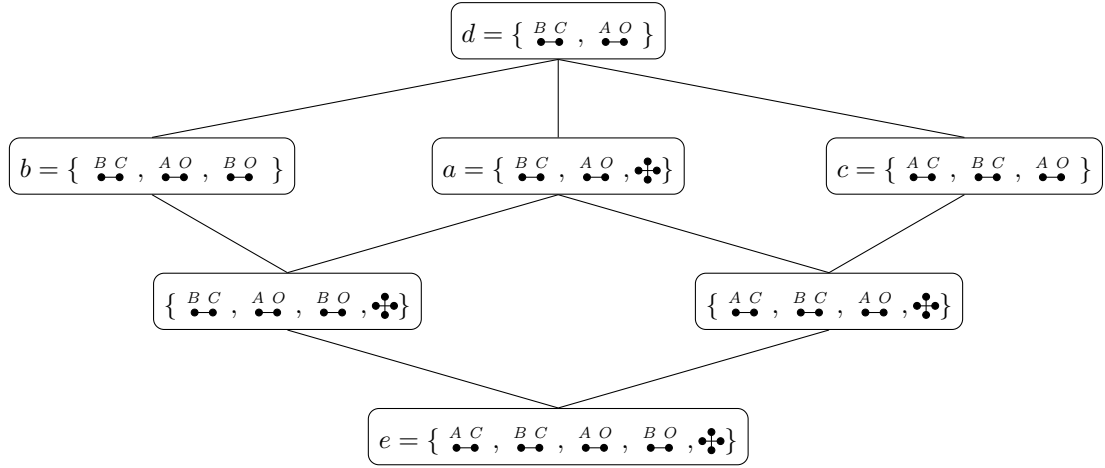
\begin{figure}
    \centering
    \begin{tikzpicture}[scale=0.9]
    \node[draw,rounded corners] (top) at (0,6) {{\footnotesize $d = \{\bp{B}{C}, \bp{A}{O}\}$}};
    \node[draw,rounded corners] (1left) at (-6,4) {{\footnotesize $b =\{\,\bp{B}{C},\bp{A}{O},\bp{B}{O}\,\}$}};
    \node[draw,rounded corners] (1mid) at (0,4) {{\footnotesize $a = \{\bp{B}{C},\bp{A}{O}, \pt\}$}};
    \node[draw,rounded corners] (1right) at (6,4) {{\footnotesize $c = \{\bp{A}{C},\bp{B}{C},\bp{A}{O}\}$}};
    \node[draw,rounded corners] (2left) at (-4,2) {{\footnotesize $\{\bp{B}{C},\bp{A}{O},\bp{B}{O}, \pt\}$}};
    \node[draw, align=center,rounded corners] (2right) at (4,2) {{\footnotesize $\{\bp{A}{C},\bp{B}{C},\bp{A}{O},\pt\}$}};
    \node[draw, align=center,rounded corners] (bottom) at (0,0) {{\footnotesize $e= \{\bp{A}{C}, \bp{B}{C}, \bp{A}{O}, \bp{B}{O}, \pt\}$}};    
    \draw[-] (top.south) -- (1left.north);
    \draw[-] (top.south) -- (1mid.north);
    \draw[-] (top.south) -- (1right.north);
    \draw[-] (1left.south) -- (2left.north);
    \draw[-] (1mid.south) -- (2left.north);
    \draw[-] (1mid.south) -- (2right.north);
    \draw[-] (1right.south) -- (2right.north);
    \draw[-] (2left.south) -- (bottom.north);
    \draw[-] (2right.south) -- (bottom.north);
    \end{tikzpicture}
    \caption{A sublattice of $\lat{KC}^{3}$, called $\mathfrak{S}_7$, showing that $\lat{KC}^{3}$ is not join-semidistributive. Notice that $a \vee b = a \vee c=d$, while $a \vee (b \wedge c)=a\vee e=a$, violating \eqref{eq:sdjoin}. }
    \label{fig:not-joinsemi}
\end{figure}

\paragraph{Semi-distributivity:} We now turn to the other generalization of distributivity mentioned above. Given arbitrary elements $x,y,z$ in a lattice $\lat{}$, the \emph{join-semidistributive law} (SD$_\vee$) and the \emph{meet-semidistributive law} (SD$_\wedge$) are respectively defined as
\begin{align}
    \text{SD$_\vee$:} & \quad u = x \vee y = x \vee z \quad \implies\quad u = x \vee (y \wedge z) \label{eq:sdjoin}\\
    \text{SD$_\wedge$:} & \quad u = x \wedge y = x \wedge z \quad\implies\quad u = x \wedge (y \vee z)\,.\label{eq:sdmeet}
\end{align}
A lattice $\lat{}$ is \textit{join-semidistributive} if it satisfies SD$_\vee$. Dually, a lattice $\lat{}$ is \textit{meet-semi-distributive} if it satisfies SD$_\wedge$.\footnote{\, A lattice that is both meet-semidistributive and join-semidistributive is said to be \textit{semidistributive} (it is not necessarily distributive).}

A sublattice of $\lat{KC}^{3}$, showing that $\lat{KC}^{3}$ is \emph{not} join-semidistributive, is depicted in \Cref{fig:not-joinsemi}. On the other hand, $\lat{KC}^{3}$ is meet-semidistributive. To see this, we will use the following characterization of meet-semidistributivity.\footnote{\, This is the order dual of the characterization given in \cite{gratzer2016} for join-semidistributivity.} 
\begin{lemma}
\label{lem:meetsemidis}
    A lattice $\lat{}$ satisfies \emph{SD}$_\wedge$ if and only if for every join-irreducible element $x\in\lat{}$, and $x^*$ the unique element covered by it, there exists a meet-irreducible element $y(x)$ which is the unique maximal element $y$ such that $y\succeq x^*$ but $y \nsucceq x$.
\end{lemma}
\begin{proof}
   See Theorem 3-1.4 of \cite{gratzer2016}. 
\end{proof}

As shown in \Cref{fig:N3KPMIlattice}, the join-irreducible elements are the atoms and the allowed collections of four $\bpg$. If $\pmi$ is an atom, $\pmi^*$ is the bottom element, and the unique meet-irreducible element $\pmi'(\pmi)$ which satisfies the conditions on $y$ in \Cref{lem:meetsemidis} is the only coatom (extreme ray) which is not greater than $\pmi$ (since $\widehat{\pmi}$ contains $\pt$ and 5 $\bpg$). On the other hand, if $\pmi$ is one of the other join-irreducible elements, $\pmi^*$ is represented by the union of $\pt$ with the elements in $\widehat{\pmi}$, so the unique meet-irreducible element $\pmi'(\pmi)$ satisfying \Cref{lem:meetsemidis} is now the coatom represented by $\pt$. This shows that the conditions of \Cref{lem:meetsemidis} are satisfied, and that $\lat{KC}^3$ is meet-semidistributive.

\subsection{Comments on the lattice \texorpdfstring{$\lat{KC}^4$}{} of KC-compatible PMIs for \texorpdfstring{$\N=4$}{N=4}}
\label{ssec:lattic4}

\begin{table}
    \centering
    \footnotesize
    \begin{tabular}{c|c|c|c|c|c|c|c|c|c|}
        \cline{2-10}
        & 0 & 1 & 2 & 3 & 4 & 5 & 6 & 7 & 8\\
        \hline
        \multicolumn{1}{|c|}{$\lat{PMI}^4$} & 1 & 3085 & 66005 & 532585 & 2254005 & 5719656 & 9301825 & 10032200 & 7275805\\
        \hline
        \multicolumn{1}{|c|}{$\lat{KC}^4$} & 1 & 20 & 175 & 840 & 2465 & 4843 & 6345 & 5875 & 4100\\
        \hline
    \end{tabular}\\
    \vspace{0.5cm}
    \begin{tabular}{|c|c|c|c|c|c|c|c|c|}
        \hline
        9 & 10 & 11 & 12 & 13 & 14 & 15\\
        \hline
        3541900 & 1138826 & 234470 & 29455 & 2100 & 75 & 1\\
        \hline
        2300 & 1072 & 430 & 150 & 45 & 10 & 1\\
        \hline
    \end{tabular}
    \caption{The number of faces of $\lat{PMI}^4$ and $\lat{KC}^4$ for each dimension $0\leq d\leq 15$.}
    \label{tab:n4}
\end{table}


\begin{figure}[tbp]
    \centering
    \begin{tikzpicture}[scale=0.85]
    \node[draw,rounded corners, fill=green!20!] (1) at (0,21) {{\footnotesize $d=1$}};
    \node[draw,rounded corners] (2) at (0,20) {{\footnotesize $d=2$}};
    \node[draw,rounded corners] (3) at (0,19) {{\footnotesize $d=3$}};
    \node[draw,rounded corners] (4a) at (-4,17) {{\footnotesize $d=4$}};
    \node[draw,rounded corners, fill=green!20!] (4b) at (4,17) {{\footnotesize $d=4$}};
    \node (5a) at (-6,15) {$\bullet$};
    \node at (-5.1,15) {{\footnotesize $(d=5)$}};
    \node[draw,rounded corners] (5b) at (-2,15) {{\footnotesize $d=5$}};
    \node[draw,rounded corners, fill=green!20!] (5c) at (2,15) {{\footnotesize $d=5$}};
    \node[draw,rounded corners] (5d) at (6,15) {{\footnotesize $d=5$}};
    \node[draw,rounded corners] (6a) at (-4,13) {{\footnotesize $d=6$}};
    \node[draw,rounded corners, fill=blue!20!] (6b) at (0,13) {{\footnotesize $d=6$}};
    \node[draw,rounded corners, fill=green!20!] (6c) at (4,13) {{\footnotesize $d=6$}};
    \node[draw,rounded corners] (7) at (0,11) {{\footnotesize $d=7$}};
    \node[draw,rounded corners] (8a) at (-3,9) {{\footnotesize $d=8$}};
    \node[draw,rounded corners, fill=blue!20!] (8b) at (3,9) {{\footnotesize $d=8$}};
    \node[draw,rounded corners] (9) at (0,7) {{\footnotesize $d=9$}};
    \node[draw,rounded corners] (10) at (0,6) {{\footnotesize $d=10$}};
    \node[draw,rounded corners] (11) at (0,5) {{\footnotesize $d=11$}};
    \node[draw,rounded corners] (12a) at (-3,3) {{\footnotesize $d=12$}};
    \node[draw,rounded corners, fill=blue!20!] (12b) at (3,3) {{\footnotesize $d=12$}};
    \node[draw,rounded corners] (13) at (0,1) {{\footnotesize $d=13$}};
    \node[draw,rounded corners, fill=blue!20!] (14) at (0,0) {{\footnotesize $d=14$}};
    \draw[dashed] (1.south) -- (2.north);
    \draw[dashed] (2.south) -- (3.north);
    \draw[dashed] (3.south west) -- (4a.north);
    \draw[dashed] (4b.north) -- (3.south east);
    \draw[dashed] (4a.south) -- ([xshift=-10pt,yshift=-9pt]5a.north east);
    \draw[dashed] (4a.south) -- (5b.north);
    \draw[dashed] (5c.north) -- (4a.south east);
    \draw[dashed] (4b.south) -- (5b.north);
    \draw[dashed,red] (5d.north) .. controls (6,17.5) and (4,18.5) .. (3.east);
    \draw[dashed] ([xshift=1pt,yshift=7pt]5a.south) -- (6a.north);
    \draw[dashed] (5b.south) -- (6a.north);
    \draw[dashed] (5c.south) -- (6a.north);
    \draw[dashed] (5d.south) -- (6a.north);
    \draw[dashed] (5b.south) -- (6b.north);
    \draw[dashed] (6c.north) -- (5b.south east);
    \draw[dashed] (6a.south) -- (7.north);
    \draw[dashed] (6b.south) -- (7.north);
    \draw[dashed] (6c.south) -- (7.north);
    \draw[dashed] (7.south) -- (8a.north);
    \draw[dashed] (7.south) -- (8b.north);
    \draw[dashed] (8a.south) -- (9.north west);
    \draw[dashed] (8b.south) -- (9.north east);
    \draw[dashed] (9.south) -- (10.north);
    \draw[dashed] (10.south) -- (11.north);
    \draw[dashed] (11.south) -- (12a.north);
    \draw[dashed] (11.south) -- (12b.north);
    \draw[dashed] (12a.south) -- (13.north west);
    \draw[dashed] (12b.south) -- (13.north east);
    \draw[dashed] (13.south) -- (14.north);
    \end{tikzpicture}
    \caption{A schematic representation of $\lat{KC}^4$ where the top (origin) and bottom ($\mathbb{R}^{\D}$) elements have been omitted. The blue (green) elements are join-irreducible (meet-irreducible). The red line highlights a cover relation between elements whose dimensions differ by more than one. The solid dot is the join of all the atoms.}
    \label{fig:N4KClattice}
\end{figure}
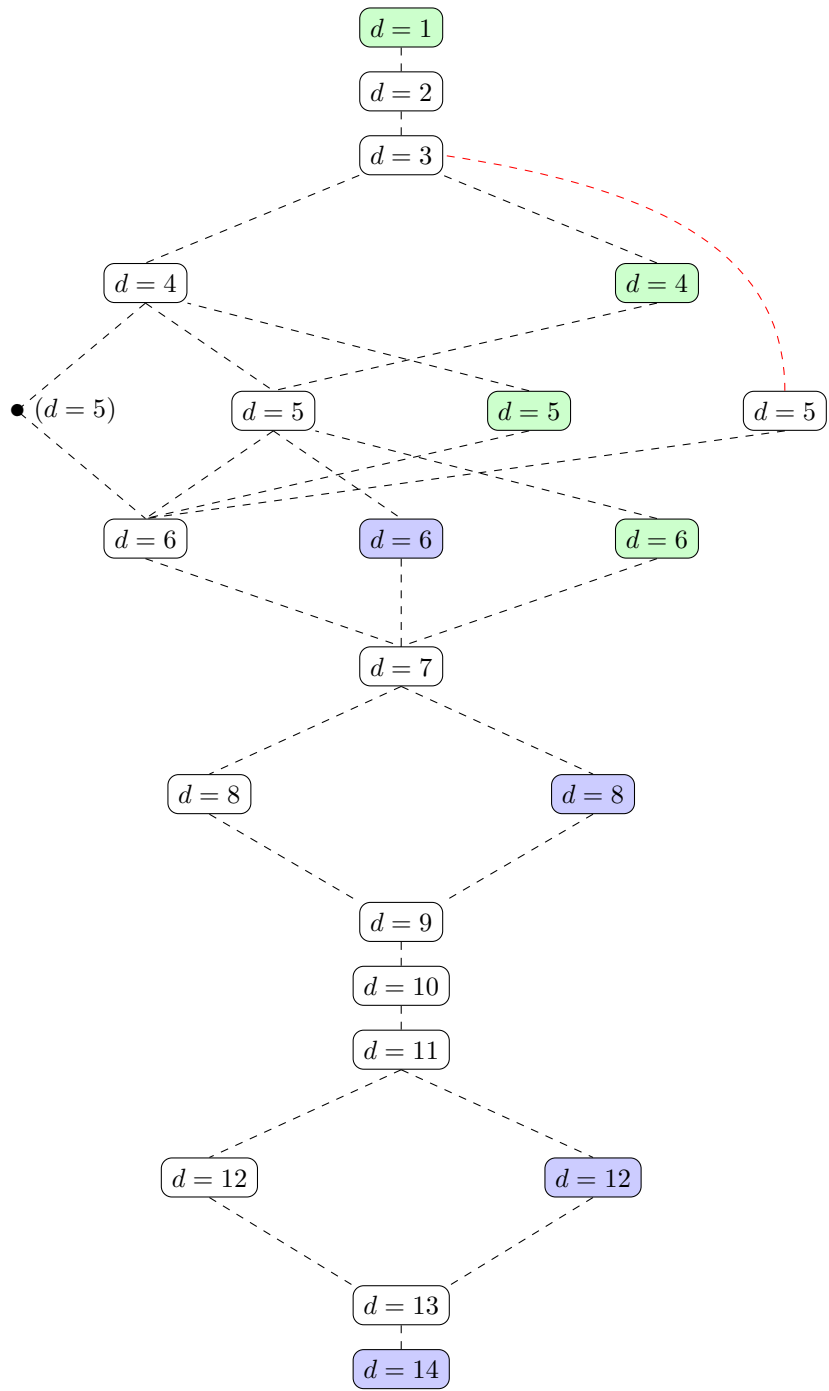

The lattice $\lat{KC}^4$ is considerably more complicated than $\lat{KC}^3$, and a graphic representation similar to the one in \Cref{fig:N3KPMIlattice} is unfeasible. The number of elements for each dimension is given in \Cref{tab:n4}, where we also show the number of faces of the SAC$_4$. One can immediately appreciate how small $\lat{KC}^4$ is compared to $\lat{PMI}^4$. Nevertheless, its structure is already sufficiently rich to mark a departure from the simple case of $\N=3$. Specifically, it turns out that $\lat{KC}^4$ is \textit{not coatomistic}, \textit{does not satisfy the} JDCC, is \textit{not a sublattice} of $\lat{PMI}^4$, and is \textit{not meet-semidistributive}. 

A cartoon of the structure of $\lat{KC}^4$ is shown in \Cref{fig:N4KClattice}, where the top and bottom elements have been omitted. The boxes represent sets of elements with the indicated dimension, and the dashed lines represent the cover relations. For each dimension we have distinguished the join-irreducible elements (in blue) from the meet-irreducible ones (in green), and the elements which are neither join-irreducible nor meet-irreducible are uncolored. Notice that the presence of meet-irreducible elements that are not coatoms immediately implies that, unlike $\lat{KC}^3$, $\lat{KC}^4$ is not coatomistic. 
On the other hand, one can verify that as in the $\N=3$ case, all join-irreducible elements are principal down-sets in the MI-poset. 

The solid dot is the $5$-dimensional PMI obtained by taking the join of all atoms, and notice that it is incomparable with all the meet-irreducible elements that are not coatoms. The dashed line in red highlights that there exist $3$-dimensional PMIs that cover $5$-dimensional ones. This in particular implies that $\lat{KC}^4$ does not satisfy the JDCC, and in turn, that it is neither lower nor upper-semimodular. The fact that $\lat{KC}^4$ is not a sublattice of $\lat{PMI}^4$ can be verified using the knowledge of the full $\lat{PMI}^4$. It suffices to check that there is at least one pair of PMIs in $\lat{KC}^4$ whose join in $\lat{PMI}^4$ is not an element of $\lat{KC}^4$. Indeed, the fact that $\lat{KC}^4$ is not meet-semidistributive can also be verified explicitly using \Cref{lem:meetsemidis}. Since these computations are not particularly illuminating, and the explicit description of the counterexamples is complicated by the large cardinality of the sets involved, we omit the details.

\subsection{General properties of \texorpdfstring{$\lat{KC}^\N$}{L-KC} for arbitrary \texorpdfstring{$\N$}{N}}
\label{ssec:latticN}

Having analyzed the structure of $\lat{KC}^3$ and $\lat{KC}^4$ in detail, we now turn to the general case. We start by proving a general result regarding a specific embedding of the $\N$-party lattice into an $\N'$-party lattice for $\N'>\N$. We then combine this result with the previous observations about $\N=3$ and $\N=4$ into theorems that summarize the main structural properties of $\lat{KC}^\N$ for an arbitrary number of parties, and list a few technical questions that remain open, which we leave for future work. 

As we will further explore in the next subsection, given any $\N$ and $\N'>\N$, there are several sublattices of $\lat{KC}^{\N'}$ which are isomorphic to $\lat{KC}^{\N}$. Intuitively, this has to do with the fact that when we specify a PMI for $\N$ parties, this does not automatically specify which MI instances involving the additional parties should vanish, at least not completely. Here we focus on one such sublattice that will be particularly useful for the proof of \Cref{thm:main_theorem} below, and for more general constructions see \S\ref{subsec:constructions}. 

The simple construction that we describe here is akin to the one in quantum mechanics, where given an $\N$-party state we simply add an ancillary system with $\N'-\N$ parties in a pure state, the new parties being both completely uncorrelated with the original system and among themselves. Notice however that while this quantum mechanical intuition is useful to formulate the construction, we are working with PMIs that a priori might be unrealizable, rather than with genuine quantum states.

As usual we denote by $\uI,\uJ,\uK,\ldots$ the subsets of $\psys{\N}$, and for the new parties we introduce the notation
\begin{equation}
    \Delta = [\N']\setminus[\N]\,,\qquad \Il,\Jl,\Kl,\dots\subseteq \Delta\,.
\end{equation}
A generic subset of $\psys{\N'}$ that contains some of the original parties in $[\N]$ (including the purifier), as well as new parties in $\Delta$, will be written as $\uI\Il=\uI\cup\Il$.\footnote{\, In general, the purifier $0$ in the larger $(\N+1)$-party system can be different from that in the original $\N$-party system. Notice that we do not underline the new indices, labeling the purifier of the entire system only with the indices that pertain to the $\N$-party system. This choice will be convenient below and intuitively it is justified by the fact that the new parties are all uncorrelated and in a pure state.}

Given a PMI $\pmi$ in $\lat{KC}^\N$, we would like to map it to a corresponding PMI $\pmi'$ in $\lat{KC}^{\N'}$ by specifying its $\zds$. First, we demand that all MI instances where one of the arguments only involves new parties must vanish, i.e.,
\begin{align}\label{eq:newPMI1}
    \mi(\Il:\Kl)' = 0 \,, \quad \mi(\Il:\uK)' = 0 \,, \quad \mi(\Il:\Kl\uK)' = 0\,,
\end{align}
where $\mi(\cdot:\cdot)'$ indicates that the MI instance is in the $\N'$-party system (as opposed to the $\N$-party system). Notice that the above equations do not depend on the choice of the PMI involving $\N$ parties. They implement the independence (and purity) of the new parties that we mentioned above, and their implications can be understood geometrically as follows.

Consider the linear subspace of $\mathbb R^{\D'}$ specified by \eqref{eq:newPMI1} and notice in particular that 
\begin{equation}  \label{eq:basis1}
    \mi(\Il:\Il^c)' = 0\quad \implies\quad \ent_{\Il}=0\qquad \forall\, \Il\,,
\end{equation}
where $\Il^c \equiv \psys{\N'}\setminus \Il$ is the complement of $\Il$ in the $\N'$-party system,\footnote{\,  Although $\Il^c$ includes the purifier, we do not underline the index for the sake of notational simplicity.} and we used \eqref{eq:trivial_ins}. From \eqref{eq:basis1}, it follows that the second condition $\mi(\Il:\uK)' = 0$ in \eqref{eq:newPMI1} becomes
\begin{equation} \label{eq:basis2}
    \mi(\Il:\uK)' = 0\quad \implies \quad \ent_{\Il\uK}=\ent_{\uK}\qquad \forall\, \Il,\uK\,.
\end{equation}
Notice that the first condition, as well as the third condition (for $\Kl\uK\subset \Il^c$), in \eqref{eq:newPMI1} are automatically implied by \eqref{eq:basis1} and \eqref{eq:basis2}, and that together these relations specify an embedding of $\mathbb{R}^{\D}$ inside $\mathbb{R}^{\D'}$. 

We can now use these relations to rewrite the instances of SA in the $\N'$-party system that are not yet specified in the $\N$-party system. These are all the SA instances of the form $\mi(\uI:\Kl\uK)' \geq 0$ and $\mi(\Il\uI:\Kl\uK)' \geq 0$. We compute 
\begin{align} \label{eq:sa_reduction}
   \begin{split}
       \mi(\Il\uI:\Kl\uK)' &= \ent_{\Il\uI} + \ent_{\Kl\uK} - \ent_{\Il\uI\Kl\uK} \\
        &= \ent_{\uI} + \ent_{\uK} - \ent_{\uI\uK}  \\
       &= \mi(\uI:\uK) \geq 0 \,,
   \end{split}
\end{align}
where we used \eqref{eq:basis2} in the second equality. It is similarly straightforward to verify that
\begin{align}\label{eq:sa_reduction2}
    \mi(\uI:\Kl\uK)' = \mi(\uI:\uK) \geq 0\,.
\end{align}
The meaning of \eqref{eq:sa_reduction} and \eqref{eq:sa_reduction2} is that on the subspace of $\mathbb{R}^{\D'}$ specified by \eqref{eq:basis1} and \eqref{eq:basis2} (which is isomorphic to $\mathbb{R}^{\D}$), the new instances of SA all reduce to the set of instances of SA for the $\N$-party system. These inequalities therefore specify a cone which is simply an embedding of the SAC$_{\N}$ in $\mathbb{R}^{\D'}$ as a face of the SAC$_{\N'}$. 

Given an $\N$-party PMI $\pmi$ corresponding to the face $\face=\mu^{-1}(\pmi)$ of the SAC$_{\N}$, we can then map it to the $\N'$-party PMI $\pmi'=\mu(\face')$ specified by the analogous face $\face'$ of the SAC$_{\N}$ embedded in SAC$_{\N'}$.\footnote{\, With a slight abuse of notation, we denoted the map between faces of SAC and PMIs by $\mu$ in both spaces.} This is achieved by including in the $\zds$ of $\pmi'$, besides the vanishing MI instances specified in \eqref{eq:newPMI1}, also the instances obtained via the following implication due to \eqref{eq:sa_reduction} and \eqref{eq:sa_reduction2}:
\begin{align}\label{eq:newPMI4}
    \mi(\uI: \uK) = 0  &\quad\implies\quad \mi(\uI: \Kl\uK)' = 0\,, \quad \mi(\Il\uI: \Kl\uK)' = 0 \,.
\end{align}
Notice that this construction is independent from KC, and we can map any $\N$-party PMI to a new $\N'$-party one in this fashion. However, it is immediately clear from \eqref{eq:newPMI1} and \eqref{eq:newPMI4} that KC-compatible PMIs are mapped to KC-compatible PMIs.

We can view the construction we just described as a map $\phi$ that maps PMIs from an $\N$-party system to an $\N'$-party system with $\N' > \N$, i.e.,
\begin{align} \label{eq:embedding_map}
\begin{split}
    \phi:\;\; \lat{KC}^{\N}\; &\hookrightarrow\; \lat{KC}^{\N'}\\
    \pmi\;\; &\mapsto\;\; \phi(\pmi)\,.
\end{split}
\end{align}
This map has particularly nice properties, since it is a ``cover-preserving embedding'' of $\lat{KC}^{\N}$ into $\lat{KC}^{\N'}$. We will prove this in \Cref{lem:canonical_embedding} below, but let us first briefly clarify what we mean by this terminology. Given two lattices $\lat{}$ and $\lat{}'$, an \textit{embedding} $\phi$ of $\lat{}$ into $\lat{}'$ is a one-to-one lattice homomorphism from $\lat{}$ to $\lat{}'$, and we will say that $\phi$ is \textit{cover-preserving}, if 
\begin{equation}\label{eq:cover-preserving}
    \phi(x)\covb\phi(y)\quad \iff\quad x\covb y\qquad \forall\, x,y\in \lat{}\,.
\end{equation}
With these definitions, we then have the following result.

\begin{lemma}
\label{lem:canonical_embedding}
    The map \eqref{eq:embedding_map} is a cover-preserving embedding.
\end{lemma}
\begin{proof}
    From the geometric description discussed above, it is clear that the image $\phi(\lat{KC}^{\N})$ is isomorphic to $\lat{KC}^{\N}$, implying that $\phi$ is an embedding.
    
   To prove that $\phi$ is cover-preserving, we need to prove both implications in \eqref{eq:cover-preserving}. For the forward implication, suppose there are two PMIs $\pmi_1,\pmi_2 \in \lat{KC}^\N$ such that $\phi(\pmi_1) \covb \phi(\pmi_2)$. As $\phi(\lat{KC}^\N)$ is isomorphic to $\lat{KC}^\N$, this immediately also implies $\pmi_1 \covb \pmi_2$. 
   For the reverse implication, suppose that $\phi$ is \emph{not} cover-preserving, and consider two PMIs $\pmi_1,\pmi_2\in\lat{KC}^{\N}$ such that $\pmi_1\covb\pmi_2$, but there exists $\pmi_{\!*}\in\lat{KC}^{\N'}$ such that $\phi(\pmi_1)\prec\pmi_{\!*}\prec\phi(\pmi_2)$. Notice that we must have $\pmi_{\!*}\notin\phi(\lat{KC}^{\N})$, since otherwise $\pmi_1 \not\covb \pmi_2$. Let $\mgs^0_{1,2}$ and $\mgs^0_{\!*}$ be the respective $\zdss$ of $\phi(\pmi_{1,2})$ and $\pmi_{\!*}$. Since $\phi(\pmi_1)\prec\pmi_{\!*}$, we have $\mgs^0_1 \subset \mgs^0_{\!*}$,\footnote{\, By $\mgs^0_1 \subset \mgs^0_*$ we mean that $\mgs^0_1$ is a \textit{strict} subset of $\mgs^0_*$, that is, $\mgs^0_1 \subsetneq \mgs^0_*$.} 
which means $\mgs^0_{\!*}$ includes all the MI instances in \eqref{eq:newPMI1}. However, by \eqref{eq:sa_reduction} and \eqref{eq:sa_reduction2}, this in turn implies that $\pmi_{\!*}$ obeys \eqref{eq:newPMI4} and is therefore an element of $\phi(\lat{KC}^{\N})$, contradicting our above assumption and completing the proof.
\end{proof}

Having shown that the map \eqref{eq:embedding_map} is an embedding in the lattice sense, we then introduce the following terminology.

\begin{defi}[Canonical embedding of a PMI]
For $\N' > \N$, the canonical embedding of a \emph{PMI} $\pmi\in\lat{\emph{KC}}^{\N}$ with $\zds$ $\mgs^0$ into $\lat{\emph{KC}}^{\N'}$ is the \emph{PMI} whose $\zds$ is obtained by appending to $\mgs^0$ the vanishing \emph{MI} instances given in \eqref{eq:newPMI1} as well as those given by the implication \eqref{eq:newPMI4}.
\end{defi}

Notice that in the definition of the canonical embedding, we have chosen \textit{not} to permute the parties (for example, $1\in\psys{\N}$ is mapped to $1\in\psys{\N'}$, etc.). By including such permutations, we can obtain several new embeddings of $\lat{KC}^{\N}$ into $\lat{KC}^{\N'}$; we will come back to this point in \S\ref{subsec:constructions}. 

We can now use \Cref{lem:canonical_embedding}, and various observations that we made for $\lat{KC}^3$ and $\lat{KC}^4$ in the previous subsections, to prove several properties of $\lat{KC}^\N$ for an arbitrary number of parties. These are summarized in the following theorem.

\begin{thm} \label{thm:main_theorem}
    For any number of parties $\N\geq 4$, the lattice $\lat{\emph{KC}}^\N$ of \emph{KC}-compatible \emph{PMIs} has the following properties: 
\begin{enumerate}[label={\footnotesize \emph{(}\roman*\emph{)}}, itemsep=0pt]
    \item It has the same meet as both $\lat{\emph{PMI}}^\N$ and $\lat{\emph{MID}}^\N$, but is not a sublattice of either.
    \item Each atom is a \emph{PMI} with a single vanishing \emph{MI} instance that is a minimal element of the \emph{MI}-poset (i.e., an instance involving only single parties).
    \item It is neither atomistic nor coatomistic.
    \item It does not satisfy the Jordan-Dedekind chain condition.\footnote{\, Therefore it is also not upper or lower semimodular, modular, or distributive.}
    \item It is neither join-semidistributive nor meet-semidistributive.
\end{enumerate}
\end{thm}
\begin{proof}
For (i), we have already seen that $\lat{KC}^\N$ is a lattice with the same meet as $\lat{PMI}^\N$ (cf.\ \Cref{thm:KC-PMIlattice}) and $\lat{MID}^\N$ (cf.\ \eqref{eq:lattice_intersection}). We have also seen in \S\ref{ssec:lattice3} that $\lat{KC}^3$ is not a sublattice of $\lat{MID}^3$ since they have different joins, and likewise in \S\ref{ssec:lattic4} that $\lat{KC}^4$ is not a sublattice of $\lat{PMI}^4$. Now for $\N \geq 4$, consider the canonical embedding of two PMIs $\pmi_1,\pmi_2\in\lat{KC}^4$, whose join is different from their join in $\lat{PMI}^4$, into $\lat{KC}^{\N}$. Letting $\pmi_3 = \pmi_1 \vee_{\text{kc}} \pmi_2$, it follows $\phi(\pmi_1) \vee_{\text{kc}} \phi(\pmi_2) = \phi(\pmi_3)$. Furthermore, it is clear that the map $\phi$ is also an embedding of $\lat{PMI}^4$ in $\lat{PMI}^{\N}$, implying that
\begin{align}
    \pmi_1 \vee \pmi_2 = \pmi_3' \not= \pmi_3 \quad\implies\quad \phi(\pmi_1) \vee \phi(\pmi_2) = \phi(\pmi_3') \not= \phi(\pmi_3)\,,
\end{align}
so the join of $\phi(\pmi_1)$ and $\phi(\pmi_2)$ in $\lat{PMI}^\N$ is different from their join in $\lat{KC}^\N$. By a similar argument, $\lat{KC}^\N$ is not a sublattice of $\lat{MID}^\N$ either.

For (ii), notice that the down-sets generated by the minimal elements of the MI-poset are precisely the down-sets with cardinality one. Since their corresponding hyperplanes are PMIs, they are the atoms of $\lat{KC}^\N$. 
    
For (iii), consider the pure quantum state $|\!\pt~\!\!\!\rangle_{1234} \otimes |0\rangle_5 \otimes \cdots \otimes |0\rangle_{\N}$. Its PMI $\pmi$ is greater than or equal to the join of all the atoms, since all the MIs between two single parties vanish. As $\pmi$ is clearly not the top element (which is the state with every MI instance vanishing), this means the top element of $\lat{KC}^{\N}$ is not the join of all the atoms, so the lattice is not atomistic.

Next, to prove that $\lat{KC}^{\N}$ is not coatomistic, recall that we showed in \S\ref{ssec:lattic4} that $\lat{KC}^4$ is not coatomistic, so there exists a PMI $\pmi\in\lat{KC}^4$ that is not the meet of a collection of coatoms in $\lat{KC}^4$. Now consider its canonical embedding $\phi(\pmi)\in\lat{KC}^{\N}$. Notice that any coatom $\overline{\pmi}$ of $\lat{KC}^{\N}$ greater than $\phi(\pmi)$ must obey \eqref{eq:newPMI1} and hence is part of the canonical embedding of the lattice $\lat{KC}^4$, and $\phi(\pmi)$ by assumption cannot be the meet of such coatoms. This proves $\phi(\pmi)$ is not the meet of an arbitrary collection of coatoms of $\lat{KC}^{\N}$, and so $\lat{KC}^{\N}$ is not coatomistic.

For (iv), we know from \S\ref{ssec:lattic4} that $\lat{KC}^4$ does not satisfy the JDCC, and since the canonical embedding is cover-preserving (cf.\ \Cref{lem:canonical_embedding}), it follows that $\lat{KC}^\N$ does not either.

For (v), we have seen in \S\ref{ssec:lattice3} that $\lat{KC}^3$ is not join-semidistributive because it contains the sublattice $\mathfrak{S}_7$ (cf.\ \Cref{fig:not-joinsemi}). This automatically implies that $\lat{KC}^\N$ is likewise not join-semidistributive for any $\N > 3$ since it also contains the same sublattice. Furthermore, in \S\ref{ssec:lattic4} we used \Cref{lem:meetsemidis} to verify that $\lat{KC}^4$ is not meet-semidistributive. As for join-semidistributivity, it is a general result of lattice theory that a lattice is meet-semidistributive if and only if it does not contain a sublattice isomorphic to one of several forbidden lattices (see for example \cite{gratzer2016}). Since $\lat{KC}^4$ is not meet-semidistributive, it contains such a sublattice, and therefore so does $\lat{KC}^{\N}$ for any $\N\geq 4$, which means $\lat{KC}^\N$ is also not meet-semidistributive.
\end{proof}

Despite the fact that \Cref{thm:main_theorem} above proves that $\lat{KC}^\N$ does not belong to any of the widely studied classes of lattices, we prove below two interesting properties of $\lat{KC}^\N$.

\begin{thm} \label{thm:JIE}
    The set of join-irreducible elements of $\lat{\emph{KC}}^\N$ includes all those of $\lat{\emph{MID}}^\N$ that are not principal down-sets of maximal elements of the \emph{MI}-poset.
\end{thm}
\begin{proof}
   Consider a non-maximal element of the MI-poset $\mi(\uI:\uK)$ (so $\uK \not= \uI^c$), and let $\mgs^0$ be its principal down-set. Because the join of two elements in $\lat{MID}^\N$ is given by their union, any principal down-set is a join-irreducible element of $\lat{MID}^{\N}$. We now want to show that $\mgs^0$ is the $\zds$ of a PMI $\pmi$, and that $\pmi$ is also join-irreducible in $\lat{KC}^{\N}$.

To show that $\mgs^0$ is the $\zds$ of a PMI $\pmi$, consider the quantum state
\begin{equation}\label{eq:quantum_state}
    \ket{\pmi}=\bigotimes_{\forall\{i,k\}\,\notin\,\Gamma} \ket{\bpg}_{ik}\,,
\end{equation}
where $\ket{\bpg}_{ik}$ is a Bell pair between parties $i$ and $k$, and $\Gamma$ is the set
\begin{equation}
    \Gamma=\{\{i,k\}:i\in\uI,k\in\uK\}\,.
\end{equation}
The $\zds$ of the PMI of $\ket{\pmi}$ is precisely $\mgs^0$, and since $\ket{\pmi}$ is a quantum state, it automatically follows that $\pmi$ is in $\lat{KC}^\N$. 

We now want to show that $\pmi$ is join-irreducible in $\lat{KC}^\N$. Let $\mgs^0_*$ be the unique element in $\lat{MID}^{\N}$ such that $\mgs^0\cov \mgs^0_*$, which is obtained from $\mgs^0$ by simply removing the generating instance $\mi(\uI:\uK)$. Notice that since $\lat{KC}^{\N}$ is a subset of $\lat{MID}^{\N}$ with the induced partial order, to show that $\pmi$ is join-irreducible in $\lat{KC}^{\N}$ it suffices to show that $\mgs^0_*$ is the $\zds$ of a PMI $\pmi_*$ that is also in $\lat{KC}^{\N}$. If $\mgs^0_*$ is the $\zds$ of a PMI, this PMI is in $\lat{KC}^{\N}$ by construction, so suppose that $\mgs^0_*$ is not the $\zds$ of a PMI. Then it must be by linear dependence and SA that the MI instances in $\mgs^0_*$ vanishing implies that additional MI instances must also vanish. However, if any instance other than $\mi(\uI:\uK)$ must be added to $\mgs^0_*$ to obtain a PMI, it must also be added to $\mgs^0$, contradicting the fact that $\mgs^0$ is already a PMI. The only remaining possibility then is that linear dependence and SA imply that $\mi(\uI:\uK)$, and nothing else, must be added to $\mgs^0_*$. We now proceed to show that this is also not possible.

First notice that every MI instance in $\mgs^0$ is of the form $\mi(\uI':\uK')$ for $\uI' \subseteq \uI$ and $\uK' \subseteq \uK$. Thus, if we write each MI instance $\mi(\uI':\uK')$ in $\mgs^0$ in terms of the entropies, they each contain a term that is the entropy $\ent_{\uI'\uK'}$, which is not present in any of the other MI instances (notice that this is true only if $\mi(\uI:\uK)$ is not a maximal element of the MI-poset, since otherwise $\ent_{\uI\uK} = 0$). This means there cannot exist a set of MI instances in $\mgs^0$ that are linearly dependent, and in particular that $\mi(\uI:\uK)$ cannot be required to vanish by the fact that all MI instances in $\mgs^0_*$ vanish.

The only remaining possibility is that the vanishing of $\mi(\uI:\uK)$ is required by SA. To clarify what we mean by this statement, denote by $\mathbb{S}_*$ the linear subspace of entropy space that is the intersection of all the hyperplanes corresponding to the vanishing MI instances in $\mgs^0_*$. If $\mathbb{S}_*$ is spanned by a face of the SAC, it is a PMI, so the possibility we are contemplating is that this is not the case, but rather that the largest PMI contained in $\mathbb{S}_*$ is a proper subspace obtained by imposing the SA instance $\mi(\uI:\uK) \geq 0$. However, imposing $\mi(\uI:\uK) \geq 0$ simply partitions $\mathbb S_*$ into two codimension-0 regions in $\mathbb S_*$. Therefore, to restrict the solution of the system of inequalities corresponding to the instances of SA which are not in $\mgs^0_*$ to a lower-dimensional subspace of $\mathbb S_*$, one needs at least two additional linear inequalities to be saturated. This contradicts our above hypothesis that $\mi(\uI:\uK)$ is the only additional SA that must be saturated, and thus completes the proof that $\mgs^0_*$ is the $\zds$ of a PMI $\pmi_*$, and that $\pmi$ is join-irreducible in $\lat{KC}$.

Finally, consider the case where $\mi(\uI:\uK)$ is a maximal element of the MI poset (so $\uK = \uI^c$). Its principal down-set $\mgs^0$ still corresponds to a PMI $\pmi$, since the construction \eqref{eq:quantum_state} gives a quantum state whose PMI has $\mgs^0$ as its $\zds$. However, while $\mgs^0$ is still a join-irreducible element of $\lat{MID}^{\N}$, $\pmi$ is no longer join-irreducible in $\lat{KC}^{\N}$. To prove this, first notice that the down-set $\mgs^0_*$ obtained from $\mgs^0$ by removing the top element $\mi(\uI:\uK)$ is not the vanishing down-set of a PMI. This follows from the linear dependence relation
\begin{equation} \label{eq:triangle}
    \mi(\uI:\uK_1)+\mi(\uI:\uK_2)- \,\mi(\uI:\uK)=0 \,,
\end{equation}
where $\{\uK_1,\uK_2\}$ is an arbitrary bipartition of $\uK$. Because the first two MI instances in \eqref{eq:triangle} are in $\mgs^0_*$ but not the third, the down-set $\mgs^0_*$ is not the vanishing down-set of a PMI. In the MI-poset, let $\mathcal{C}$ denote the set of MI instances covered by $\mi(\uI:\uK)$, and let $\mgs^0_1,\ldots, \mgs^0_{|\mathcal C|}$ be the principal down-sets generated by the elements of $\mathcal C$. As we showed above, the principal down-sets $\mgs^0_i$ all correspond to the vanishing down-sets of PMIs in $\lat{KC}^\N$. The join of all these PMIs is a new PMI whose $\zds$ is the smallest down-set that corresponds to a PMI and also includes the union of all the down-sets $\mgs^0_i$. However, notice that the union of these down-sets is precisely $\mgs^0_*$, so the smallest down-set that both contains $\mgs^0_*$ and corresponds to a PMI is the principal down-set $\mgs^0$, since by \eqref{eq:triangle} we must at least include $\mi(\uI:\uK)$ in the down-set. This proves that $\pmi$, whose vanishing down-set is the principal down-set of $\mi(\uI:\uK)$, can be obtained as the join of a collection of PMIs in $\lat{KC}^\N$ and is thus not join-irreducible.
\end{proof}

In \Cref{thm:JIE}, we determined that the set of join-irreducible elements of $\lat{KC}^\N$ contains a specific subset of those in $\lat{MID}^\N$. Naturally, it would be interesting to know if these are all the join-irreducible elements of $\lat{KC}^\N$. If this is true, then the set of join-irreducible elements would be isomorphic to the subposet of the MI-poset where the maximal elements have been removed. We checked explicitly that this is true for $\N=3$ and $\N=4$, and we phrase this formally as a question for arbitrary $\N$ below.

\begin{question} \label{question:question1}
Is the set of join-irreducible elements of $\lat{\emph{KC}}$, with the induced partial order, isomorphic to the subposet of the \emph{MI}-poset obtained by removing the maximal elements? 
\end{question}

We conclude this subsection with an additional result about the order relation between atoms and coatoms of $\lat{KC}^\N$ and a few comments about its usefulness for explicit computations.

\begin{thm} \label{thm:new_coatoms}
    Given any coatom $\overline{\pmi}$ that is not a permutation of the canonical embedding of an $\N=2$ \emph{PMI} realized by a Bell pair, and any atom $\underline{\pmi}$, we have $\overline{\pmi} \succ \underline{\pmi}$. Moreover, any coatom that corresponds to $\bp{i}{j}$, with $i,j\in\psys{\N}$, is comparable with every atom except for the one with $\zds$ $\{\mi(i:j)\}$.
\end{thm}
\begin{proof}
   Let us denote by $\overline{\pmi}_{ij}$ the PMI of a quantum state which is the tensor product of a Bell pair between parties $i,j\in\psys{\N}$ and arbitrary pure, uncorrelated states for the remaining $\N-1$ parties. From the point of view of the MI-poset, the down-set $\mgs^0_{ij}$ of vanishing instances of $\overline{\pmi}_{ij}$ is the complement of the \textit{principal up-set}\footnote{\, Similarly to principal down-sets, these are up-sets generated by a single MI instance.} of $\mi(i:j)$, which means $\mgs^0_{ij}$ contains all the atoms except the one where $\mi(i:j) = 0$. This proves the second claim in the theorem. In addition, any PMI $\pmi\in\lat{KC}^{\N}$ such that $\mgs^0\subset\mgs^0_{ij}$ cannot be a coatom, and by permutations, the same implication obviously holds for any $i,j$. Thus, for any coatom $\overline{\pmi}$ which is not of the form $\overline{\pmi}_{ij}$ for some $i,j$, its vanishing set of MI instances $\mgs^0$ must include all atoms, and so the first claim is also proven.
\end{proof}

It is clear that any $\N$-party $1$-dimensional KC-compatible PMI is a coatom of $\lat{KC}^\N$, and \Cref{thm:new_coatoms} immediately implies the following interesting result about the extreme rays of the SAC that can possibly be realized by quantum states.

\begin{cor}
\label{cor:er_corollary}
For any number of parties $\N$, all extreme rays of the \emph{SAC}$_{\N}$ that satisfy SSA and are not realized by Bell pairs are localized on the face $\face$ spanning the codimension-$\binom{\N+1}{2}$ subspace of entropy space given by
\begin{equation}
\label{eq:atoms_sub}
    \mi(i:j)=0, \quad \forall\, i,j\in \psys{\N} \,.
\end{equation}
\end{cor}
\begin{proof}
    Since all extreme rays of the SAC$_\N$ satisfying SSA generate 1-dimensional KC-compatible PMIs, which are coatoms of $\lat{KC}^\N$, it follows immediately from \Cref{thm:main_theorem}(ii) and \Cref{thm:new_coatoms} that they are contained in the subspace given by \eqref{eq:atoms_sub}. To show that there is a face of the SAC$_\N$ that spans this subspace, it suffices to find an entropy vector that satisfies the condition \eqref{eq:atoms_sub} with all other MI instance are strictly positive. We leave it as an exercise for the reader to verify that this is indeed the case for $\N=3$, where $\face$ is simply a positive rescaling of $\pt$. For any $\N>3$, denote by $\vec{\ent}$ the sum of all entropy vectors obtained via the canonical embedding of $\pt$ and all permutations of the parties. It is then a straightforward calculation to verify that the $\zds$ of $\pi(\vec{\ent})$, which is the intersection of the $\zds$s of the PMIs of all these vectors (cf.~\eqref{eq:PMImeet}), is precisely the set of MI instances in \eqref{eq:atoms_sub}. 
\end{proof}

\Cref{cor:er_corollary} is particularly useful for computations, since it reduces the search for the KC-compatible extreme rays of the SAC$_{\N}$ to the search of extreme rays of the codimension-$\binom{\N+1}{2}$ face defined in \Cref{cor:er_corollary}, which is a simpler polyhedral cone. As mentioned in the proof, the reader can verify that for $\N=3$ the face $\face$ of the SAC$_3$ defined in  \Cref{cor:er_corollary} is 1-dimensional, and is in fact the extreme ray realized by $|\!\pt~\!\!\!\rangle_{1234}$. We will come back to the application of \Cref{cor:er_corollary} in \S\ref{sec:realizability}, where we comment about the relation between KC and SSA, and leave the exploration of possible generalizations of this result for future work \cite{He:2022wip}. 

Lastly, note that although every 1-dimensional KC-compatible PMI is a coatom of $\lat{KC}^\N$, it is not clear that the converse is true. We observe via direct computation that the coatoms for $\lat{KC}^\N$ with $\N=3,4$ are indeed all 1-dimensional PMIs. It is then intriguing to speculate that this is true for arbitrary $\N$, and we pose the following question, which we leave for future investigation.

\begin{question} \label{question:question2}
Are all coatoms of $\lat{\emph{KC}}^\N$
\emph{1}-dimensional \emph{PMI}s?
\end{question}

\subsection{Constructing \texorpdfstring{$\N$}{N}-party PMIs from fewer parties}
\label{subsec:constructions}

In the previous subsection, we have seen how given the lattice $\lat{KC}^{\N}$ (for any $\N$) we can embed it inside the bigger lattice $\lat{KC}^{\N'}$ for any $\N'>\N$. This embedding has particularly nice properties (cf.\ \Cref{lem:canonical_embedding}) that will be useful for the proof of \Cref{thm:embedding}, but it is highly non-unique. The goal of this subsection is to discuss the multitude of different possible embeddings, and how they can be ``combined'' to construct a part of $\lat{KC}^{\N'}$. The main motivation for this analysis is that in the construction of $\lat{KC}^{\N'}$, we may assume complete knowledge of the lattices involving fewer parties, and we would then like to focus on the ``genuinely new'' elements.

As we already mentioned, the first obvious reason why the canonical embedding is non-unique has to do with permutations of the parties.  In \eqref{eq:newPMI1} and \eqref{eq:newPMI4}, each party labeled by an element of $\psys{\N}$ was given the same label in $\psys{\N'}$, but in general this does not have to be the case. Indeed, given $\lat{KC}^{\N}$ and $\N'>\N$, one can obtain different embeddings into $\lat{KC}^{\N'}$ by first constructing the canonical embedding, and then simply permuting the parties.

To describe more general embeddings, it will be convenient to use a standard lattice construction, which we will now briefly review. Given two lattices $\lat{1}$ and $\lat{2}$, their \textit{product} $\lat{}=\lat{1}\times\lat{2}$ is obtained by first constructing the Cartesian product of the two sets $\lat{1}$ and $\lat{2}$, and then specifying the following partial order:\footnote{\, The generalization to an arbitrary number of factors is obvious.} 
\begin{equation} \label{eq:prodorder}
    (x_1,y_1) \preceq (x_2,y_2) \quad \Longleftrightarrow \quad x_1\preceq x_2\;\; \text{and}\;\;
    y_1\preceq y_2\,.
\end{equation}

For PMIs that are realizable in quantum mechanics, we can intuitively think of this construction as follows. Suppose we have two quantum systems with $\N_1$ and $\N_2$ parties. As long as the two systems are independent, we are free to construct a quantum state for each one, and realize any quantum mechanical PMI in both $\lat{KC}^{\N_1}$ and $\lat{KC}^{\N_2}$. The tensor product of two such states will then realize a PMI in the larger lattice $\lat{KC}^{\N}$, with $\N=\N_1+\N_2$. Since the choices for the two systems are independent, we have an embedding of the subset of quantum mechanically realizable PMIs of $\lat{KC}^{\N_1}\times\lat{KC}^{\N_2}$ inside $\lat{KC}^{\N}$. 

While this construction is clear in quantum mechanics, not all KC-compatible PMIs are a priori realizable, and we should therefore reformulate it purely in the language of PMIs without relying on quantum states. We will soon discuss this reformulation precisely, but first it is instructive to consider the simple example $\lat{KC}^{1}\times\lat{KC}^{1}\hookrightarrow\lat{KC}^{2}$, where all KC-compatible PMIs are realizable by quantum states and we can conveniently use the construction mentioned above. 

\begin{figure}[tb]
\centering
\begin{subfigure}[b]{0.3\textwidth}
\centering
\begin{tikzpicture}
\node[draw,rounded corners] (top) at (0,2) {{\footnotesize $\mi(A:O)$}};
\end{tikzpicture}
\vspace{0.7cm}
\subcaption[]{}
\end{subfigure}
\begin{subfigure}[b]{0.3\textwidth}
\centering
\begin{tikzpicture}
\draw[dotted] (-1,-1) -- (2,2);
\draw[-] (-0.5,-0.5) -- (1.5,1.5);
\draw[->, gray, very thick] (0,0) -- (1,1) ;
\node[gray] (5a) at (0,0) {$\bullet$};
\node at (1.7,1.2) {{\footnotesize $\mathbb{R}^1$}};
\end{tikzpicture}
\subcaption[]{}
\end{subfigure}
\begin{subfigure}[b]{0.3\textwidth}
\centering
\begin{tikzpicture}
\node[draw,rounded corners] (top) at (0,2) {{\footnotesize $\{\mi(A:O)$\}}};
\node[draw,rounded corners] (bottom) at (0,0) {{\footnotesize $\varnothing$}};
\draw[-] (top.south) -- (bottom.north);
\end{tikzpicture}
\vspace{0.2cm}
\subcaption[]{}
\end{subfigure}
\caption{Various constructs for $\N=1$: (a) the MI poset, (b)  SAC$_1$, and (c) $\lat{KC}^1$.}
\label{fig:n1}
\end{figure}
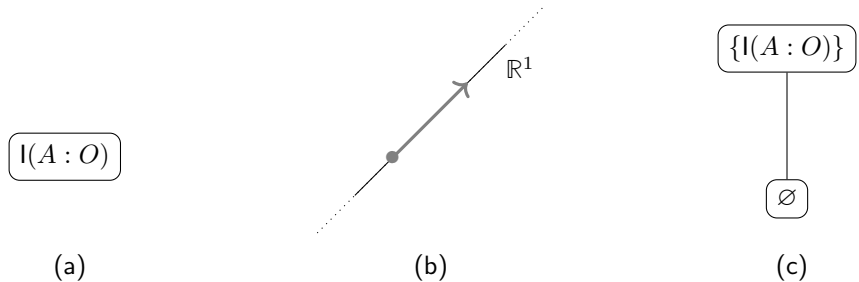

Let us first clarify how to construct the trivial $\lat{KC}^1$. For $\N=1$, the MI-poset contains the single instance $\mi(A:O)$, and SAC$_1$, which is a cone in $\mathbb{R}^1$, reduces to a single ray. There are only two faces in SAC$_1$: the ray, and the origin. Clearly the corresponding PMIs are KC-compatible and quantum mechanically realizable, and in fact they simply distinguish mixed from pure states (all the constructs for $\N=1$ are shown in \Cref{fig:n1}).

The lattice $\lat{KC}^{1}\times\lat{KC}^{1}$, is shown in \Cref{subfig:l1xl1}, where for each element we also show an explicit realization by a quantum state. Notice that these are 2-party states, and that each corresponds to a unique $\N=2$ PMI. In particular, all correlations between the parties $A$ and $B$ and the purifier $O$ are completely fixed by the specification of which subsystems are pure or mixed. The three possible embeddings of this lattice inside $\lat{KC}^{2}$, obtained by permuting the parties, are shown in \Cref{subfig:l1xl1inl2ab}, \ref{subfig:l1xl1inl2ao} and \ref{subfig:l1xl1inl2bo}.

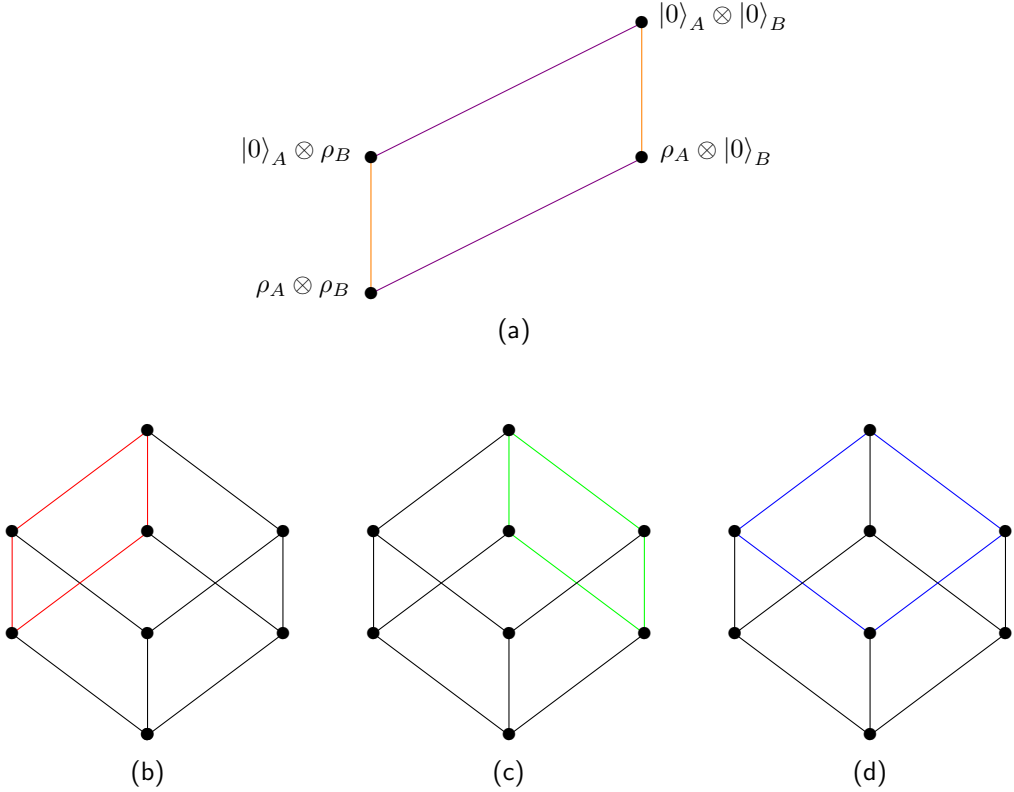
\begin{figure}[tb]
    \centering
    \begin{subfigure}{0.9\textwidth}
\centering
\begin{tikzpicture}[scale=0.9, every node/.style={inner sep=-1pt, outer sep=-1pt}]
    \node (1) at (-2,2) {};
    \node (2) at (-2,0) {};
    \node (1x) at (2,4) {};
    \node (2x) at (2,2) {};
    \draw[-,orange] (1) -- (2);
    \draw[-,orange] (1x) -- (2x);
    \draw[-, violet] (1) -- (1x);
    \draw[-, violet] (2) -- (2x);
    \node at (-3,0.1) {{\footnotesize $\rho_A\otimes\rho_B$}};
    \node at (-3.1,2.1) {{\footnotesize $\ket{0}_A\otimes\rho_B$}};
    \node at (3.1,2.1) {{\footnotesize $\rho_A\otimes\ket{0}_B$}};
    \node at (3.2,4.1) {{\footnotesize $\ket{0}_A\otimes\ket{0}_B$}};
    \node at (-2,2) {$\bullet$};
    \node at (-2,0) {$\bullet$};
    \node at (2,4) {$\bullet$};
    \node at (2,2) {$\bullet$};
    \end{tikzpicture}
\subcaption[]{}
\label{subfig:l1xl1}
\end{subfigure}   

\vspace{1cm}
\begin{subfigure}{0.3\textwidth}
\centering
\begin{tikzpicture}[scale=0.9, every node/.style={inner sep=-1pt, outer sep=-1pt}]
    \node (1) at (-3,4.5) {};
    \node (2a) at (-5,3) {};
    \node (2b) at (-3,3) {};
    \node (2c) at (-1,3) {};
    \node (3a) at (-5,1.5) {};
    \node (3b) at (-3,1.5) {};
    \node (3c) at (-1,1.5) {};
    \node (4) at (-3,0) {$\bullet$};
    \draw[-,red] (1) -- (2a);
    \draw[-,red] (1) -- (2b);
    \draw[-] (1) -- (2c);
    \draw[-,red] (2a) -- (3a);
    \draw[-] (2a) -- (3b);
    \draw[-,red] (2b) -- (3a);
    \draw[-] (2b) -- (3c);
    \draw[-] (2c) -- (3b);
    \draw[-] (2c) -- (3c);
    \draw[-] (3a) -- (4);
    \draw[-] (3b) -- (4);
    \draw[-] (3c) -- (4);
    \node  at (-3,4.5) {$\bullet$};
    \node  at (-5,3) {$\bullet$};
    \node  at (-3,3) {$\bullet$};
    \node  at (-1,3) {$\bullet$};
    \node  at (-5,1.5) {$\bullet$};
    \node  at (-3,1.5) {$\bullet$};
    \node  at (-1,1.5) {$\bullet$};
    \node  at (-3,0) {$\bullet$};
\end{tikzpicture}
\subcaption[]{}
\label{subfig:l1xl1inl2ab}
\end{subfigure}
\begin{subfigure}{0.3\textwidth}
\centering
\begin{tikzpicture}[scale=0.9, every node/.style={inner sep=-1pt, outer sep=-1pt}]
    \draw[-] (1) -- (2a);
    \draw[-,green] (1) -- (2b);
    \draw[-,green] (1) -- (2c);
    \draw[-] (2a) -- (3a);
    \draw[-] (2a) -- (3b);
    \draw[-] (2b) -- (3a);
    \draw[-,green] (2b) -- (3c);
    \draw[-] (2c) -- (3b);
    \draw[-,green] (2c) -- (3c);
    \draw[-] (3a) -- (4);
    \draw[-] (3b) -- (4);
    \draw[-] (3c) -- (4);
    \node (1) at (-3,4.5) {$\bullet$};
    \node (2a) at (-5,3) {$\bullet$};
    \node (2b) at (-3,3) {$\bullet$};
    \node (2c) at (-1,3) {$\bullet$};
    \node (3a) at (-5,1.5) {$\bullet$};
    \node (3b) at (-3,1.5) {$\bullet$};
    \node (3c) at (-1,1.5) {$\bullet$};
    \node (4) at (-3,0) {$\bullet$};
\end{tikzpicture}
\subcaption[]{}
\label{subfig:l1xl1inl2ao}
\end{subfigure}
\begin{subfigure}{0.3\textwidth}
\centering
\begin{tikzpicture}[scale=0.9, every node/.style={inner sep=-1pt, outer sep=-1pt}]
    \draw[-,blue] (1) -- (2a);
    \draw[-] (1) -- (2b);
    \draw[-,blue] (1) -- (2c);
    \draw[-] (2a) -- (3a);
    \draw[-,blue] (2a) -- (3b);
    \draw[-] (2b) -- (3a);
    \draw[-] (2b) -- (3c);
    \draw[-,blue] (2c) -- (3b);
    \draw[-] (2c) -- (3c);
    \draw[-] (3a) -- (4);
    \draw[-] (3b) -- (4);
    \draw[-] (3c) -- (4);
    \node (1) at (-3,4.5) {$\bullet$};
    \node (2a) at (-5,3) {$\bullet$};
    \node (2b) at (-3,3) {$\bullet$};
    \node (2c) at (-1,3) {$\bullet$};
    \node (3a) at (-5,1.5) {$\bullet$};
    \node (3b) at (-3,1.5) {$\bullet$};
    \node (3c) at (-1,1.5) {$\bullet$};
    \node (4) at (-3,0) {$\bullet$};
\end{tikzpicture}
\subcaption[]{}
\label{subfig:l1xl1inl2bo}
\end{subfigure}
\vspace{0.4cm}
\caption{The three possible embeddings of $\lat{KC}^{1}\times\lat{KC}^{1}$ into $\lat{KC}^{2}$. In (a), we give a detailed description of $\lat{KC}^{1}\times\lat{KC}^{1}$ (one factor in orange and one in violet), with an explicit realization of the element by quantum states (all density matrices are assumed to be mixed). Notice that the correlation between $A,B$ and the purifier $O$ (not shown) is completely fixed by the specification of which subsystems are pure and which are mixed. In (b), (c), and (d), we illustrate the three possible embeddings of $\lat{KC}^{1}\times\lat{KC}^{1}$ in $\lat{KC}^{2}$, which are described by \Cref{thm:embedding}.}
\label{fig:l1timesl1}
\end{figure}

For an arbitrary number of parties $\N$, the embedding of these product lattices inside $\lat{KC}^{\N}$ are described by the following result.

\begin{thm}
\label{thm:embedding}
    For any given $\N$, and $\uI_1,\uI_2\subset \psys{\N}$ with $\N_1=|\I_1|$ and $\N_2=|\I_2|$, such that $\uI_1\cap\uI_2=\varnothing$ and $\N_1+\N_2=\N$, the subset of $\lat{\emph{KC}}^{\N}$ of \emph{PMI}s whose $\zdss$ include the principal down-set of $\mi(\uI_1:\uI_2)'$, with the induced partial order, is isomorphic to $\lat{\emph{KC}}^{\N_1}\times\lat{\emph{KC}}^{\N_2}$.
\end{thm}
\begin{proof}
   It is sufficient to prove the theorem in the case of $\uI_1= [\N_1]$ and $ \uI_2=[\N]\setminus[\N_1]$  (where neither $\uI_1$ nor $\uI_2$ contain the purifier), since any other case can be reduced to it by a permutation of the parties. We will proceed in a similar fashion as in the description of the canonical embedding in \S\ref{ssec:latticN}. We denote by $\hat{\I},\hat{\J},\hat{\K}$ the subsets of $[\N_1]$, and by $\Il,\Jl,\Kl$ those of $[\N]\setminus[\N_1]$. As usual, we will continue the convention where the underlined version of these indices can include the purifier 0.
   
   For a PMI such that its $\zds$ includes the principal down-set of $\mi(\uI_1:\uI_2)$, we have
    \begin{equation} \label{eq:direct_sum_dec}
    \mi(\hat{\I}:\Kl)' = 0\quad \implies\quad \ent_{\hat{\I}\Kl}=\ent_{\hat{\I}}+\ent_{\Kl}\qquad \forall\, \hat{\I},\Kl\,.
    \end{equation}
    We now proceed to show that given \eqref{eq:direct_sum_dec}, all the MI instances in a PMI $\pmi$ in the $\N$-party system reduce to those in the smaller $\N_1$ and $\N_2$-party systems. Obviously, we have
    \begin{align}\label{eq:newPMI2}
        \mi(\uIh:\uKh)' = \mi(\uIh:\uKh) \geq 0\,, \qquad \mi(\uIl:\uKl)' = \mi(\uIl:\uKl) \geq 0\,, 
    \end{align}
    so it suffices to focus on the MI instances involving parties from both the $\N_1$-party system and $\N_2$-party system. Proceeding as in \eqref{eq:sa_reduction}, it is straightforward to evaluate using \eqref{eq:direct_sum_dec} the cases not involving the purifier $0$:
    \begin{align}
        \mi(\Ih:\Kh\Kl)' &= \mi(\Ih:\Kh) \geq 0 \label{eq:newPMI5a} \\
        \mi(\Il:\Kh\Kl)' &= \mi(\Il:\Kl) \geq 0 \label{eq:newPMI5b} \\
        \mi(\Ih\Il:\Kh\Kl)' &= \mi(\Ih:\Kh) + \mi(\Il:\Kl) \geq 0\,. \label{eq:newPMI5c}
    \end{align}
    The cases involving the purifier $0$ are a little more tricky. However, we will show that the MI instances in the larger $\N$-party system reduce in the same manner to MI instances in the smaller systems. It is convenient to use the notation $\Ihc \equiv [\N_1] \setminus \Ih$, which is a subset of parties in $[\N_1]$, and similarly $\Ilc \equiv ([\N] \setminus [\N_1]) \setminus \Il$, which is a subset of parties in $[\N] \setminus [\N_1]$.\footnote{\, Note that $\Ihc$ ($\Ilc$) is almost the complement of $\Ih$ ($\Il$) in the $\N_1$-party ($\N_2$-party) subsystem, except that we are ignoring the purifier.} It then follows\footnote{\, In the rest of the proof we denote by $0_1$ and $0_2$ the purifiers of the smaller subsystems.}
    \begin{align}\label{eq:newPMI6a}
    \begin{split}
        \mi(\Ih: 0\Kh\Kl)' &= \ent_{\Ih} + \ent_{0\Kh\Kl} - \ent_{0\Ih\Kh\Kl} \\
        &= \ent_{\Ih} + \ent_{\Khc\Klc} - \ent_{(\Ih\Kh)^\dag\Klc} \\
        &= \ent_{\Ih} + \ent_{\Khc} - \ent_{(\Ih\Kh)^\dag} \\
        &= \mi(\Ih:0_1\Kh) \geq 0\,,
    \end{split}
    \end{align}
    where we first used purification symmetry in the full $\N$-party system, then \eqref{eq:direct_sum_dec}, and finally applied purification once more but now in the smaller $\N_1$-party system. Applying identical arguments, we also derive
    \begin{align}\label{eq:newPMI6b}
    \begin{split}
        \mi(\Il:0\Kh\Kl)' &= \mi(\Il:0_2\Kl) \geq 0\,.
    \end{split}
    \end{align}
    The equations \eqref{eq:newPMI5a}, \eqref{eq:newPMI5b}, \eqref{eq:newPMI6a}, and \eqref{eq:newPMI6b} (note that \eqref{eq:newPMI5c} is redundant) demonstrate how some of the new MI instances of the $\N$-party system reduces to simply the MI instances of the $\N_1$-party and $\N_2$-party systems. There are still a few MI instances in the $\N$-party system that we have to check whether they can be written in terms of the $\N_1$ and $\N_2$-party MI instances. We demonstrate one such MI instance computation below: 
    \begin{align}\label{eq:redundant}
    \begin{split}
        \mi(0\Ih\Il:\Kh\Kl)' &= \ent_{0\Ih\Il} + \ent_{\Kh\Kl} - \ent_{0\Ih\Kh\Il\Kl} \\
        &= \ent_{\Ihc\Il} + \ent_{\Kh\Kl} - \ent_{(\Ih\Kh)^\dag \Il\Kl} \\
        &= \ent_{\Ihc} + \ent_{\Il} + \ent_{\Kh} + \ent_{\Kl} - \ent_{(\Ih\Kh)^\dag} - \ent_{\Il\Kl} \\
        &= \mi(0_1\Ih:\Kh) + \mi(\Il:\Kl) \geq 0\,,
    \end{split}
    \end{align}
    where we utilized purification symmetry (in both the larger $\N$-party system and in the smaller $\N_1$-party system) and \eqref{eq:direct_sum_dec}. We leave the reader to check that all other MI instances in the $\N$-party system similarly reduce to a positive sum of two MI instances involving the $\N_1$ and $\N_2$-party systems (analogous to \eqref{eq:newPMI5c}).
    
    Geometrically, \eqref{eq:direct_sum_dec} means that we have an embedding of $\mathbb{R}^{\D_1}\oplus\mathbb{R}^{\D_2}$ inside $\mathbb{R}^{\D}$, while the SA instances in \eqref{eq:newPMI2}, \eqref{eq:newPMI5a}, \eqref{eq:newPMI5b}, \eqref{eq:newPMI6a}, and \eqref{eq:newPMI6b} specify two cones on $\mathbb{R}^{\D_1}$ and $\mathbb{R}^{\D_2}$ which are isomorphic respectively to SAC$_{\N_1}$ and SAC$_{\N_2}$ (the SA instances in \eqref{eq:newPMI5c} and \eqref{eq:redundant} are redundant). Globally then, we have a cone in $\mathbb{R}^{\D_1}\oplus\mathbb{R}^{\D_2}$ which is the direct sum of SAC$_{\N_1}$ and SAC$_{\N_2}$. The set of faces of this cone is the Cartesian product of the sets of faces the two cones, and the partial order of this set is given by \eqref{eq:prodorder}, completing the proof.
\end{proof}

The theorem above describes the possible embeddings of  $\lat{KC}^{\N_1}\times\lat{KC}^{\N_2}$ inside $\lat{KC}^{\N}$ for $\N=\N_1+\N_2$. An alternative way of understanding it is that it reveals the structure of the set of embeddings of $\lat{KC}^{\N}$ inside $\lat{KC}^{\N'}$ for all $\N < \N'$, where each element of $\lat{KC}^{\N'-\N}$ corresponds to one such embedding. From this perspective, the canonical embedding of \Cref{lem:canonical_embedding} can be seen as corresponding to the embedding of $\lat{KC}^{\N}$ in $\lat{KC}^{\N'}$ specified by the top element of $\lat{KC}^{\N'-\N}$. A similar reasoning also clarifies what are the embeddings of $\lat{KC}^{\N_1}\times\lat{KC}^{\N_2}$ inside $\lat{KC}^{\N}$ for $\N>\N_1+\N_2$, since one can simply construct $\lat{KC}^{\N_1}\times\lat{KC}^{\N_2}\times\lat{KC}^{\N'}$, with $\N'=\N-\N_1-\N_2$. Indeed, starting from $\lat{KC}^{\N}$, one can apply \Cref{thm:embedding} repeatedly to obtain the embeddings of $\lat{KC}^{\N_1}\times\ldots\times\lat{KC}^{\N_{n}}$ inside $\lat{KC}^{\N}$, with $\sum_{i=1}^{n}\N_i=\N$.

We conclude this section with a comment about how \Cref{thm:embedding} can be utilized to make the derivation of $\lat{KC}^{\N}$ more efficient. Suppose that we know $\lat{KC}^{\N-1}$, and we want to construct $\lat{KC}^{\N}$. Using \Cref{thm:embedding} we can construct the possible embeddings of $\lat{KC}^{\N-1}\times\lat{KC}^{1}$ in $\lat{KC}^{\N}$, including all possible permutations of the parties. Notice that this construction automatically includes all the embeddings of $\lat{KC}^{\N'}$ in $\lat{KC}^{\N}$, with $\N'<\N-1$, since these lattices are already embedded in $\lat{KC}^{\N-1}$. We can then take the meet for all possible collections of the elements obtained via these embeddings to construct a large portion of $\lat{KC}^{\N}$. For example, in the $\N=2$ case, this procedure yields the entire $\lat{KC}^{2}$ starting from $\lat{KC}^{1}$, and in this sense there are no ``genuinely $2$-party'' PMIs. This can be seen in \Cref{fig:l1timesl1}.

It would be interesting to investigate possible generalizations of \Cref{thm:embedding}, and to analyze more systematically the necessary conditions that must be imposed on the $\zds$ of a PMI such that it cannot be obtained from fewer party PMIs via these type of constructions. We leave this problem for future work.

\section{Comments on the realizability of KC-compatible PMIs}
\label{sec:realizability}

In this short section, we make a few comments on the realizability of PMIs that are KC-compatible, as well as on the relation between KC and SSA. In \Cref{subsec:ssa_k_condition}, we defined realizability of PMIs for an arbitrary class of quantum states $\Omega$ (cf.\ \Cref{def:realizable_pmi}). Here we make a few more assumptions about $\Omega$. Specifically, we assume that for any $\N$, $\Omega$ includes the completely uncorrelated state 
\begin{equation}\label{eq:uncorr}
    \ket{0}_1\otimes\ket{0}_2\otimes\cdots\otimes\ket{0}_{\N} \,,
\end{equation}
and that it is \textit{closed under tensor products}, i.e., the following implication holds:\footnote{\, Notice that we are fixing the number of parties $\N$ even if the tensor product of the density matrices is an operator acting on a larger Hilbert space. This simply amounts to enlarging the dimension of the Hilbert space for each party in $[\N]$.}
\begin{equation}\label{eq:prod_clos}
    \forall\,\rho_{\N},\rho_{\N}'\in\Omega \quad\implies\quad \rho_{\N}\otimes\rho_{\N}'\in\Omega\,.
\end{equation}
With these (very mild) assumptions, we then have the following result. 
\begin{thm} \label{thm:realizability}
    For any given $\N$, the set $\lat{$\Omega$}^{\N}$ of \emph{PMI}s which are realizable by a class of quantum states $\Omega$ that include the state \eqref{eq:uncorr} and satisfy \eqref{eq:prod_clos} is a lattice, with the meet given by \eqref{eq:PMImeet}.
\end{thm}
\begin{proof}
   For any $\N$, the top element of $\lat{KC}^{\N}$ is realized in quantum mechanics by the fully uncorrelated state \eqref{eq:uncorr}, which by assumption belongs to $\Omega$. The set of PMIs realized by $\Omega$ is in general a subset of $\lat{KC}^{\N}$, which is finite, and by (the dual of) \Cref{thm:lattice_condition} it is then sufficient to show that there is a meet operation. To prove this, we use  \eqref{eq:prod_clos} to show that the PMI of the tensor product of two density matrices in $\Omega$ is the meet (in $\lat{KC}^{\N}$) of the PMIs of the individual density matrices, i.e.,
   \begin{equation}
   \label{eq:meet_tp}
       \pi\left(\,\vec{\ent}\,(\,\rho\otimes\rho'\,)\,\right)=\,\pi\left(\vec{\ent}(\rho)\right)\,\wedge\,\pi\left(\vec{\ent}(\rho')\right)\,,
   \end{equation}
   where we keep the $\N$ dependence implicit for notational clarity. To see this, notice that from the additivity of the von Neumann entropy for tensor products, namely
   \begin{equation}
       \vec{\ent}\,(\rho\otimes\rho')=\,\vec{\ent}(\rho)+\vec{\ent}(\rho')\,,
   \end{equation}
   we have
   \begin{equation}
       \mi(\uI:\uK)(\rho\otimes\rho') = \mi(\uI:\uK)(\rho) + \mi(\uI:\uK)(\rho')\,, 
   \end{equation}
   where the argument of $\mi(\uI:\uK)$ is the quantum state on which we are evaluating the MI. This implies by SA that
   \begin{equation}
       \mi(\uI:\uK)(\rho\otimes\rho')=0\quad \iff \quad \mi(\uI:\uK)(\rho)=0\;\; \text{and}\;\; \mi(\uI:\uK)(\rho')=0\,.
   \end{equation}
   Therefore, the $\zds$ for $\rho\otimes\rho'$, in terms of those for $\rho$ and $\rho'$, is given by

   \begin{equation}
       \mgs^0_{\rho\otimes\rho'} = \mgs^0_{\rho} \cap \mgs^0_{\rho'}\,,
   \end{equation}%
   which we know from \eqref{eq:PMImeet} and \Cref{thm:KC-PMIlattice} to be the meet in $\lat{KC}^{\N}$. This completes the proof.
\end{proof}

In light of \Cref{thm:realizability}, to tackle the marginal independence problem for a class of states $\Omega$ satisfying the assumptions \eqref{eq:uncorr} and \eqref{eq:prod_clos}, one might hope it suffices to only focus on the meet-irreducible elements of $\lat{KC}^\N$. Indeed, if all meet-irreducible elements are realizable by states in $\Omega$, then by \eqref{eq:meet_tp} all other elements are also realizable. 
However, the following result shows that for $\N\geq 4$ this can never be the case.\footnote{\, For $\N=2,3$ the lattice $\lat{KC}^\N$ is coatomistic, and we have given explicit realizations of the coatoms by quantum states (see \Cref{fig:ALatticeMIPosetN2} and \Cref{fig:N3KPMIlattice}).}

\begin{thm}
\label{thm:KC_SSA}
    For any $\N\geq 4$, there exist \emph{KC}-compatible \emph{PMI}s that are not \emph{SSA}-compatible, and therefore cannot be realized by any quantum state.
\end{thm}
\begin{proof}
    It was shown in \cite{Hernandez-Cuenca:2019jpv} that for $\N=4$, the set of all PMIs compatible with SSA can be realized by quantum states.\footnote{\, In fact, it is sufficient to consider stabilizer states.} By \Cref{thm:realizability}, it follows that this set is in fact a lattice with meet given by \eqref{eq:PMImeet}, and the meet-irreducible elements are the PMIs realized by the extreme rays of the cone carved by all instances of SA and SSA.\footnote{\, Notice that this does not imply that all the extreme rays of this cone can be realized by quantum states. In fact, this is not the case because of the constrained inequalities mentioned in the introduction.} Via explicit computation, it is straightforward to verify that the number of these meet-irreducible elements is strictly smaller than the number of meet-irreducible elements of $\lat{KC}^4$. This shows that the set of $\N=4$ PMIs that are SSA-compatible is properly contained in the set of $\N=4$ PMIs that are KC-compatible, proving the claim for $\N=4$. The generalization to any $\N>4$ is immediate using the canonical embedding described in \Cref{thm:embedding}.  
\end{proof}

\Cref{thm:KC_SSA} implies that for any class of states $\Omega$, and any $\N\geq 4$, the lattice $\lat{$\Omega$}^{\N}$ identified by \Cref{thm:realizability} is not $\lat{KC}^{\N}$. However, since the meet for both lattices is simply given by intersection, it suffices to know all the meet-irreducible elements of these lattices to construct all their elements. A useful strategy to find $\lat{$\Omega$}^{\N}$ then might be to look for its meet-irreducible elements starting from those of $\lat{KC}^{\N}$.

The proof of \Cref{thm:KC_SSA} also shows that for $\N\geq 4$ there exist PMIs that are KC-compatible but not compatible with SSA. Geometrically, this means that there exist faces of the SAC$_{\N}$ whose corresponding PMIs are KC-compatible, while all the vectors in their \textit{interior} violate SSA. Thus, KC should be interpreted as an approximation of SSA, and correspondingly the KC lattice as an approximation of the lattice of SSA-compatible PMIs given by the following result. 

\begin{thm}
\label{thm:SSA_lattice}
    For any given $\N$, the set $\lat{\emph{SSA}}^{\N}\subseteq\lat{\emph{KC}}^{\N}$ of \emph{PMI}s which are compatible with \emph{SSA} is a lattice, with the meet given by \eqref{eq:PMImeet}.
\end{thm}
\begin{proof}
    The set inclusion $\lat{SSA}^{\N}\subseteq\lat{KC}^{\N}$ trivially follows from the fact that KC-compatibility is a necessary condition for SSA-compatibility for PMIs. To show that $\lat{SSA}^{\N}$ is a lattice and that the meet is given by \eqref{eq:PMImeet}, we can proceed as in the proof of \Cref{thm:KC-PMIlattice}. Clearly, the top element is the origin of the SAC$_\N$, which satisfies SSA trivially. To show that the meet is given by \eqref{eq:PMImeet}, notice that for any $\pmi,\pmi' \in \lat{SSA}^{\N}$, both $\pmi$ and $\pmi'$ must contain at least an entropy vector $\vec{\ent},\vec{\ent}'$ satisfying SSA and belonging to the corresponding faces $\face=\mu^{-1}(\pmi)$ and $\face'=\mu^{-1}(\pmi')$ of the SAC$_{\N}$. We can then consider the sum $\vec{\ent}''=\vec{\ent}+\vec{\ent}'$ and proceed as in the proof of \Cref{thm:realizability} (notice that whether $\vec{\ent},\vec{\ent}'$ are realizable by quantum states is immaterial). Finally, the resulting PMI is guaranteed to be SSA-compatible since by convexity $\vec{\ent}''$ satisfies all instances of SSA.  
\end{proof}

In light of this result, a natural first step in the derivation of $\lat{$\Omega$}^{\N}$ is then the construction of $\lat{SSA}^{\N}$. Again, since the meet of $\lat{SSA}^{\N}$ is, like $\lat{KC}^{\N}$, also given by \eqref{eq:PMImeet}, one could hope to efficiently derive the meet-irreducible elements of $\lat{SSA}^{\N}$ starting from those of $\lat{KC}^{\N}$. We leave this problem, as well as the structural analysis of $\lat{SSA}^{\N}$, and in particular the question of how well it approximates the KC lattice for future work, and conclude this section with a few comments about the particularly interesting case of 1-dimensional PMIs. 

It was already found in \cite{Hernandez-Cuenca:2019jpv} that for all $\N\leq 5$, all extreme rays of the SAC$_{\N}$ which satisfy SSA can be realized by quantum states.\footnote{\, In fact they can all be realized by stabilizer states, and they all are extreme rays of the holographic entropy cone \cite{Bao:2015bfa,Cuenca:2019uzx}.} This result was obtained by the explicit computation of the extreme rays of the cone carved out by all instances of SA and SSA for $\N=5$. To explore how good KC is as an approximation of SSA for 1-dimensional PMIs, we need to find all KC-compatible extreme rays of SAC$_5$. Although it is tempting to achieve this via a brute force search by simply determining all extreme rays of SAC$_5$ and checking which satisfy KC, the combinatorics render this search very difficult due to the doubly exponential timescales involved. This is where the power of \Cref{cor:er_corollary} becomes apparent, since by utilizing it, we are able to easily compute the KC-compatible extreme rays for $\N=5$, leading us to make the following observation.  

\begin{obs}
    For all $\N\leq 5$, all \emph{KC}-compatible extreme rays of the \emph{SAC}$_{\N}$ satisfy \emph{SSA}.
\end{obs}

The same fact also holds for all currently known KC-compatible extreme rays of the highly non-trivial case of the SAC$_6$ \cite{He:2022wip}. Motivated by this, we pose the following conjecture. 

\begin{conjecture}
\label{conj:1d_KC_SSA}
    For any $\N$, all \emph{KC}-compatible extreme rays of the \emph{SAC}$_{\N}$ satisfy \emph{SSA}. 
\end{conjecture}

While compatible with all known data, we stress that if true, this statement would be quite a remarkable result, given that the analogous statement does not hold for all PMIs (and more generally holds only on a set of measure zero of all entropy vectors).
It would be particularly interesting to understand what property of extreme rays ensures this restricted equivalence between KC and SSA.

\section{Discussion}
\label{sec:discussion}

The main goal of this work was to analyze the subset of PMIs (or equivalently faces of the SAC) which are compatible with the basic quantum mechanical fact that the mutual information between two systems vanishes only if their joint state factorizes. By elementary properties of the tensor product, this implies that the mutual information between their respective subsystems must also vanish, imposing a restriction on the set of PMIs that can be realized in quantum mechanics.
We have shown that for an arbitrary number of parties, the set of PMIs compatible with this requirement, which we dubbed Klein's condition and conveniently rephrased using a partial order on the set of MI-instances, forms a lattice. We then explored various structural properties of this lattice, which is the content of Theorems \ref{thm:main_theorem}, \ref{thm:JIE}, \ref{thm:new_coatoms}, and \ref{thm:embedding}, and proved a general result that simplifies the explicit computation of 1-dimensional KC-compatible PMIs (\Cref{cor:er_corollary}) for an arbitrary number of parties. Lastly, we commented on the realizability of the elements of $\lat{KC}^\N$ by quantum states (cf. \Cref{thm:realizability}), and on the relation between KC and SSA (cf. \Cref{thm:KC_SSA} and \Cref{thm:SSA_lattice}). Using \Cref{cor:er_corollary}, we were able to find the full set of KC-compatible extreme rays of the SAC$_5$ and verify that it is precisely the set of extreme rays which satisfy SSA. Based on this observation, and a partial result for the case involving six parties \cite{He:2022wip}, we conjectured the equivalence of these sets for an arbitrary number of parties (cf. \Cref{conj:1d_KC_SSA}). 

There are still many questions that remain open, both about the structure of the lattices $\lat{KC}^\N$ and $\lat{SSA}^\N$ for arbitrary $\N$, and about their applicability to studying the set of PMIs realizable either by arbitrary quantum states, or by states in restricted classes of particular interest. Some technical questions pertaining to the structure of $\lat{KC}^\N$ were already mentioned in \S\ref{ssec:latticN}. In this section, we would like to comment on more general open questions and future directions of investigations.

\paragraph{Meet-irreducible elements and closure systems:} If the join (meet) operation is known explicitly, any lattice can be constructed from its set of join-irreducible (meet-irreducible) elements. In the case of $\lat{KC}^{\N}$, the join does not appear to have a simple form, as it depends on the intricate relations among the various instances of SA, which vary with $\N$. On the other hand, since PMIs are described by their $\zdss$, the meet in $\lat{KC}^\N$ is simply given by the
intersection of $\zdss$ (cf.\ \eqref{eq:PMImeet}), so $\lat{KC}^{\N}$ can be easily  constructed from the meet-irreducible elements. To make further progress in the analysis of $\lat{KC}^\N$ for larger $\N$, one can then focus on these elements only.
 
Since (as a set) $\lat{KC}^{\N}$ is a subset of both $\lat{PMI}^{\N}$ and $\lat{MID}^{\N}$ (cf.\ \eqref{eq:lattice_intersection}), and the meet operation is the same for all these lattices (cf.\ \Cref{thm:main_theorem}), one may attempt to derive the meet-irreducible elements of $\lat{KC}^{\N}$ starting from those of $\lat{PMI}^{\N}$ or $\lat{MID}^{\N}$, and if necessary, take intersections of their corresponding sets of vanishing MI instances. However, in the case of $\lat{PMI}^{\N}$, the computation of the meet-irreducible elements is challenging, since it amounts to finding the extreme rays of the full SAC starting from its facets, and in general there is no efficient algorithm for such a computation. On the other hand, the meet-irreducible elements of $\lat{MID}^{\N}$ are straightforward to derive. Since the lattice is distributive (see the paragraph above \Cref{def:mid_lattice}), they correspond, from the point of view of the MI-poset, to the complements of the principal up-sets. An example of these elements are the PMIs of single Bell pairs, which are the complements of the principal up-sets of the minimal elements of the MI-poset (cf.\ proof of \Cref{thm:new_coatoms}). Notice that these Bell pairs are also meet-irreducible in $\lat{KC}^{\N}$ and $\lat{PMI}^{\N}$ since they are coatoms, but they are not coatoms in $\lat{MID}^{\N}$. In contrast, the complements of principal up-sets of non-minimal elements of the MI-poset do not correspond to PMIs, since they all include $\D$ linearly independent MI instances,\footnote{\, This follows from the fact that these sets strictly contain the complements of the principal up-sets of minimal elements, which are coatoms of $\lat{KC}$.} and the intersection of the corresponding hyperplanes in entropy space is the origin.
 
From these simple considerations one can then expect that the meet-irreducible elements of $\lat{KC}^{\N}$ will typically be given by the meet of several meet-irreducible elements of $\lat{MID}^{\N}$, since as $\N$ grows, the width of the MI poset grows much faster than its height.\footnote{\, The width of a poset is the cardinality of its maximal antichain.} Therefore, in order to efficiently derive the meet-irreducible elements of $\lat{KC}^{\N}$, one needs to better understand the dependence relations among the various instances of SA. A natural framework for this analysis lies in the theory of matroids\footnote{\, Matroids are combinatorial structures that abstractify the notion of linear dependence and can be axiomatized in several different ways.} \cite{oxley} (and its oriented version \cite{MR1226888}), which provides tools to phrase the implications of the dependence among the instances of SA purely in combinatorial terms \cite{He:2022wip}. 
 
 A matroid can also be seen as a particular instance of a more general structure called a \textit{closure system}, i.e., a pair composed of a set $\mathcal{X}$ and a \textit{closure operator} $\text{cl}_{\mathcal{X}}$.\footnote{\, A closure operator $\text{cl}_{\mathcal{X}}$ is map on the power-set of $\mathcal{X}$ such that for all $\mathcal{Y},\mathcal{Z}\subseteq\mathcal{X}$: 
 \begin{align*}
 \begin{split}
    & \mathcal{Y}\subseteq \text{cl}_{\mathcal{X}}(\mathcal{Y}) \\
    & \mathcal{Y}\subseteq\mathcal{Z} \implies \text{cl}_{\mathcal{X}}(\mathcal{Y})\subseteq\text{cl}_{\mathcal{X}}(\mathcal{Z})\\
    & \text{cl}_{\mathcal{X}}(\text{cl}_{\mathcal{X}}(\mathcal{Y}))=\text{cl}_{\mathcal{X}}(\mathcal{Y})\,.
 \end{split}
 \end{align*}
}
For any closure system, the set of closed subsets (with partial order given by inclusion) is a lattice. Moreover, any lattice such that its elements can be represented by sets, and its meet by intersection, is the lattice of closed sets of a closure system \cite{davey1990introduction}. Therefore $\lat{KC}$, $\lat{PMI}$, and $\lat{MID}$ can all be thought of as lattices of closed sets for different closure operators on the set $\mgs$ of MI instances, and it would be interesting to explore in detail the properties of the closure operator of $\lat{KC}$, and its relation to the closure operators of $\lat{PMI}$ and $\lat{MID}$. The closure operator of $\lat{MID}$ is simple, as it just associates to a collection of MI instances the smallest down-set in the MI-poset that contains the collection. On the other hand, the closure operator of $\lat{PMI}$ is more complicated.\footnote{\, It is associated to the closure operator that defines the oriented matroid \cite{BUCHI1988293} of SA instances.} Since $\lat{KC}$ is the intersection of $\lat{MID}$ and $\lat{PMI}$, and given the simplicity of the closure operator for $\lat{MID}$, it would be interesting to understand how much information is necessary about the closure operator of $\lat{PMI}$ to extract the closure operator of $\lat{KC}$. The general theory of lattices of closure systems, closure operators, and implicational bases \cite{gratzer2016,CASPARD2003241,BERTET201893} provides the appropriate framework for this analysis.

Finally, we have shown in \Cref{thm:JIE} that the set of join-irreducible elements of $\lat{KC}$ includes all the PMIs whose $\zdss$ correspond to the join-irreducible elements of $\lat{MID}$ that are not principal down-sets of maximal elements in the MI-poset. Furthermore, we have also seen that for $\N\leq 4$ there is no join-irreducible element in $\lat{KC}$ whose $\zds$ corresponds to an element that is not join-irreducible in $\lat{MID}$. It would be interesting to understand whether this property holds for an arbitrary number of parties (cf.\ \Cref{question:question1}), as well as its implications for the closure operator of $\lat{KC}$ \cite{finkbeiner}.

\paragraph{Realizability in quantum mechanics:} 

We mentioned at the end of \S\ref{sec:realizability} that for $\N\geq 4$ there exist PMIs in $\lat{KC}^{\N}$ that are not compatible with SSA, and therefore are not realizable by quantum states. In search of the solution to the QMIP, the next natural step is then the extraction of the set of PMIs which are compatible with SSA. As shown in \Cref{thm:SSA_lattice}, this set is also a lattice, with meet corresponding to the intersection of $\zdss$. 

A natural strategy to obtain this lattice would be to first construct a cone specified not only by the instances of SA, but also of SSA. For each face of this cone one could then determine the PMI, and the set of PMIs obtained in this fashion would give the desired lattice. However, as we mentioned in \S\ref{subsec:ssa_k_condition}, this procedure quickly becomes extremely inefficient as $\N$ grows, due to the fact that many faces correspond to the same PMI (cf.\ \Cref{ft:ssa}). 

We suggest that a more promising approach instead is to utilize $\lat{KC}^\N$ as the starting point, viewing KC as an approximation of SSA for PMIs. Starting from the meet-irreducible elements of $\lat{KC}^{\N}$, one can first verify which elements are consistent with SSA.\footnote{\, For example, this can be done efficiently with a linear program, by first reducing all instances of SSA to the subspace of interest, and then checking for the feasibility of the resulting system of inequalities. Alternatively, it would also be interesting to obtain a purely combinatorial characterization of SSA-compatible PMIs using matroid theory.} By taking the meet of such elements, one can then determine new elements of $\lat{KC}^{\N}$ and proceed to repeat the process until one has the generating set of the desired lattice. The efficiency of this procedure will depend on how effective KC is to approximating SSA, as well as on the strategy utilized to determine which pairs of elements should be chosen to compute the meet. For instance, if one could prove \Cref{conj:1d_KC_SSA}, it suffices then to focus only on higher dimensional meet-irreducible elements, and it would be interesting to explore whether such elements also, up to some maximal dimension independent of $\N$, all satisfy SSA automatically. More generally, it would be useful to explore these aspects of the construction in more detail, as well as the structural properties of $\lat{SSA}^{\N}$.

Once the lattice of SSA-compatible PMIs has been found for some $\N$, one would like to explore if all such PMIs are realizable in quantum mechanics. As we mentioned in the introduction, for $\N\geq 4$ there are no known unconstrained inequalities for the von Neumann entropy beyond SA and SSA, but infinitely many constrained inequalities have been found \cite{Linden:2004ebt, Cadney_2012}. It would be an intriguing research direction to ascertain if any of these inequalities provide new restrictions on the set of PMIs which are independent from SSA. However, we already know that this does not happen for $\N=4$, since in this case all SSA-compatible PMIs are realizable by quantum states \cite{Hernandez-Cuenca:2019jpv}.

Given an arbitrary PMI that satisfies all known entropy inequalities, it is typically hard to determine whether it can be realized by a quantum state. This is because there is no standard procedure to construct a quantum state with a desired pattern of correlation. Nevertheless, at least for the restricted class of stabilizer states, there are more tools available. One possibility is to utilize the hypergraph models of \cite{Bao:2020zgx}, which provide direct means of constructing entropy vectors that can be realized by stabilizer states \cite{Walter:2020zvt}.\footnote{\, The hypergraph models were introduced in \cite{Bao:2020zgx} as a generalization of the graph models of holographic entanglement \cite{Bao:2015bfa}, and one may wonder if the latter may also be useful for explicit realizations. However, it was shown in \cite{Hernandez-Cuenca:2019jpv} that already for $\N=4$, there exist PMIs that are realizable by stabilizer states, and in particular by hypergraph models, but not by graph models.} Even though it was shown in \cite{Bao:2020mqq} that there are stabilizer states whose entropy vectors cannot be realized by hypergraph models, it is in principle also possible that hypergraph models are nevertheless sufficient to realize all PMIs that can be realized by stabilizer states.

Finally, it is certainly possible that there exist PMIs which can be realized by quantum states but not by stabilizer states, and presumably these would be harder to find. However, it was speculated in \cite{Cadney_2012} that all balanced inequalities satisfied by classical probability distributions could also hold for the von Neumann entropy. As these inequalities are likewise known to be satisfied by stabilizer states \cite{Gross_2013}, it is also interesting to contemplate the possibility that stabilizer states are in fact sufficient to realize all quantum mechanical PMIs. If this were the case, the set of PMIs realizable by quantum states would be a proper subset of the set of PMIs compatible with SSA, since already for $\N=5$ there exist PMIs which are SSA-compatible but not realizable by stabilizer states \cite{Hernandez-Cuenca:2019jpv}.

\paragraph{1-dimensional PMIs and applications to quantum gravity (holography):}

In the previous paragraph we commented on the relation between SSA and KC, both of which are necessary conditions for the realizability of PMIs. While these conditions are in general not equivalent (cf. \Cref{thm:KC_SSA}), an open question is whether they are equivalent for the specific case of 1-dimensional PMIs. As we mentioned at the end of \S\ref{sec:realizability}, we have verified this equivalence up to $\N=5$ (and partially for $\N=6$ \cite{He:2022wip}), and we conjecture that it holds for any $\N$ (cf. \Cref{conj:1d_KC_SSA}).

The interest in 1-dimensional PMIs stems from its potential application to the analysis of entropic constraints in the gauge/gravity duality. It was recently argued in \cite{Hernandez-Cuenca:2022pst} that all holographic entropy inequalities might in principle be reconstructed from the knowledge of which extreme rays of the SAC can be realized by the graph models of \cite{Bao:2015bfa}.\footnote{\, Crucially, for the reconstruction of the $\N$-party inequalities, it is in general not sufficient to know the extreme rays of the SAC$_\N$ realized by graph models. However, \cite{Hernandez-Cuenca:2022pst} argued that more generally, for any $\N$, there exists a finite $\N'\geq \N$ such that this type of information would be sufficient for the reconstruction (see \cite{Hernandez-Cuenca:2022pst} for more details).} These observations suggest a natural question which aims at understanding the very origin of holographic entropic constraints, namely whether the extreme rays of the SAC that can be realized by graph models are precisely the ones that can be realized in quantum mechanics, or if instead there exist extreme rays that can be realized by quantum states but not by graph models.

To further comment on this question, it is useful to consider the following subsets of extreme rays of SAC$_\N$ from \cite{Hernandez-Cuenca:2022pst}:
\begin{align}
\begin{split}
    \er &\coloneqq \{\pmi\in \lat{PMI}^{\N}:\, \text{dim}(\pmi)=1 \}\\
    \erssa  &\coloneqq \{\pmi\in\er:\, \pmi\; \text{is SSA-compatible}\}\\
    \erq &\coloneqq \{\pmi\in\er:\, \pmi\; \text{is realizable by a quantum state}\}\\
    \erh &\coloneqq \{\pmi\in\er:\, \pmi\; \text{is realizable by a graph model}\}\,,
\end{split}
\end{align}
and similarly define the new set
\begin{equation}
    \erkc \!\coloneqq \{\pmi\in \lat{KC}^{\N}:\, \text{dim}(\pmi)=1\}\,.
\end{equation}
From these definitions we clearly have
\begin{equation} \label{eq:erinclusions}
     \erh \subseteq \erq \subseteq \erssa \subseteq  \erkc \subseteq \er \,,
\end{equation}
where the last inclusion is strict for any $\N\geq 3$. While it seems reasonable to expect that $\erh$ could be derived explicitly using graph model constructions, determining $\erq$ in principle could be much harder. A tantalizing possibility suggested in \cite{Hernandez-Cuenca:2022pst} was that $\erh\!=\erssa$, and the initial hope was to show it by demonstrating that in fact $\erh\!=\erkc$. By \eqref{eq:erinclusions}, this would imply $\erh\!=\erq$, which would answer the question mentioned above, as well as $\erq\!=\erssa$. This would provide a partial solution to the QMIP and interesting structural information about the quantum entropy cone. However, subsequently to the first version of this work, partial analysis involving $\N=6$ parties already showed that $\erh\subset\erq$ strictly \cite{he2023gap}. On the other hand, the possibility that $\erq\!=\erssa\!=\erkc$ remains open, and is currently being investigated in \cite{He:2022wip}.

\paragraph{Unbounded number of parties:} Finally, we observe that throughout this paper, as well as in all other previous works on entropy cones, the number of parties $\N$ was always assumed to be finite. While this is natural from an operational perspective (for example, qubits in a quantum computer), in quantum field theory (QFT) there are no real restrictions to the possible refinements of a given collection of subsystems. In other words, given a state in a QFT, and $\N$ regions of the manifolds on which the theory is defined, one can always imagine partitioning these regions into $\N' > \N$ smaller ones.\footnote{\, As in the case of the multi-boundary systems used in \cite{Bao:2015bfa}, to define the holographic entropy cone we could consider a state in the tensor product of $\N$ QFTs, but here we are focusing on a choice of regions that probes a single QFT.} Fixing the value of $\N$ to be a specific finite number is just a matter of convenience, but this number has no real meaning and might introduce various combinatorial artifacts. For example, the structure of the MI-poset depends on the partitions of the integers between 2 and $\N$, which affects the dependence relations among the instances of SA and in turn might render the lattice more complicated. An interesting research direction would be to investigate a generalization of the analysis presented here to a set-up where $\N$ is infinite, in the hope that some of these artifacts will disappear. The theory of infinite lattice might then be a natural framework for such an investigation.

\acknowledgments
 
It is a pleasure to thank S. Hern\'andez-Cuenca and M. Parisi for useful discussions. MR would also like to thank QMAP, the University of California Davis, and California Institute of Technology for their hospitality during various stages of this project. TH has been supported by funds from the University of California, the U.S. Department of Energy grant DE-SC0020360 under the HEP-QIS QuantISED program, the Heising-Simons Foundation ``Observational Signatures of Quantum Gravity'' collaboration grant 2021-2817, the U.S. Department of Energy grant DE-SC0011632, and the Walter Burke Institute for Theoretical Physics.
VH has been supported in part by the U.S. Department of Energy grant DE-SC0009999 and funds from the University of California. 
MR has been supported by the University of Amsterdam, via the ERC Consolidator Grant QUANTIVIOL, and by the Stichting Nederlandse Wetenschappelijk Onderzoek Instituten (NWO-I).

\bibliography{KTbib}{}
\bibliographystyle{utphys}

\end{document}